\documentclass[fontsize=11pt, a4paper, oneside, DIV=12, toc=listofnumbered,bibnumbered, numbers = noendperiod]{scrartcl}
\addtokomafont{disposition}{\rmfamily}
\DeclareOldFontCommand{\bf}{\normalfont\bfseries}{\mathbf}

\usepackage[utf8]{inputenc}
\usepackage[T1]{fontenc}
\usepackage{lmodern}

\usepackage{newunicodechar}
\newunicodechar{ρ}{\rho}
\newunicodechar{σ}{\sigma}
\newunicodechar{λ}{\lambda}
\newunicodechar{Λ}{\Lambda}
\newunicodechar{Γ}{\Gamma}
\newunicodechar{μ}{\mu}
\newunicodechar{ν}{\nu}
\newunicodechar{ψ}{\psi}
\newunicodechar{ϕ}{\phi}
\newunicodechar{φ}{\varphi}
\newunicodechar{π}{\pi}
\newunicodechar{α}{\alpha}
\newunicodechar{ϵ}{\epsilon}
\newunicodechar{ε}{\varepsilon}
\newunicodechar{δ}{\delta}
\newunicodechar{ω}{\omega}
\newunicodechar{β}{\beta}
\newunicodechar{ξ}{\xi}
\newunicodechar{∗}{\textasteriskcentered}
\newunicodechar{η}{\eta}
\newunicodechar{ζ}{\zeta}
\newunicodechar{θ}{\theta}
\newunicodechar{τ}{\tau}
\newunicodechar{γ}{\gamma}
\newunicodechar{Φ}{\Phi}
\newunicodechar{κ}{\kappa}

\usepackage{comment}
\usepackage{microtype}

\usepackage[english]{babel}
\usepackage{csquotes}

\usepackage{xparse}
\usepackage{xspace}
\usepackage{makecell}
\usepackage{titling}
\usepackage{appendix}

\usepackage{versions}
\includeversion{DRAFT}
\excludeversion{TRASH}
\excludeversion{EXCLUDED}

\usepackage{amssymb}
\usepackage{amsmath}
\usepackage{mathrsfs}
\usepackage{mathtools}
\usepackage{optidef}
\usepackage[shortlabels]{enumitem}

\usepackage{cool}
\usepackage{xcolor}
\usepackage{amsthm}
\usepackage{hyperref}
\hypersetup{
	colorlinks,
	linkcolor={red!50!black},
	citecolor={red!50!black},
	urlcolor={blue!80!black},
	linktoc=all
}
\usepackage{thmtools}
\usepackage[nameinlink]{cleveref}
\addto\extrasenglish{%
}

\makeatletter
\@ifundefined{newcounteralias}{}{%
    \renewcommand*\thmt@autorefsetup{%
        \expandafter\def\csname\thmt@envname autorefname\expandafter\endcsname
            \expandafter{\thmt@thmname}%
    }%
}
\makeatother

\newtheorem{theorem}{Theorem}[section]
\newtheorem{lemma}[theorem]{Lemma}
\newtheorem{corollary}[theorem]{Corollary}

\newtheorem{proposition}[theorem]{Proposition}
\theoremstyle{definition}
\newtheorem{remark}[theorem]{Remark}
\newtheorem{example}[theorem]{Example}
\newtheorem{definition}[theorem]{Definition}

\newtheorem*{definition*}{Definition}

\makeatletter
\@for\thmname:=theorem,lemma,corollary,conjecture,proposition,remark,example,definition,problem\do{%
    \expandafter\def\csname theH\thmname\endcsname{\theHsection.\arabic{theorem}}%
}
\makeatother

\usepackage{graphicx}     
\setkeys{Gin}{width=\textwidth} %
\usepackage[multidot,space]{grffile} %% Enables \includegraphics to parse files with multiple dots and spaces
\usepackage[compatibility=false]{caption}
\usepackage{subcaption}
\graphicspath{{./}{img/}}

\makeatletter
\let\x@caption\caption% original (caption package) \caption
\newcommand{\x@@caption}[2][\empty]{%
	\ifx\empty#1\relax\x@caption{#2}%
	\else\x@caption[#1]{\textbf{#1.} #2}%
	\fi}% modified caption
\let\caption\x@@caption% new global \caption for other float types
\renewcommand{\captionabove}[2][\empty]{\captionsetup{position=above}%
	\x@@caption[#1]{#2}%
}% new \captionabove
\renewcommand{\captionbelow}[2][\empty]{\captionsetup{position=below}%
	\x@@caption[#1]{#2}%
}% new \captionbelow
\makeatother

\usepackage[style=alphabetic, maxbibnames=99]{biblatex} %minalphanames = 1, maxalphanames = 3

\AtEveryBibitem{
	\clearfield{note}
	\clearfield{issn}
	\clearfield{urlyear}
	\clearfield{urlmonth}
	\iffieldundef{journaltitle}{}{\clearfield{version}}
	\iffieldundef{journaltitle}{\iffieldundef{booktitle}{\iffieldundef{eventtitle}{\iffieldundef{eprint}{}{\clearfield{doi}}}{}}{}}{}
}

\DeclareSourcemap{
    \maps[datatype=bibtex]{
        \map[overwrite]{
            \step[fieldsource=eprint, match=\regexp{\A([0-9.v]+)\s*\[.*?\].*\Z}, replace={$1}]
        }
    }
}

\renewbibmacro*{url}{%
	\iffieldundef{doi}{\iffieldundef{eprint}{%
			\printfield{url}%
		}{}}{}
}

\usepackage{suffix}
\usepackage{xparse}
\usepackage{braket}

\usepackage{physics}

\renewcommand{\tr}{\Tr}

\usepackage{mathtools}
\usepackage{amssymb}
\usepackage{tensor}

\usepackage{ stmaryrd }

\newcommand{\br}[1]{\left(#1\right)}

\NewDocumentCommand{\stared}{m m}{\IfBooleanTF{#2}{#1*}{#1}}
\NewDocumentCommand{\es}{m o}{
	\IfNoValueTF{#2}{
		\mathbb{\uppercase{#1}}^{\lowercase{#1\/}}
	}{
		\mathbb{\uppercase{#1}}^{#2\/}
	}
}

\newcommand{\complex}{\mathbb{C}}
\newcommand{\reals}{\mathbb{R}}

\newcommand{\HS}{\mathcal{H}}

\NewDocumentCommand{\BO}{o}{\mathcal{B}\br{\IfValueTF{#1}{#1}{\HS}}}

\NewDocumentCommand{\DM}{o}{\mathcal{D}\br{\IfValueTF{#1}{#1}{\HS}}}

\NewDocumentCommand{\pos}{o}{\mathscr{P}\!\br{\IfValueTF{#1}{#1}{\HS}}}

\DeclarePairedDelimiterX{\infdivx}[2]{(}{)}{
	#1\;\delimsize|\delimsize|\;#2
}

\DeclareDocumentCommand{\qrd}{o o m m}{\ensuremath{\IfValueTF{#1}{#1}{D}\IfValueTF{#2}{_{#2}}{}\/\infdivx*{#3}{#4}}}
\NewDocumentCommand{\ent}{d()g}{\ensuremath{H\br{\IfValueTF{#1}{#1}{#2}}}}
\NewDocumentCommand{\qent}{d()g}{\ensuremath{H\br{\IfValueTF{#1}{#1}{#2}}}}

\newcommand*\ric[1]{\vphantom{#1}\smash{#1_{}\kern-\scriptspace}}

\renewcommand{\abs}[1]{\left\lvert#1\right\rvert}

\renewcommand{\norm}[1]{\left\lVert #1 \right\rVert}

\newcommand{\renyi}{R\'enyi\xspace}

\newcommand{\Sn}{\mathfrak{S}_n}
\NewDocumentCommand{\End}{o}{\mathrm{End}^{\Sn}(\IfValueTF{#1}{#1}{R^n B^n})}

\newcommand{\n}{^{\otimes n}}

\newcommand{\id}{\mathrm{id}}

\renewcommand{\D}{\mathbf{D}}

\usepackage{bm}

\newcommand{\M}{\ensuremath{\mathcal{M}}}
\newcommand{\Ms}{\ensuremath{\mathcal{M}_{*}}}

\newcommand{\N}{\ensuremath{\mathcal{N}}}
\newcommand{\Ns}{\ensuremath{\mathcal{N}_{*}}}

\let\textcaronaccent\v
\DeclareRobustCommand{\v}[1]{\ifmmode\bm{#1}\else\textcaronaccent{#1}\fi}
\renewcommand{\ip}[2]{\langle {#1}, {#2}\rangle}

\renewcommand{\S}{\mathcal{S}}

\DeclareMathOperator{\Ind}{Ind}

\bibliography{references.bib}
\title{Conditioning on Subalgebras -- Entropy Duality and Generalized Quantum Stein's Lemma in von Neumann Algebras}
\author{Bjarne Bergh\thanks{Department of Applied Mathematics and Theoretical Physics, University of Cambridge, United Kingdom} \and  Li Gao\thanks{School of Mathematics and Statistics, Wuhan University, China} \and Mizanur Rahaman\thanks{Department of Mathematical Sciences, Chalmers University of Technology, Gothenburg, Sweden}}

\begin{document}
    \maketitle
    \begin{abstract}
       We establish duality relations for conditional entropies relative to subalgebras of arbitrary von Neumann algebras, including the Umegaki relative entropy as well as the sandwiched- and Petz-\renyi divergences. Even in finite dimensions, these results extend the familiar dualities for conditional entropies from tensor-product subsystems to general subalgebras, while also lifting the framework to the fully infinite-dimensional setting. As an application, we establish an infinite-dimensional extension of quantum Stein’s lemma when the alternative hypothesis is the state space of a finite-index subalgebra, which gives an operational interpretation of the relative entropy to a subalgebra as the optimal asymptotic error exponent in composite quantum hypothesis testing—the first beyond finite dimensions.

    \end{abstract}
\tableofcontents
\section{Introduction}

\subsection{Entropic duality relations}\label{sec:intro-duality-origin}
In finite-dimensional quantum systems, it is well known that for any tripartite pure state $\rho_{ABC}$, 
\begin{equation}\label{duality-f.d vN cond.}
	H(A|B)_\rho + H(A|C)_\rho = 0,
\end{equation}
where $H(A|B)_\rho$ is the conditional von Neumann entropy of  the system $A$ given system $B$ in the state $\rho_{ABC}$, and similarly $H(A|C)_\rho$ of the system $A$ given system $C$. This is known as the \textit{duality relation for conditional entropy}. Similar duality relations have been discovered for other entropic quantities, for example, sandwiched \renyi conditional entropies \cite{muller-lennert_quantum_2013, beigi_sandwiched_2013, tomamichel_relating_2014}, where it holds that for any tripartite pure  $\rho_{ABC}$
\begin{equation}\label{duality-f.d cond.}
	\widetilde{H}_{\alpha}(A|B)_{\rho}+ \widetilde{H}_{\beta}(A|C)_{\rho}=0, \text{ when } \ \frac{1}{\alpha}+\frac{1}{\beta}=2; \ \alpha, \beta\in [1/2, \infty].
\end{equation}
Here, for any bipartite state $\rho_{AB}$, the sandwiched \renyi conditional entropies $\widetilde{H}_{\alpha}(A|B)_{\rho_{AB}}$ is defined as
\[
	\widetilde{H}_{\alpha}(A|B)_{\rho_{AB}}:=-\inf_{\sigma_B} \widetilde{D}_{\alpha}(\rho_{AB}\|\IdentityMatrix_A\otimes \sigma_B),
\]
and $\widetilde{D}_{\alpha}(\rho\|\sigma)$ is the sandwiched R\'enyi divergence of states $\rho$ and $\sigma$, defined as
\[
	\widetilde{D}_{\alpha}(\rho\|\sigma):= \frac{1}{\alpha-1}\log \Tr\big[(\sigma^{\frac{1-\alpha}{2\alpha}}\rho\, \sigma^{\frac{1-\alpha}{2\alpha}})^\alpha\big],\qquad \alpha \in [1/2, 1)\cup (1, \infty).
\]
Analogous duality relations also hold for the Petz-R\'enyi type divergences \cite{tomamichel_fully_2009}, and both recover the standard relative entropy duality \eqref{duality-f.d vN cond.} in the limit $\alpha \to 1$. Duality relations of this type have proven crucial in many parts of quantum information -- for example in cryptography, where the information an adversarial party, say $C$, has about a quantum system $A$ can be estimated using local state tomography by the two honest parties $A$ and $B$ \cite{tomamichel_tight_2012}.

In the present paper, we extend these duality relations to $(i)$ conditioning on a subalgebra rather than a tensor factor, and $(ii)$ states on
general von Neumann algebras rather than finite-dimensional ones. As an application,
we obtain a natural generalization of the celebrated Generalized Quantum Stein’s Lemma
whenever the alternative hypothesis is the state space of a finite-index subalgebra, in an otherwise
completely general von Neumann algebra.

\subsection{Conditioning to subalgebras}\label{sec:intro-subalgebra}

The entropic duality relations we prove in this paper are motivated by re-stating \eqref{duality-f.d vN cond.}--\eqref{duality-f.d cond.} in a language that no longer refers to a tensor decomposition. For conditional entropies, conditioning is implemented by restricting a state $\rho$ to the subsystem we condition on. In general, such a subsystem is represented as a subalgebra. Operationally, a subsystem is nothing more than the collection of operations available to an agent with only partial access to the full system. Under natural axioms these are exactly the operations whose Kraus operators lie in a von Neumann subalgebra, and an inclusion $\N \subseteq \M$ then describes a pair of agents, one of whom can implement every operation available to the other \cite[Secs.~4.2--4.3]{vanluijk_thesis_2025}. To see how the tensor-product case fits into this picture, let's revisit the relation \eqref{duality-f.d vN cond.}. Here in finite dimensions, we have  $$\N = \IdentityMatrix_A \otimes \BO[\HS_B] \subset \BO[\HS_A \otimes \HS_B]=\M$$
as a subalgebra. $\M$ is further embedded into the algebra of a purifying Hilbert space, $\M \subset \BO := \BO[\HS_A \otimes \HS_B \otimes \HS_C]$. Then the partial trace map  $$E:\M \to \N, \ \  M_{AB} \mapsto \frac{\IdentityMatrix_A}{d_A} \otimes \Tr_A(M_{AB}) $$ is a trace-preserving conditional expectation. The embedding of states on $B$ into states on $AB$ is via the pre-dual map $E_*: \rho_B \mapsto \IdentityMatrix_A/d_A \otimes \rho_B$ (or, $E_*(ρ_B) = ρ_B\circ E$ considering states as functionals on the algebras). 

Additionally, within the total algebra $\BO[\HS_A \otimes \HS_B \otimes \HS_C]$, the subalgebra $\IdentityMatrix_{AB} \otimes \BO[\HS_C]$ is the commutant $\M'$ of $\M = \BO[\HS_{A} \otimes \HS_B] \otimes \IdentityMatrix_C$, and the commutant $\N'$ is isomorphic to $\BO[\HS_A \otimes \HS_C]\otimes \IdentityMatrix_B$. Notice that we also have a natural conditional expectation $E': \N' \to \M'$ which is again given by a partial trace: $$E'(M_{AC}\otimes \IdentityMatrix_B)=\frac{\IdentityMatrix_A}{d_A}\otimes \IdentityMatrix_B\otimes \Tr_A(M_{AC}).$$

Now for any pure state $\rho_{ABC}$ acting on $\HS_A \otimes \HS_B \otimes \HS_C$, we see that
\begin{align}\label{eq:cond-normalization}
	H(A|B)_{\rho}=-\inf_{\sigma_B}D(\rho_{AB}\| \frac{\IdentityMatrix_A}{d_A}\otimes \sigma_B)+\log d_A=-D(\rho|_{\M}\|E_*(\N_*))+ \log d_A.
\end{align}
Here we adopt the notation $\displaystyle D(\rho\|E_*(\N_*)):=\inf_{\sigma\in \S(\N)} D(\rho\|E_*(\sigma))$, where the infimum is over all states of $\N$, and the state $ρ|_{\M}$ is the restriction of $ρ$ (seen as a linear functional) to the observables of $\M$, which corresponds to the reduced density matrix $ρ_{AB}$ of $ρ_{ABC}$. Similarly, we also see that $$H(A|C)_{\rho}=-D(\rho|_{\N'}\|E'_*(\M'_*))+\log d_A,$$ with the analogous notation $\displaystyle D(\rho|_{\N'}\|E'_*(\M'_*)) := \inf_{\sigma \in \S(\M')} D(\rho|_{\N'}\|E'_*(\sigma))$.

Thus, the duality relation in equation \eqref{duality-f.d vN cond.} can be written as
\begin{equation}\label{eq:conditional entropy with dim.}
	D(\rho|_{\M}\|E_*(\N_*)) + D(\rho|_{\N'}\|E'_*(\M'_*)) = 2 \log d_A \, .
\end{equation}
The sandwiched R\'enyi conditional entropies share the normalization \eqref{eq:cond-normalization},
 so the duality relation in \eqref{duality-f.d cond.} reads
\begin{equation}\label{eq:(sandwiched) conditional entropy with dim.}
	\widetilde{D}_\alpha(\rho|_{\M}\|E_*(\N_*)) + \widetilde{D}_\beta(\rho|_{\N'}\|E'_*(\M'_*)) = 2 \log d_A.
\end{equation}
Nothing in \eqref{eq:conditional entropy with dim.}-\eqref{eq:(sandwiched) conditional entropy with dim.} refers to subsystem $B$ or $C$ any more, only to algebraic notions: the pair of algebras $(\M, \N)$, their commutants, and the conditional expectations relating them. The quantity $2\log d_A$ turns out to be the logarithm of the \emph{index} $\Ind E$ of the conditional expectation $E$, as explained below.

This reformulation immediately suggests two independent directions of generalization, which are the main contributions of this paper:
\begin{enumerate}[label=(\roman*), leftmargin=2.2em]
	\item \textbf{Subalgebra Conditioning.} We extend this formalism from subalgebras $\N \subseteq \M$ of tensor-product form to \emph{any} general subalgebra $\N$ of finite index. %Given an inclusion $\N \subseteq \M$ with a faithful normal conditional expectation, the conditional entropy can still be expressed in terms of a relative entropy with respect to the subalgebra: the role of the complementary system is played by the commutant $\N'$ (which now strictly contains $\M'$), and the entropy duality reflects a structural symmetry between a subalgebra and its commutant within a purification. 
    The quantity $2\log d_A$ turns out to be a special case of the logarithm of the \emph{index} $\Ind E$ of a conditional expectation $E: \M \to \N$ (\Autoref{sec:operator_valued_weights}), so that \eqref{eq:conditional entropy with dim.}--\eqref{eq:(sandwiched) conditional entropy with dim.} become the special case of our results in which $\N$ has the tensor-product form discussed above.
	\item \textbf{Beyond finite dimension.} The subalgebra perspective above is inherently infinite-dimensional and identifies entropy duality as a genuinely operator algebraic phenomenon. 
\end{enumerate}
To the best of our knowledge, this provides the first entropic duality relation for the sandwiched- and Petz- R\'enyi relative entropy in von Neumann algebras. As an application, we establish a Generalized Quantum Stein's Lemma (GQSL) for finite-index subalgebras in a general von Neumann algebra. The correct formulation of GQSL for arbitrary convex sets, after a gap was found in the original argument \cite{berta_gap_2023}, has recently been re-established in the finite-dimensional setting by \cite{hayashi_generalized_2025} and, independently, by \cite{lami_solution_2025}. Our result provides a first extension of (the subalgebra version of) the GQSL to infinite dimensions.

\vspace{6pt}

\subsection{Our results}

\paragraph{Finite-dimensional Cases.}\label{sec:intro-finite-dim}
Before stating our results in the full generality of von Neumann algebras, it is worth highlighting what the subalgebra conditioning achieves in finite-dimensional quantum systems. This requires no operator-algebraic vocabulary beyond the notion of \emph{commutant} and \emph{conditional expectation}. The prototypical examples of the latter are the partial trace and the pinching map.

Let $\HS$ be a finite-dimensional Hilbert space, and let $\N\subseteq\M=\BO$ be a unital $*$-subalgebra. As explained more in \autoref{sec:finite-dim-cond-exp} below, every such subalgebra has the direct sum form
\begin{equation}\label{eq:intro-finite-dim-inclusion}
	\HS = \bigoplus_{j \in J} \mathcal{L}_j \otimes \mathcal{V}_j\,, \qquad \N = \bigoplus_{j} \BO[\mathcal{L}_j] \otimes \IdentityMatrix_{\mathcal{V}_j}\,.
\end{equation}
Denote by $\S(\N)$ the states of $\N$; a state $ρ_\N \in \S(\N)$ is given by its blocks, positive operators $ρ_{\mathcal{L}_j}$ on $\mathcal{L}_j$ with $\sum_j \Tr ρ_{\mathcal{L}_j} = 1$. Given a conditional expectation $E: \M \to \N$, the pre-adjoint map $E_*$ acts on density matrices as
\[
E_*(ρ_{\N}) = \bigoplus_{j} ρ_{\mathcal{L}_j} \otimes τ_j \in \S(\M),\ \  ρ_{\N} \in \S(\N)
\]
for some family of density matrices $(τ_j \in \mathcal{B}(\mathcal{V}_j))_{j \in J}$.

Let $\rho\in \S(\M)$ be purified by a vector $\ket{\v{\rho}}\in \HS\otimes\HS_C$ with an auxiliary space $\HS_C$, and let $\M',\N'\subseteq \BO[\HS\otimes\HS_C]$ denote the commutants of $\M$ and $\N$ in $\BO[\HS\otimes\HS_C]$. We find that
\[
 \IdentityMatrix_{\HS} \otimes \BO[\HS_C]=\M' \subseteq \N' = \bigoplus_{j} \IdentityMatrix_{\mathcal{L}_j} \otimes \BO[\mathcal{V}_j \otimes \HS_C]\,.
\]
Given $E$ as above, we can then define the map $E^{-1}:\N'\to{\M'}$ whose pre-adjoint acts on states of $\M'$, i.e.\ density matrices $ρ_C$ on $\HS_C$, as
\begin{equation}
    E^{-1}_*(ρ_{C}) = \bigoplus_{j} τ^{-1}_j \otimes  ρ_{C}, \qquad ρ_{C} \in \S(\M'), \label{eq:E-1}
\end{equation}
a non-normalized positive functional on $\N'$. The motivation behind this construction is explained in \autoref{sec:finite-dim-cond-exp}. With this definition, we get:
\begin{theorem}[Finite-dimensional subalgebra duality]\label{thm:finite_dim_duality}
	In the setting above, 
    \begin{enumerate}
        \item[(a)] for $\alpha,\beta\in[1/2,\infty]$ with $\frac1\alpha+\frac1\beta=2$,
	\begin{equation}
		\widetilde{D}_\alpha\big(\rho|_{\M}\,\|\,E_*(\S(\N))\big)=-\widetilde{D}_\beta\big(\rho|_{\N'}\,\|\,E^{-1}_*(\S(\M'))\big).
	\end{equation}
  \item[(b)] for all $α, β \in [0,2]$ with $β = 2 - α$,
    \begin{equation}
        D_α(ρ|_{\M}\|E_*(ρ|_{\N})) = - D_{β}(ρ|_{\N'}\|E_*^{-1}(ρ|_{\M'}))
    \end{equation}
    \end{enumerate}
In particular, for $α = β = 1$,
    \begin{equation*}
       D(ρ|_{\M}\|E_*(\S(\N)))= D(ρ|_{\M}\|E_*(ρ|_{\N}))  =  - D(ρ|_{\N'}\|E_*^{-1}(ρ|_{\M'}))=- D(ρ|_{\N'}\|E_*^{-1}(\S(\M'))) 
    \end{equation*}
\end{theorem}

\Autoref{thm:finite_dim_duality} is the finite-dimensional case of \Autoref{thm:duality_relation} below, and  it also holds, with the proper definition of $E^{-1}$, when $\M \subsetneq \BO$ itself is a proper subalgebra. Note that the map $E^{-1}$ defined above \eqref{eq:E-1} is a conditional expectation up to the scalar
\begin{equation}\label{eq:intro-finite-dim-index}
    \Ind E\coloneqq E^{-1}(\IdentityMatrix) =\sum_{j \in J} \Tr τ_j^{-1}\,,
\end{equation}
called the \emph{index} of $E$ (\Autoref{sec:operator_valued_weights}). For the trace-preserving conditional expectation, $τ_j = \IdentityMatrix_{\mathcal{V}_j}/\dim \mathcal{V}_j$ and $\Ind E = \sum_j (\dim \mathcal{V}_j)^2$; if $\N$ is a factor, i.e.\ $\M = \BO[\HS_A\otimes\HS_B]$ and $\N=\IdentityMatrix_A\otimes\BO[\HS_B]$, this is the Jones index $[\M:\N] = d_A^2$. We nevertheless write $\Ind E$ instead of $[\M:\N]$ throughout, since in general the index depends on the conditional expectation and not only on the inclusion $\N\subset \M$. For algebras $\M$ with a nontrivial center,  $\Ind E$ is an element of the center of $\M$ which need not be a scalar (see \Autoref{sec:duality-additive-form}).

Dividing by the index brings the duality relations into their familiar additive form.

\begin{corollary}[Finite-dimensional duality, additive form]\label{cor:finite_dim_duality_subfactor}
	In the setting of \Autoref{thm:finite_dim_duality}, let $\Ind E$ be the index of $E$ as in \eqref{eq:intro-finite-dim-index}, and $E^{-1}=(\Ind E)\,E'$ for a conditional expectation $E':\N'\to\M'$. Then
    \begin{enumerate}\item[(a)]
    for $\alpha,\beta\in[1/2,\infty]$ with $\frac1\alpha+\frac1\beta=2$
	\begin{equation}
		\widetilde{D}_\alpha\big(\rho|_{\M}\,\|\,E_*(\S(\N))\big)+\widetilde{D}_\beta\big(\rho|_{\N'}\,\|\,E'_*(\S(\M'))\big)=\log \Ind E\,,
	\end{equation}
  \item[(b)] for $α, β \in [0,2]$ with $β = 2 - α$
    \begin{equation}
    D_α(ρ|_{\M}\|E_*(ρ|_{\N}))  + D_{β}(ρ|_{\N'}\|E'_*(ρ|_{\M'})) = \log \Ind E
    \end{equation}
    \end{enumerate}
    In particular, this recovers
    \begin{equation}
    D(ρ|_{\M}\|E_*(ρ|_{\N})) + D(ρ|_{\N'}\|E'_*(ρ|_{\M'})) = \log \Ind E\,.
    \end{equation}

\end{corollary}

Taking $\N=\IdentityMatrix_A\otimes\BO[\HS_B]\subset\M=\BO[\HS_A\otimes\HS_B]$ recovers \eqref{eq:(sandwiched) conditional entropy with dim.}, as $\Ind E=d_A^2$ for the trace-preserving conditional expectation. Moreover, \Autoref{thm:finite_dim_duality} applies to a genuinely \emph{reducible} $\N$ of several blocks with multiplicities, which do not correspond to any tensor product subsystem. Examples include observables invariant under a group symmetry, a coarse-grained algebra generated by a pinching map, or a fixed superselection sector (see the explicit examples at the end of \autoref{sec:finite-dim-cond-exp}). To the best of our knowledge, even within finite dimensions, \Autoref{thm:finite_dim_duality} is new for the R\'enyi divergences: prior duality statements for the sandwiched R\'enyi entropy \cite{muller-lennert_quantum_2013, beigi_sandwiched_2013, tomamichel_relating_2014} are only stated for subalgebras of tensor-product form, whereas for the standard relative entropy ($\alpha = 1$) a finite-dimensional subalgebra duality was obtained in \cite{magan_quantum_2021}, see \Autoref{sec:intro-related-work}.

\paragraph{Beyond finite dimensions: von Neumann algebras.}\label{sec:intro-mainthms}
To go beyond \Autoref{thm:finite_dim_duality}, we consider $\M\subseteq \BO$ to be a von Neumann algebra acting on a (possibly infinite-dimensional) Hilbert space $\HS$.  Let $\N \subset \M$ be a von Neumann subalgebra with a normal faithful conditional expectation $E: \M \to \N$. By a theorem of Haagerup \cite[Thms.~5.9, 6.13]{haagerup_operator_1979a}, every such $E$ has a dual normal faithful semi-finite (n.f.s.)\ operator-valued weight $E^{-1}$ from $\N'$ to $\M'$, which can be thought of as a generalized conditional expectation (see \autoref{sec:operator_valued_weights} below). The main technical engine in our argument is the following spatial derivative identity \cite[Rem.~11.26]{hiai_lectures_2021} $$\Delta(E_*(\rho)/\sigma)=\Delta(\rho/E^{-1}_*(\sigma)).$$
Here $\rho$ is a normal state on $\N$, so that $E_*(\rho)$ is a state on $\M$, and $\sigma$ is a normal state on the commutant $\M'$.  Note that for a general operator-valued weight $E^{-1}$, $E_*^{-1}(σ)$ is in general not a linear functional but only a weight. Nevertheless, this does not happen if $E$ has \emph{finite index}, i.e.\ if the index $\Ind E \coloneqq E^{-1}(\IdentityMatrix)$ is a bounded element in the center $\mathcal{Z}(\M)$ (see \autoref{sec:operator_valued_weights}); then everything stays as linear functionals. Note that finite index is a property of $E$, not only of the inclusion $\N \subset \M$.

\iffalse
For a normal state $\rho$. We then represent the von Neumann algebra $\M\subseteq \BO$ such that the state $\rho$ is represented by a vector $\ket{\v{\rho}} \in \HS$. For example, this can be the case if $\HS$ is the GNS representation space for $\M$ with the state $\rho$, or also if $\M$ is in standard form with $\HS = L_2(\M)$. 
\fi 
Our main result is:
\begin{theorem}[Duality relation for conditional entropies]\label{thm:duality_relation}
	Let $\M \subset \BO$ be a von Neumann algebra and  $\rho \in \Ms$ be a normal state that is represented by the vector $\ket{\v{\rho}} \in \HS$. Let $E:\M \to \N$ be a normal faithful conditional expectation of finite index onto a von Neumann subalgebra $\N$ (see \Autoref{sec:operator_valued_weights}). Let $\N'$ and $\M'$ be the commutants of $\N$ and $\M$ within $\BO$, and let $E^{-1}$ be the canonical n.f.s.\ operator-valued weight from the commutant $\N'$ to $\M'$. Then:
	\begin{enumerate}[label=(\alph*)]
		\item  \label{item:dr-sandwiched} For $\alpha \in [1/2, \infty]$ and $\frac{1}{\alpha} + \frac{1}{\beta} = 2$,
		\begin{equation}
			\widetilde{D}_\alpha(\rho|_{\M}\|E_*(\N_*)) = - \widetilde{D}_\beta(\rho|_{\N'}\|E^{-1}_*(\M'_*)),
		\end{equation}
		with both sides allowed to be $\pm\infty$. 
		\item  \label{item:dr-petz} For $\alpha \in [0, 2]$ and $\beta = 2 - \alpha$,
		\begin{equation}
			D_\alpha(\rho|_{\M}\|E_*(\rho|_{\N})) = - D_\beta(\rho|_{\N'} \|E^{-1}_{*}(\rho|_{\M'})),
		\end{equation}
		where $\rho|_{\M}$ and $\rho|_{\N'}$ are the states on $\M$ and $\N'$ induced by $\ket{\v{\rho}}$.
		\item  \label{item:dr-relative-entropy} for the relative entropy ($α = β = 1$),
		\begin{equation}
			D(\rho|_{\M}\|E_*(\N_*)) = D(\rho|_{\M}\|E_*(\rho|_{\N}))= - D(\rho|_{\N'} \|E^{-1}_{*}(\rho|_{\M'})) = - D(\rho|_{\N'} \|E^{-1}_{*}(\M'_*))\,.
		\end{equation}
	\end{enumerate}
\end{theorem}
This is precisely \Autoref{thm:finite_dim_duality} in the von Neumann algebraic setting. It is, to our knowledge, the first duality relation of the sandwiched R\'enyi relative entropy proved for genuinely infinite-dimensional von Neumann algebras. 
Part~\ref{item:dr-sandwiched} includes the two extreme cases $(α, β) = (\infty, 1/2)$ and $(α, β) = (1/2, \infty)$, where $\widetilde{D}_\infty = D_{\max}$ is the max-relative entropy. Part~\ref{item:dr-relative-entropy} -- for subfactor inclusions -- recovers a relative entropy duality previously known due to Xu \cite{xu_relative_2019} and Longo--Witten \cite{longo_note_2022}.

As in the finite-dimensional case, $E^{-1}$ is a conditional expectation up to the index: $E' \coloneqq (\Ind E)^{-1} E^{-1}:\N' \to \M'$ is a normal faithful conditional expectation (\Autoref{lem:dual-conditional-expectation}). If $\Ind E$ is a scalar, which is the case whenever $\M$ is a factor, all three parts of \Autoref{thm:duality_relation} take their familiar additive form:

\begin{corollary}[Additive form of the duality relations for scalar index]\label{cor:finite_index_duality_cond_exp}
	In the setting of \Autoref{thm:duality_relation}, suppose in addition that $\Ind E$ is a finite multiple of the identity. Then $E' \coloneqq (\Ind E)^{-1} E^{-1}:\N' \to \M'$ is a normal faithful conditional expectation, and
	\begin{enumerate}[label=(\alph*)]
		\item for $\alpha,\beta\in[1/2,\infty]$ with $\tfrac1\alpha+\tfrac1\beta=2$,
        $$\widetilde{D}_\alpha(\rho|_{\M}\|E_*(\N_*)) + \widetilde{D}_\beta(\rho|_{\N'}\|E'_*(\M'_*)) = \log \Ind E ;$$
		\item for $\alpha\in[0,2]$ and $\beta=2-\alpha$,
        $$D_\alpha(\rho|_{\M}\|E_*(\rho|_{\N})) + D_\beta(\rho|_{\N'}\|E'_*(\rho|_{\M'})) = \log \Ind E;$$ 
		\item for relative entropy,
        $$D(\rho|_{\M}\|E_*(\rho|_{\N})) + D(\rho|_{\N'}\|E'_*(\rho|_{\M'})) = \log \Ind E.$$
	\end{enumerate}
\end{corollary}

If $\N \subset \M$ is an inclusion of type II$_1$ factors and $E$ is the trace-preserving conditional expectation, then $\Ind E$ is the Jones index $[\M:\N]$. If $\Ind E$ is not a scalar, the index correction depends on the state, see \Autoref{rem:non-scalar-index}.

\medskip

{\bf Application: a Generalized Quantum Stein's Lemma for subalgebras. }\label{sec:intro-gqsl} Quantum Stein's Lemma states that the optimal (type-II) error exponent for asymptotically distinguishing i.i.d.\ copies of a state $\rho$ from i.i.d.\ copies of a fixed state $\sigma$ is given by $D(\rho\|\sigma)$. The \emph{generalized} quantum Stein's lemma (GQSL), first proposed by Brand\~ao and Plenio \cite{brandao_generalization_2010}, extends this to the composite case where the alternative hypothesis is not a single state $\sigma$ but a suitably axiomatized family of states (the \enquote{resource-free} states) -- a statement of considerable importance, since it underlies the interpretation of the regularized relative entropy of a resource as an operationally meaningful, reversible rate of conversion. The original argument \cite{brandao_generalization_2010} for the fully general convex-set version of GQSL was found in 2023 to contain a gap \cite{berta_gap_2023}; this was subsequently fixed, in finite dimensions, by Hayashi and Yamasaki \cite{hayashi_generalized_2025} and, independently, by Lami \cite{lami_solution_2025}.

In parallel, Gao and Rahaman showed, by an independent route, that the GQSL holds whenever the alternative hypothesis is the state space of a \emph{subalgebra} $\N$ of the ambient (finite-dimensional) algebra $\M$ \cite{gao_generalized_2024} -- a setting that covers many resource theories of interest (e.g.\ group-symmetric free states) and for which the relevant convexity/closure axioms are automatic. As an application, we use the duality relations from \Autoref{thm:duality_relation} to lift this subalgebra version of the GQSL to arbitrary von Neumann algebras, as long as the conditioning subalgebra has finite index:

\begin{theorem}[GQSL and AEP for subalgebras]\label{thm:gqsl_subalgebras}
	Let $\M$ be a von Neumann algebra and  let $E: \M \to \N$ be a normal faithful conditional expectation of finite index. 
     Then for every normal state $\rho\in\S(\M)$ and $\varepsilon \in (0, 1)$,
	\begin{equation*}
		\lim_{n \to \infty} \frac{1}{n} D_{\max}^\varepsilon(\rho\n\|E_*\n(\N_*\n)) =\lim_{n \to \infty} \frac{1}{n} D_H^\varepsilon(\rho\n\|E_*\n(\N_*\n))=D(\rho\|E_*(\N_*)) = D(\rho\|E_*(\rho|_\N)),
	\end{equation*}
	where $D_{\max}^\varepsilon$ (resp. $D_H^\varepsilon$) is the smoothed max (resp. hypothesis testing) relative entropy.
\end{theorem}

Operationally, this shows that, if $E$ has finite index, the optimal type-II error for testing i.i.d.\ copies of $\rho$ against \emph{all} $E\n$-invariant states decays as $\exp\big(-nD(\rho\|E_*(\N_*)) + o(n)\big)$. In the two-agent picture of \Autoref{sec:intro-subalgebra}, the states $ρ$ and $E_*(ρ|_\N) = ρ \circ E$ agree on $\N$, so an agent restricted to $\N$ cannot tell them apart at all \cite[Secs.~4.2 and~4.7]{vanluijk_thesis_2025}, whereas an agent with access to $\M$ can distinguish $ρ$ from $ρ \circ E$, and indeed from every state of the form $σ \circ E$ with $σ \in \S(\N)$, at the rate $D(ρ\|E_*(ρ|_\N))$. In this sense, \Autoref{thm:gqsl_subalgebras} gives the relative entropy with respect to the subalgebra an operational meaning as the information about $ρ$ that is available to the stronger agent but invisible to the weaker one, cf.\ \cite{naaijkens_subfactors_2018, gao_relative_2020}. \Autoref{thm:gqsl_subalgebras} is, as far as we are aware, the first extension of the Generalized Quantum Stein's Lemma to a fully infinite-dimensional von Neumann algebraic setting,  with Type-I $\mathcal{B}(\mathcal{H})$, Type-II, and Type-III all included. Its proof uses two duality relations of \Autoref{thm:duality_relation}: the extreme case $(α, β) = (\infty, 1/2)$ of part~\ref{item:dr-sandwiched}, in a smoothed form, and the relative entropy duality of part~\ref{item:dr-relative-entropy}, which identifies the resulting exponent. %Together they make the (originally purely finite-dimensional) achievability/converse argument of Gao--Rahaman go through without any dimension bound.

\subsection{Proof idea: spatial derivatives and Haagerup's identity}
%\comment{
While the proofs of our main results in \autoref{sec:duality-relations} in the von Neumann algebra setting have to deal with quite a few technicalities, the idea behind them is actually fairly simple and makes use of only a few key properties, which extend to the general setting. To illustrate this, let us reprove the known finite-dimensional duality relations for tensor factors (i.e.\ $\N = \IdentityMatrix_A \otimes \BO[\HS_B] \subset \M = \BO[\HS_{AB}] \subset \BO[\HS_{ABC}]$), in a language that will make the extension to arbitrary subalgebras quite immediate.

Let $ρ_{AB} = \sum_i λ_i \ketbra{i}{i}$ be purified as $\ket{\rho_{ABC}} = \sum_{i} \sqrt{λ_i} \ket{i}_{AB}\ket{i}_C$. Writing $\Delta(ρ_{C}/σ_{AB}) \coloneqq ρ_{C} \otimes σ_{AB}^{-1}$, explicit computation finds that
\begin{align}
   \Tr(ρ_{AB}^α σ_{AB}^{1 - α}) = \bra{ρ}\Delta(ρ_C/σ_{AB})^{α - 1}\ket{ρ} \\
   \Tr(ρ (\log ρ - \log σ)) = \bra{ρ} \log(\Delta(ρ_C/σ_{AB}))\ket{ρ}
\end{align}
For $τ_A$ being the maximally mixed state on $A$, 
one easily sees that $$E_*(ω_B) \coloneqq τ_A \otimes ω_B,\qquad E_*^{-1}(ω_{C}) \coloneqq τ_A^{-1} \otimes ω_C$$ satisfy the relations $$\Delta(ρ_C/E_*(σ_{B})) = (τ_A^{-1} \otimes σ_B^{-1} \otimes ρ_C) = \Delta(E^{-1}_*(ρ_C)/σ_{B}), \qquad \Delta({ρ_{C}/σ_{AB}}) = (\Delta(σ_{AB}/ρ_{C}))^{-1}.$$
Then, we obtain the relative entropy duality relation in a simple form
\begin{align}
    D(ρ_{AB}\|E_*(ρ_B)) &= \bra{ρ} \log \Delta(ρ_C/E_*(ρ_B))\ket{ρ} \\&= \bra{ρ} \log\Delta(E^{-1}_*(ρ_C)/ρ_B)\ket{ρ} \\&= - \bra{ρ}\log\Delta(ρ_B/E^{-1}_*(ρ_C)) \ket{ρ} = - D(ρ_{AC}\|E^{-1}_*({ρ_C}))\,.
\end{align}
and the Petz-\renyi duality relation follows from a similar argument.
For the sandwiched \renyi duality relation, we use the following optimized characterization \cite{berta_renyi_2018, gao_recoverability_2021}:
\begin{equation}
    \widetilde{D}_α(ρ_{AB}\|σ_{AB}) = {1 \over α - 1} \log(\Tr(σ^{1 - α \over 2 α } ρ σ^{1 - α \over 2 α})^α) = \sup_{ω_C \in \S(\HS_C)} {1 \over α'} \log \bra{ρ}\Delta(ω_C/σ_{AB})^{α'}\ket{ρ}
\end{equation}
where $α' = {α - 1 \over α}$.
The duality relation proof then follows the same idea:
\begin{align}
    \widetilde{D}_α(ρ_{AB}\|E_*(\S(\N))) &= \inf_{σ_B} \widetilde{D}_α(ρ_{AB}\|E_*(σ_B)) \\
    &= \inf_{σ_B}\sup_{ω_C} {1 \over α'} \log \bra{ρ}\Delta(ω_C/E_*(σ_{B}))^{α'}\ket{ρ} \\
    &= \inf_{σ_B}\sup_{ω_C} {1 \over α'} \log \bra{ρ}\Delta(E_*^{-1}(ω_C)/σ_{B})^{α'}\ket{ρ} \\
    &= \inf_{σ_B}\sup_{ω_C} {1 \over α'} \log \bra{ρ}\Delta(σ_B/E^{-1}_*(ω_C))^{-α'}\ket{ρ} \\
    &= - \sup_{σ_B} \inf_{ω_C} {1 \over (-α')} \log \bra{ρ}\Delta(σ_B/E^{-1}_*(ω_C))^{-α'}\ket{ρ} \\
    &= - \inf_{ω_C} \sup_{σ_B}  {1 \over (-α')} \log \bra{ρ}\Delta(σ_B/E^{-1}_*(ω_C))^{-α'}\ket{ρ} \\
    &= - \widetilde{D}_β(ρ_{AC}\|E^{-1}_*(\S(\M')))
\end{align}
where in the second to last step we employ a minimax theorem to swap the inf and sup.

The key insight to extend these arguments is that all the individual properties used in the proof here hold in the most general setting: the operator $\Delta(ρ/σ)$, called the spatial derivative, is the main ingredient in defining quantum divergences in general von Neumann algebras (see \autoref{sec:prelim-spatial}). It satisfies the identity $\Delta(ρ/σ) = \Delta(σ/ρ)^{-1}$ at least when $ρ$ and $σ$ are strictly positive (for details see \autoref{prop:spatial-derivative-properties}). It has been shown by Haagerup \cite{haagerup_operator_1979, haagerup_operator_1979a} (see also \cite{kosaki_extension_1986}) that such an $E^{-1}$ always exists (see also \autoref{lem:haagerup-identity-finite-index} below), and $\Delta(E_*(ρ)/σ) = \Delta(ρ/E^{-1}_*(σ))$ is the defining equation for the dual weight $E^{-1}$ \cite[Rem.~11.26]{hiai_lectures_2021}. The only part where the generalization is non-obvious is the minimax theorem for the sandwiched duality relation, because the continuity and compactness conditions required are much harder to establish in the general setting. Most of \autoref{sec:duality-relation-sandwiched} is devoted to justifying this step.

\subsection{Related works and organization}\label{sec:intro-related-work}

The idea of studying entropy through relative-entropy-type quantities attached to a general subalgebra goes back to Petz and Ohya's operator-algebraic treatment of quantum entropy \cite{ohya_quantum_1993} and has recently been revisited by Xu \cite{xu_relative_2019} and Longo--Witten \cite{longo_note_2022} for the standard relative entropy in the subfactor case; part~\ref{item:dr-relative-entropy} of \Autoref{thm:duality_relation} recovers and generalizes their duality statement. For inclusions $\N \subset \M \subset \BO$ of finite-dimensional algebras, the duality for the relative entropy was obtained by Mag\'an and Pontello \cite{magan_quantum_2021} under the name \emph{entropic certainty principle}: for a conditional expectation $E: \M \to \N$ whose index is a scalar and its dual conditional expectation $E' = (\Ind E)^{-1} E^{-1}$ they show that $$D(\rho|_{\M}\|E_*(\rho|_{\N})) + D(\rho|_{\N'}\|E'_*(\rho|_{\M'})) = \log \Ind E$$ for every pure state of $\BO$, and they pose the extension to general von Neumann algebras as an open problem. For $\M = \BO$ this is the case $\alpha = 1$ of \Autoref{thm:finite_dim_duality}, written in terms of $E'$ instead of $E^{-1}$, and \Autoref{thm:duality_relation} provides the extension to von Neumann algebras of finite index, to conditional expectations whose index is not a scalar, and to the R\'enyi divergences. On the finite-dimensional/tensor-factor side, the sandwiched R\'enyi duality relation \eqref{duality-f.d cond.} is due to \cite{muller-lennert_quantum_2013, beigi_sandwiched_2013, tomamichel_relating_2014}, and the (finite-dimensional) subalgebra GQSL that we extend to von Neumann algebras is due to Gao and Rahaman \cite{gao_generalized_2024}, itself part of the broader recent effort \cite{berta_gap_2023, hayashi_generalized_2025, lami_solution_2025, fang_generalized_2025} to repair and generalize the Generalized Quantum Stein's Lemma. Our contribution is orthogonal to that effort: rather than enlarging the class of alternative-hypothesis sets within finite dimensions, we show that the ``subalgebra'' axis of generalization survives passage to arbitrary von Neumann algebras.

The reading of an inclusion $\N \subset \M$ as a pair of agents with different access to a quantum system, and of the index as a measure of how much more the stronger agent can do, goes back to Naaijkens \cite{naaijkens_subfactors_2018}; Fiedler, Naaijkens and Osborne \cite{fiedler_jones_2017} turn it into a secret sharing scheme in which the index quantifies the classical information that the stronger agent can hide from the weaker one, and van Luijk's thesis \cite[Sec.~4.10]{vanluijk_thesis_2025} surveys these and related operational interpretations of the index and of sufficient subalgebras. For inclusions of finite von Neumann algebras, Gao, Junge and LaRacuente \cite{gao_relative_2020} characterized the Pimsner--Popa index through relative entropies with respect to the subalgebra of the kind appearing in \Autoref{thm:duality_relation}. Our duality relations and the GQSL are formulated for finite-index inclusions of arbitrary von Neumann algebras, including those of type III, where no trace is available.

The rest of the paper is organized as follows. \Autoref{sec:examples} collects examples of finite-index subalgebras, in finite dimensions, in tracial and type-III von Neumann algebras, and in topologically ordered spin systems, illustrating the scope of \Autoref{thm:finite_dim_duality} and \Autoref{cor:finite_index_duality_cond_exp}. \Autoref{sec:mathematical-preliminary} recalls the operator-algebraic background -- relative modular theory, the entropies of \Autoref{sec:intro-mainthms} defined via the spatial derivative, and Haagerup's theory of operator-valued weights and their duals -- that is needed to make \Autoref{thm:duality_relation} precise. \Autoref{sec:duality-relations} contains the proof of \Autoref{thm:duality_relation} and its corollaries. The technical heart of the paper is built around a minimax exchange (\Autoref{prop:sandwiched_minimax}). Finally, \autoref{sec:gqsl} proves \Autoref{thm:gqsl_subalgebras}.

\section{Examples}\label{sec:examples}
In this section, we give various examples of von Neumann subalgebras, where we specify a finite-index normal faithful conditional expectation $E:\M\to\N$. These are thus examples where our duality relations (\Autoref{thm:duality_relation}, and if the index is a scalar also \Autoref{cor:finite_index_duality_cond_exp}), as well as the generalized quantum Stein's lemma (\autoref{thm:gqsl_subalgebras}) apply. For an inclusion of II$_1$ factors with the trace-preserving expectation, the index is equal to the Jones index, usually denoted by $[\M:\N]$.%We also contrast these examples with infinite-index inclusions not covered by the theorem.

\subsection{Finite-dimensional subalgebras}\label{sec:examples-finite-dim}
\begin{example}[Bipartite subsystems]\label{ex:examples-subsystems}
For $\N=\IdentityMatrix_A\otimes\BO[\HS_B]\subset\M=\BO[\HS_A\otimes\HS_B]$, the trace-preserving expectation is
\[
 E(X)=\IdentityMatrix_A\otimes\frac{\Tr_A X}{d_A},\qquad d_A=\dim\HS_A.
\]
It has $\Ind E=d_A^2$, and \Autoref{cor:finite_index_duality_cond_exp} gives the familiar subsystem duality \eqref{eq:(sandwiched) conditional entropy with dim.}.

\end{example}

\begin{example}[Group-symmetric observables]\label{ex:examples-group-invariant}
Let $G$ be a compact group and $U:G \to \mathcal{B}(\HS)$ be a projective unitary representation on a finite-dimensional $\HS$. Write its irreducible decomposition as
\[
 \HS=\bigoplus_{\lambda\in S}\mathcal L_\lambda\otimes\mathcal V_\lambda,\qquad
 U_g=\bigoplus_{\lambda\in S}\IdentityMatrix_{\mathcal L_\lambda}\otimes u_\lambda(g),
\]
where $u_\lambda$ are irreducible representations that occur in $U$, and $d_\lambda=\dim\mathcal V_\lambda$. The fixed subalgebra and its trace-preserving expectation are
\[
 \N=\BO[\HS]^G
   =\bigoplus_{\lambda\in S}\BO[\mathcal L_\lambda]\otimes\IdentityMatrix_{\mathcal V_\lambda},
 \qquad E(X)=\int_G U_g XU_g^*\,dg .
\]
The index is
\[
 \Ind E=\sum_{\lambda\in S}d_\lambda^2 .
\]
It only depends on which irreducible representations occur, not on their multiplicities. If $G$ is finite, $\Ind E\leq |G|$, with equality exactly when all irreducible representations occur (as for the regular representation). With at least two inequivalent irreducibles the fixed subalgebra has a nontrivial center and cannot be a tensor-factor subalgebra of $\BO[\HS]$. This models observables accessible without a shared reference frame \cite{bartlett_reference_2007}; for quantum reference frames in quantum field theory see \cite{fewster_quantum_2025}. An agent lacking the frame can only implement operations that are covariant under $G$, in the sense that they admit a unitary dilation commuting with $U_g \otimes \IdentityMatrix$ for all $g \in G$, and these are precisely the operations whose Kraus operators lie in the fixed-point algebra $\N$ \cite[Lemma~4.3]{vanluijk_thesis_2025}. For $n$ independent copies, the alternative hypotheses in \Autoref{thm:gqsl_subalgebras} are invariant under the product group $G^n$; this differs from the diagonal action $g\mapsto U_g^{\otimes n}$, whose decomposition relies on fusion rules.

\end{example}

\subsection{Tracial von Neumann algebras}\label{sec:examples-type-II}

\begin{example}[Spin chain]\label{ex:examples-spin-sites}
Let $\M=\bigotimes_{j\in \mathbb{Z}} (M_d(\mathbb{C}), \tr/d )$ be the hyperfinite II$_1$ factor obtained from the product trace on an infinite chain of $d$-level spins. Fix a finite set $F$ of $k$ sites and let $\N$ be generated by the remaining sites. Then $\M\cong\N\bar\otimes M_{d^k}(\mathbb C)$, and the trace-preserving expectation $E=\id_\N\bar\otimes\operatorname{tr}_{d^k}$ has
\[
 \Ind E=d^{2k}.
\]
Thus both $\M$ and $\N$ are infinite-dimensional, although only finitely many local degrees of freedom are omitted. By contrast, retaining every other site, or any periodic sublattice, omits infinitely many sites and gives \emph{infinite} index.

\end{example}

\begin{example}[An outer symmetry of the spin chain]\label{ex:examples-outer-symmetry}
For the product-trace infinite qubit chain $R=\bigotimes_{j\in \mathbb{Z}} (M_2(\mathbb{C}), \tr/2 )$, let $\alpha=\bigotimes_{j\in\mathbb Z}\operatorname{Ad}(\sigma_z)$, defined first on local observables and extended to $R$. This order-two automorphism is outer: the Pauli $\sigma_x$ at a site tending to infinity is a central sequence sent to its negative by $\alpha$. Therefore $\N=R^\alpha\subset R=\M$ is an inclusion of II$_1$ factors, and its trace-preserving expectation $E=(\id+\alpha)/2$ has $\Ind E=2$ \cite{jones_index_1983}. More generally, an \emph{outer} finite-group $G$ action on a II$_1$ factor $R$ gives index $|G|$; outerness is needed for this formula.

\end{example}

\begin{example}[Temperley--Lieb--Jones subfactors]\label{ex:examples-tlj}
For each integer $n\geq4$, the Temperley--Lieb--Jones construction produces an inclusion $\N_n\subset R$ of hyperfinite II$_1$ factors with the trace-preserving expectation \cite{jones_index_1983} and
\[
 \Ind E=4\cos^2(\pi/n).
\]
For $n=5$ the index is $(3+\sqrt5)/2$, giving a noninteger correction $\log \Ind E$ in the additive form of the duality relations. 

\end{example}

\begin{example}[Group inclusions]\label{ex:examples-group-subgroup}
Let $G$ be a countably infinite group with infinite conjugacy classes (ICC), i.e., every non-identity conjugacy class is infinite, and let $H\leq G$ be a subgroup of finite index $m$. Then $H$ is ICC as well, and $L(H)\subset L(G)$ is an inclusion of II$_1$ factors. The trace-preserving expectation discards the Fourier coefficients outside $H$:
\[
 E_H\Big(\sum_{g\in G}a_g\lambda_g\Big)=\sum_{h\in H}a_h\lambda_h,
\]
with index $\Ind E_H=m$ \cite{jones_index_1983}. For example, for $G=F_2$ and $H$ the kernel of a surjection $F_2\to\mathbb Z/m\mathbb Z$, we have $\Ind E_H=[F_2 : H]=m$.

\end{example}

\subsection{Type III local algebras}\label{sec:examples-type-III}
\begin{example}[Localized superselection sectors]
\label{ex:examples-dhr}
Let $\mathcal A$ be a local chiral conformal net on $S^1$ in its
vacuum representation: it assigns to each interval $I$ a von Neumann
algebra $\mathcal A(I)$ of observables localized in $I$.
Assume that the local algebras are type III factors.
If $\rho$ is a transportable Doplicher--Haag--Roberts (DHR)
endomorphism localized in $I$ with finite statistical dimension
$d(\rho)$, then localization and Haag duality give
\[
 \rho(\mathcal A(I))\subseteq\mathcal A(I).
\]
The minimal-index normal faithful conditional expectation
$E:\mathcal A(I)\to\rho(\mathcal A(I))$ satisfies
\[
 \Ind E=d(\rho)^2
\]
\cite{longo_index_1989,longo_index_1990}.
Thus both algebras are associated with the same interval.

For properly nested intervals $I_1\Subset I_2$, the ordinary
inclusion $\mathcal A(I_1)\subset\mathcal A(I_2)$ instead has
infinite index. Indeed, any nonempty interval
$K\subset I_2\setminus\overline{I_1}$ satisfies
\[
 \mathcal A(K)\subseteq
 \mathcal A(I_1)'\cap\mathcal A(I_2)
\]
by locality. The relative commutant is therefore
infinite-dimensional, whereas a finite-index inclusion of factors
has finite-dimensional relative commutant.
\end{example}

\begin{example}[Finite-group orbifolds]
\label{ex:examples-orbifold}
Let $\mathcal A$ be a nontrivial irreducible local chiral conformal
net in its vacuum representation. Suppose a finite group $G$
acts faithfully and properly on $\mathcal A$ by internal symmetries,
in the sense of \cite{xu_orbifold_2000}.
Such an action is outer on each local factor. Consequently,
\[
 \mathcal A(I)^G\subseteq\mathcal A(I)
\]
is an irreducible inclusion of type III factors, with unique normal
faithful conditional expectation
\[
 E_I(x)=\frac{1}{|G|}\sum_{g\in G}\alpha_g(x),
 \qquad \Ind E_I=|G|.
\]
In particular, a faithful proper $\mathbb Z_2$ action gives a local
subfactor inclusion of index $2$. This construction yields a
finite-index subnet; an arbitrary coset or conformal subnet
need not have finite index.
\end{example}

\begin{example}[Two separated intervals]
\label{ex:examples-two-intervals}
Let $\mathcal A$ be a nontrivial completely rational local conformal
net in its vacuum representation, and let $I_1,I_2\subset S^1$
have disjoint closures. Put $J=I_1\cup I_2$ and
$J'=\operatorname{int}(S^1\setminus J)$.
Here $\mathcal A(J')$ denotes the von Neumann algebra generated
by the local algebras of the two connected components of $J'$.
Then
\[
 \N=\mathcal A(I_1)\vee\mathcal A(I_2)
 \ \subseteq\ \M=\mathcal A(J')',
\]
with the commutant taken on the vacuum Hilbert space, is an
irreducible finite-index inclusion of type III factors.
Its unique normal faithful conditional expectation
$E:\M\to\N$ is the minimal-index expectation, and
\[
 \Ind E=\mu_{\mathcal A}=\sum_{[\rho]}d(\rho)^2,
\]
where the sum runs over the equivalence classes of irreducible
DHR sectors of $\mathcal A$
\cite{kawahigashi_longo_muger_multiinterval_2001}.
For the chiral Ising net $\operatorname{Vir}_{1/2}$, the three
statistical dimensions are $1,\sqrt2,1$, so $\mu_{\mathcal A}=4$.
This is the index of the two-interval inclusion of one net.
\end{example}

\subsection{Cone algebras in topologically ordered spin systems}\label{sec:examples-topological}
\begin{example}[Two cones in the toric code]
\label{ex:examples-toric-code}
The two-interval inclusion of \Autoref{ex:examples-two-intervals} has a lattice counterpart in topologically ordered quantum spin systems, where the role of the intervals is played by cones, i.e.\ infinite wedge-shaped regions of the lattice.
Consider Kitaev's toric code on the infinite lattice $\mathbb{Z}^2$ in the GNS representation $π_0$ of its unique ground state, and write $R(Λ) = π_0(\mathcal{A}(Λ))''$ for the von Neumann algebra generated by the quasi-local observables $\mathcal{A}(Λ)$ of a region $Λ$.
For cones $Λ$, the representation $π_0$ satisfies Haag duality, $R(Λ) = R(Λ^c)'$, as well as the distal split property \cite{naaijkens_haag_2012}, and the algebras of convex regions such as cones are of type II$_\infty$ \cite{ogata_type_2024}.
Now let $Γ = Λ_1 \cup Λ_2$ be the union of two cones which are sufficiently far apart, in the sense of \cite[Def.~3.1]{naaijkens_kosakilongo_2013}. Haag duality fails for $Γ$, and
\[
 \N = R(Λ_1) \vee R(Λ_2)\ \subsetneq\ \M = R(Γ^c)'
\]
is an irreducible inclusion of factors \cite[Lemma~3.2]{naaijkens_kosakilongo_2013}.
The operators in $\M$ that are missing from $\N$ are the string operators creating a pair of anyonic excitations with one excitation in each cone: they commute with all observables localized in $Γ^c$, but they are not contained in $\N$.
In fact, $\M$ is generated by $\N$ together with two such unitaries, one for each of the two generating anyon types, and $\M$ is isomorphic to the crossed product of $\N$ by the resulting action of $G = \mathbb{Z}_2 \times \mathbb{Z}_2$ \cite[Lemmas~4.1 and~4.7]{naaijkens_kosakilongo_2013}.
Consequently, the unique normal faithful conditional expectation $E: \M \to \N$ has
\[
 \Ind E = \abs{G} = \sum_{[ρ]} d(ρ)^2 = 4\,,
\]
the total quantum dimension of the toric code, whose four superselection sectors all have quantum dimension one \cite[Thm.~4.9]{naaijkens_kosakilongo_2013}.
Fiedler, Naaijkens and Osborne \cite{fiedler_jones_2017} discuss the same structure for abelian quantum double models and interpret $\log \Ind E$ as the amount of classical information that two parties acting in $Λ_1$ and $Λ_2$ can hide from an adversary confined to $Γ^c$.
Note that in this example the commutant $\M' = R(Γ^c)$ is again the algebra of a region, the complement of the two cones, so that the duality relations of \Autoref{thm:duality_relation} relate entropies with respect to the two-cone region to entropies with respect to its complement.
\end{example}

\section{Mathematical Preliminaries}\label{sec:mathematical-preliminary}
This section collects the operator-algebraic background needed to make the statements of \Autoref{thm:duality_relation} and \Autoref{thm:gqsl_subalgebras} precise. \Autoref{sec:prelim-vn} is a short dictionary between finite-dimensional quantum information theory and the language of von Neumann algebras; readers familiar with the latter can skip it, apart from the conventions on unbounded and \enquote{generalized} positive operators at its end, which are used throughout. \Autoref{sec:prelim-cond-exp} defines conditional expectations and the map $E_*$ they induce on states, \Autoref{sec:prelim-spatial} recalls Connes' spatial derivative, the object that takes over the role of the density matrix, and \Autoref{sec:prelim-divergences} contains the definitions of all entropic quantities that appear in this paper. General references for von Neumann algebras are \cite{takesaki_theory_1979, takesaki_theory_2003, bratteli_operator_1987}. For modular theory, standard forms and spatial derivatives we follow Hiai's lecture notes \cite{hiai_lectures_2021} (theorem numbers below refer to the published version; the numbering in the earlier arXiv version \cite{hiai_concise_2020} differs), and for quantum divergences in von Neumann algebras we refer to \cite{ohya_quantum_1993, hiai_quantum_2021}.

\subsection{Von Neumann algebras for the finite-dimensional reader}\label{sec:prelim-vn}

\paragraph{Algebras and commutants.}
Throughout, $\HS$ is a Hilbert space of arbitrary dimension with inner product $\ip{\cdot}{\cdot}$, linear in the second argument, and $\BO$ is the algebra of bounded operators on $\HS$. The \emph{commutant} of a set $S \subset \BO$ is $S' \coloneqq \Set{y \in \BO | xy = yx \text{ for all } x \in S}$. A \emph{von Neumann algebra} is a unital $*$-subalgebra $\M \subset \BO$ with $\M = \M''$; by von Neumann's bicommutant theorem this is the same as a unital $*$-subalgebra that is closed in the weak (equivalently, the strong) operator topology \cite[Thm.~II.3.9]{takesaki_theory_1979}. The commutant $\M'$ of a von Neumann algebra is again a von Neumann algebra, and since $\M'' = \M$ the two algebras $\M$ and $\M'$ play completely symmetric roles. Their intersection $Z(\M) \coloneqq \M \cap \M' = Z(\M')$ is the \emph{center}, and $\M$ is a \emph{factor} if $Z(\M) = \complex \IdentityMatrix$. A von Neumann algebra contains the spectral projections of all its self-adjoint elements, so it is generated by its projections. We write $\M_+$ for the cone of positive elements of $\M$.

In finite dimensions the structure of a von Neumann algebra is elementary (see e.g.\ \cite[Sec.~4]{gao_generalized_2024} for a discussion in the present context): after a unitary change of basis,
\begin{equation}\label{eq:finite-dim-structure}
    \HS = \bigoplus_{k} \HS_k \otimes \mathcal{K}_k, \qquad
    \M = \bigoplus_{k} \BO[\HS_k] \otimes \IdentityMatrix_{\mathcal{K}_k}, \qquad
    \M' = \bigoplus_{k} \IdentityMatrix_{\HS_k} \otimes \BO[\mathcal{K}_k]\,,
\end{equation}
i.e.\ $\M$ consists of block-diagonal operators, each block carrying a multiplicity space $\mathcal{K}_k$ on which $\M$ acts trivially, and $\M'$ is the algebra acting on the multiplicity spaces only. The case of a single block with $\HS = \HS_A \otimes \HS_B$, $\M = \BO[\HS_A] \otimes \IdentityMatrix_B$ and $\M' = \IdentityMatrix_A \otimes \BO[\HS_B]$ is the familiar subsystem picture: the observables of $A$ form a factor, and the commutant consists of the observables of $B$. The case $\HS_k = \complex$ for all $k$ gives the operators that are diagonal in a fixed orthonormal basis, i.e.\ a classical register, for which $\M' = \M$. Between these extremes lie the algebras describing symmetries and superselection rules: the operators commuting with a unitary representation of a group, or the block-diagonal operators with respect to a decomposition into sectors, which are exactly the images of pinching maps. In infinite dimensions, algebras of the form \eqref{eq:finite-dim-structure} (with infinite-dimensional $\HS_k$, $\mathcal{K}_k$ allowed) are the von Neumann algebras of \emph{type~I}. There are others: the algebras of type II and III of the Murray--von Neumann classification (see \cite[Ch.~V]{takesaki_theory_1979} or \cite[Sec.~1.4]{hiai_lectures_2021}) contain no minimal projections, and those of type III, which describe local observables in quantum field theory, do not even admit a trace. For them nothing in the finite-dimensional toolbox survives: there are no density matrices, no pure normal states and no partial traces. What survives is the algebra, its commutant and the normal states, and this is all we shall use.

\paragraph{States.}
A linear functional $ω$ on $\M$ is \emph{normal} if it is continuous with respect to the $σ$-weak topology; for positive $ω$ this is the same as $ω(x_α) \nearrow{ω}(x)$ for every bounded increasing net $x_α \nearrow x$ in $\M_+$ (see \cite[Ch.~III]{takesaki_theory_1979} or \cite[Sec.~1.3]{hiai_lectures_2021}). The normal functionals form the \emph{predual} $\Ms$ of $\M$ (indeed $\M = (\Ms)^*$ as Banach spaces, and this property characterizes von Neumann algebras among $C^*$-algebras \cite{sakai_calgebras_1998}). We write $\M_*^+$ for the positive normal functionals and
\begin{equation}
    \S(\M) \coloneqq \Set{ρ \in \M_*^+ | ρ(\IdentityMatrix) = 1}
\end{equation}
for the \emph{normal states}.\footnote{Some authors write $\S_*(\M)$ to stress normality. In this paper all states are normal unless explicitly said otherwise, so we drop the subscript.} For $\M = \BO$ the predual is the space of trace-class operators, $ρ(x) = \Tr(\hat{ρ} x)$, and $\S(\BO)$ is the set of density operators. In finite dimensions every linear functional is normal, and every state of $\M \subset \BO$ is the restriction to $\M$ of a density matrix on $\HS$ (two density matrices define the same state of $\M$ iff they have the same pinching onto the block structure \eqref{eq:finite-dim-structure}). In infinite dimensions $\M$ also has non-normal states, for instance limits of normal states along a free ultrafilter; they are not countably additive on orthogonal families of projections and have no physical interpretation. They play a technical role in the minimax argument of \Autoref{sec:duality-relations}, but all states appearing in the statements of this paper are normal. For $ρ \in \M_*^+$ and a von Neumann subalgebra $\N \subset \M$ (so that $\M' \subset \N'$) we write $ρ|_{\N} \in \N_*^+$ for the restriction; when $\N = \IdentityMatrix_A \otimes \BO[\HS_B] \subset \BO[\HS_A \otimes \HS_B] = \M$ the density matrix of $ρ|_\N$ is the partial trace of $ρ = ρ_{AB}$ over $A$. The cone $\M_*^+$ defines an order on $\Ms$: $ρ \leq{σ}$ iff $σ - ρ \in \M_*^+$, i.e.\ iff $ρ(x) \leq{σ}(x)$ for all $x \in \M_+$. Note that $\norm{ρ} = ρ(\IdentityMatrix)$ for $ρ \in \M_*^+$.

\paragraph{Supports.}
For $ρ \in \M_*^+$ the \emph{support projection} $s(ρ) \in \M$ is the smallest projection $p \in \M$ with $ρ(p) = ρ(\IdentityMatrix)$; equivalently, $\IdentityMatrix - s(ρ)$ is the largest projection on which $ρ$ vanishes. One has $ρ(x) = ρ(s(ρ)\, x\, s(ρ))$ for all $x \in \M$, $ρ(p) = 0$ iff $p \leq \IdentityMatrix - s(ρ)$ for projections $p \in \M$, and $ρ \leq{σ}$ implies $s(ρ) \leq s(σ)$. We call $ρ$ \emph{faithful} if $s(ρ) = \IdentityMatrix$, i.e.\ if $ρ(x^*x) = 0$ implies $x = 0$. For $\M = \BO$, $s(ρ)$ is the projection onto the range of the density operator, and $s(ρ) \leq s(σ)$ is the familiar support condition for the finiteness of $D(ρ\|σ)$. Note that if $ρ \in \M_*^+$ and $ω' \in (\M')_*^+$, then $s(ρ) \in \M$ and $s(ω') \in \M'$ commute.

\paragraph{Vector states, cyclic and separating vectors, and the standard form.}
Every unit vector $ξ \in \HS$ defines a normal state $ω_ξ(x) \coloneqq \ip{ξ}{x ξ}$ on every von Neumann algebra on $\HS$, in particular on $\M$ and on $\M'$. We say that $ξ$ is a \emph{vector representative} of $ρ \in \S(\M)$ if $ω_ξ|_{\M} = ρ$. We write $[\M{ξ}]$ for the projection onto the closed subspace $\overline{\M{ξ}}$; since $\M{ξ}$ is invariant under $\M$, this projection lies in $\M'$, and it is the support of the vector state on the commutant, $[\M{ξ}] = s(ω_ξ|_{\M'})$. Likewise $[\M' ξ] = s(ω_ξ|_{\M}) \in \M$. The vector $ξ$ is \emph{cyclic} for $\M$ if $[\M{ξ}] = \IdentityMatrix$ and \emph{separating} for $\M$ if $xξ = 0$ implies $x = 0$ for $x \in \M$; $ξ$ is cyclic for $\M$ iff it is separating for $\M'$ iff $ω_ξ|_{\M'}$ is faithful.

Every normal state $ρ$ of $\M$ has a vector representative in \emph{some} representation of $\M$ (for instance in the GNS representation of $ρ$), but a given Hilbert space $\HS$ need not contain vector representatives of all normal states: $\BO$ acting on $\HS$ has only the pure states as vector states. In finite dimensions one enlarges $\HS$ by a purifying system; the von Neumann algebraic counterpart is the \emph{standard form} \cite{araki_properties_1974, connes_caracterisation_1974, haagerup_standard_1975} (see \cite[Ch.~3]{hiai_lectures_2021} or \cite[Ch.~IX]{takesaki_theory_2003}). This is a representation $\M \subset \BO$ together with an antiunitary involution $J$ on $\HS$ and a self-dual closed convex cone $\mathcal{P} \subset \HS$, the \emph{natural cone}, such that $J \M J = \M'$, $JzJ = z^*$ for $z \in Z(\M)$, $Jξ = ξ$ for $ξ \in \mathcal{P}$ and $xJxJ\mathcal{P} \subset \mathcal{P}$ for $x \in \M$ \cite[Def.~3.5]{hiai_lectures_2021}. Its defining feature is that $ξ \mapsto{ω}_ξ|_{\M}$ is a bijection (indeed a homeomorphism) from $\mathcal{P}$ onto $\M_*^+$ \cite[Thm.~3.12]{hiai_lectures_2021}: every $ρ \in \M_*^+$ has a unique vector representative $ξ_ρ \in \mathcal{P}$. The standard form is unique up to unitary equivalence \cite[Thm.~3.13]{hiai_lectures_2021}. Since $J\M J = \M'$, the commutant is in standard form as well, with the same $J$ and $\mathcal{P}$, so that $ρ \mapsto{ω}_{ξ_ρ}|_{\M'}$ is a bijection between the normal states of $\M$ and of $\M'$, with $ω_{ξ_ρ}(JxJ) = ρ(x^*)$ for $x \in \M$.

For $\M = \BO[\HS_A]$ the standard form is the space of Hilbert--Schmidt operators on $\HS_A$, on which $\M$ acts by left multiplication, $J$ is the adjoint, $\mathcal{P}$ consists of the positive Hilbert--Schmidt operators, and $ξ_ρ = \hat{ρ}^{1/2}$ \cite[Ex.~3.6]{hiai_lectures_2021}. Identifying Hilbert--Schmidt operators with vectors in $\HS_A \otimes \overline{\HS_A}$ via $\ketbra{i}{j} \mapsto \ket{i} \otimes \ket{\bar{\jmath}}$, the left multiplications become $\BO[\HS_A] \otimes \IdentityMatrix$, the commutant (the right multiplications) becomes $\IdentityMatrix \otimes \BO[\overline{\HS_A}]$, and $ξ_ρ$ becomes the canonical purification $\sum_i \sqrt{p_i}\, \ket{i} \otimes \ket{\bar{\imath}}$ of $\hat{ρ} = \sum_i p_i \ketbra{i}{i}$. The standard form is thus the von Neumann algebraic version of purification, with the commutant $\M'$ in the role of the purifying system. This is the intuition behind \Autoref{sec:intro-subalgebra}: the \enquote{complementary system} of a subalgebra $\N \subset \M \subset \BO$ is its commutant $\N' \supset \M'$.

\emph{Convention.} In many statements below, given $ρ \in \S(\M)$ we fix a representation $\M \subset \BO$ together with a vector representative $\ket{\v{ρ}} \in \HS$ of $ρ$ (for instance the standard form, or the GNS representation of $ρ$). For an arbitrary von Neumann algebra $\mathcal{A} \subset \BO$ we then write $ρ|_{\mathcal{A}} \coloneqq{ω}_{\v{ρ}}|_{\mathcal{A}}$. For $\mathcal{A} \subset \M$ this is the restriction of $ρ$ and does not depend on the choice of $\ket{\v{ρ}}$; for $\mathcal{A} = \M'$ or $\mathcal{A} = \N'$ it does, and it is understood that $\ket{\v{ρ}}$ has been fixed. We abbreviate $ρ' \coloneqq{ρ}|_{\M'} \in \S(\M')$.

\paragraph{Weights.}
A \emph{weight} on $\M$ is a map $φ: \M_+ \to [0, \infty]$ that is additive and positively homogeneous (with $0 \cdot \infty = 0$), i.e.\ a positive functional that is allowed to take the value $\infty$. It is \emph{normal} if $φ(x_α) \nearrow{φ}(x)$ for all bounded increasing nets $x_α \nearrow x$ in $\M_+$; by Haagerup's theorem \cite{haagerup_normal_1975} (see \cite[Thm.~6.2]{hiai_lectures_2021}) this is equivalent to $φ$ being a pointwise sum of elements of $\M_*^+$, and to $φ(x) = \sup\Set{ω(x) | ω \in \M_*^+,\, ω \leq{φ}}$. It is \emph{faithful} if $φ(x^*x) = 0$ implies $x = 0$, and \emph{semifinite} if $\Set{x \in \M | φ(x^*x) < \infty}$ is $σ$-weakly dense in $\M$; \emph{n.f.s.} abbreviates \enquote{normal, faithful and semifinite}. The normal weights with $φ(\IdentityMatrix) < \infty$ are exactly the elements of $\M_*^+$. The basic example of an unbounded weight is the trace $\Tr$ on $\BO$, which is n.f.s.; a von Neumann algebra admits an n.f.s.\ trace iff it is of type I or II, and type III algebras admit no semifinite normal trace at all. Weights enter this paper in two ways: as the reference weight in the definition of the spatial derivative, and through Haagerup's operator-valued weights in \Autoref{sec:operator_valued_weights}. All entropic quantities we define will have positive normal functionals in both arguments.
\paragraph{Unbounded and generalized positive operators; conventions.}
Let $A$ be a positive self-adjoint operator on $\HS$, possibly unbounded, with spectral resolution $A = \int_{[0,\infty)} λ\, \dd E_A(λ)$. For $ξ \in \HS$ we write $μ^A_ξ$ for the spectral measure $μ^A_ξ(B) = \ip{ξ}{E_A(B) ξ}$, and for a Borel function $f: [0, \infty) \to [0, \infty]$ we define
\begin{equation}\label{eq:form-convention}
    \ip{ξ}{f(A) ξ} \coloneqq \int f(λ) \, \dd{μ}_ξ^A(λ) \in [0, \infty]\,,
\end{equation}
which equals $\norm{f(A)^{1/2} ξ}^2$ if $ξ$ lies in the domain $\mathcal{D}(f(A)^{1/2})$ and $+\infty$ otherwise. In particular $\ip{ξ}{Aξ} = +\infty$ for $ξ \notin \mathcal{D}(A^{1/2})$; the map $ξ \mapsto \ip{ξ}{Aξ}$ is the lower semicontinuous quadratic form of $A$, which determines $A$ \cite[Thm.~A.11]{hiai_lectures_2021}. The \emph{support} of $A$ is the projection $s(A) \coloneqq E_A((0, \infty))$ onto the orthogonal complement of $\ker A$; $A$ is a bijection of $s(A)\mathcal{D}(A)$ onto $s(A)\HS$, and we write $A^{-1}$ for its inverse on $s(A)\HS$.

To keep track of kernels under inversion it is convenient to allow the value $+\infty$ on a subspace. Following Haagerup \cite{haagerup_operator_1979} (see \cite[Sec.~7.1]{hiai_lectures_2021}), a \emph{generalized positive operator} on $\HS$ is a pair $(\mathcal{K}, A_0)$ consisting of a closed subspace $\mathcal{K} \subset \HS$ and a positive self-adjoint operator $A_0$ on $\mathcal{K}$; we think of it as \enquote{$A_0$ on $\mathcal{K}$ and $+\infty$ on $\mathcal{K}^\perp$} and simply write $A$ for it, extending the notation \eqref{eq:form-convention} by
\begin{equation}\label{eq:generalized-form-convention}
    \ip{ξ}{Aξ} \coloneqq \begin{cases} \norm{A_0^{1/2}ξ}^2 & \text{if } ξ \in \mathcal{D}(A_0^{1/2}) \subset \mathcal{K}, \\ +\infty & \text{otherwise.}\end{cases}
\end{equation}
This is a lower semicontinuous positive form on $\HS$, and conversely every lower semicontinuous positive form on $\HS$ arises in this way from a unique generalized positive operator \cite[Thm.~A.16]{hiai_lectures_2021}. We write $\widehat{\BO_+}$ for the set of generalized positive operators; it contains the positive self-adjoint operators ($\mathcal{K} = \HS$) and is ordered by $A \leq B$ iff $\ip{ξ}{Aξ} \leq \ip{ξ}{Bξ}$ for all $ξ$. The spectral measure of $ξ$ with respect to $A = (\mathcal{K}, A_0)$ is the measure on the extended half-line $[0, \infty]$ given by $μ^A_ξ \coloneqq{μ}^{A_0}_{P_{\mathcal{K}}ξ} + \norm{(\IdentityMatrix - P_{\mathcal{K}})ξ}^2 δ_\infty$, where $P_{\mathcal{K}}$ is the projection onto $\mathcal{K}$. For a Borel function $f: [0, \infty] \to [0, \infty]$ we let $f(A) \in \widehat{\BO_+}$ be \enquote{$f(A_0)$ on $\mathcal{K}$ and $f(\infty)$ on $\mathcal{K}^\perp$}, so that again $\ip{ξ}{f(A)ξ} = \int_{[0, \infty]} f \, \dd{μ}^A_ξ$. Two special cases are used constantly:
\begin{itemize}
    \item The \emph{generalized inverse} $A^{-1}$, obtained from $f(λ) = 1/λ$ with $1/0 = \infty$ and $1/\infty = 0$. Thus $A^{-1}$ is the inverse of $A_0$ on $s(A_0)\mathcal{K}$, it is $+\infty$ on $\ker A_0$, and it vanishes on $\mathcal{K}^\perp$; in particular $\ip{ξ}{A^{-1}ξ} = \infty$ whenever $ξ$ has a nonzero component in $\ker A$, and $(A^{-1})^{-1} = A$.
    \item The powers $A^{s}$, $s \in \reals \setminus \set{0}$, with $0^s = \infty$ and $\infty^s = 0$ for $s < 0$, and $0^s = 0$, $\infty^s = \infty$ for $s > 0$; note $A^{-s} = (A^{-1})^{s}$.
\end{itemize}
For real-valued $f$, in particular for $f = \log$ with $\log 0 = -\infty$ and $\log \infty = +\infty$, we define $\ip{ξ}{f(A) ξ} \coloneqq \int_{[0, \infty]} f \, \dd{μ}^A_ξ$ whenever the integral exists in $[-\infty, +\infty]$. Finally, a generalized positive operator $A$ with $\ip{ξ}{Aξ} \leq c \norm{ξ}^2$ for all $ξ$ is a bounded operator with $0 \leq A \leq c \IdentityMatrix$, and an increasing sequence $A_n \nearrow A$ in $\widehat{\BO_+}$ (i.e.\ $\ip{ξ}{A_n ξ} \nearrow \ip{ξ}{A ξ}$ for all $ξ$) always has a limit in $\widehat{\BO_+}$ \cite[Prop.~A.19]{hiai_lectures_2021}.

For a von Neumann algebra $\M \subset \BO$, the \emph{extended positive part} $\widehat{\M_+}$ consists of those $A \in \widehat{\BO_+}$ whose form is invariant under the unitaries of $\M'$, equivalently $\mathcal{K} = P_{\mathcal{K}}\HS$ with $P_{\mathcal{K}} \in \M$ and $A_0$ affiliated with $P_{\mathcal{K}} \M P_{\mathcal{K}}$ \cite[Thm.~7.3, Rem.~7.4]{hiai_lectures_2021}. Equivalently, and this is Haagerup's original definition, an element of $\widehat{\M_+}$ is a lower semicontinuous, additive and positively homogeneous map $m: \M_*^+ \to [0, \infty]$, the correspondence being $m(ω_ξ) = \ip{ξ}{Aξ}$ \cite[Def.~7.1]{hiai_lectures_2021}. Every normal weight $φ$ on $\M$ extends uniquely to an additive, positively homogeneous and normal map $φ: \widehat{\M_+} \to [0, \infty]$ \cite[Prop.~7.5]{hiai_lectures_2021}, and an element $m \in \widehat{\M_+}$ with $m \leq c\IdentityMatrix$ belongs to $\M_+$. The extended positive part is the natural range of operator-valued weights, see \Autoref{sec:operator_valued_weights}.

\subsection{Conditional expectations}\label{sec:prelim-cond-exp}
In the subsystem picture of \Autoref{sec:intro-subalgebra}, conditioning on $B$ involves two maps: the partial trace, which restricts a state of $AB$ to the observables of $B$, and the embedding $σ_B \mapsto \IdentityMatrix_A/d_A \otimes σ_B$ of the states of $B$ into the states of $AB$. For a general subalgebra $\N \subset \M$ the first map is the restriction $ρ \mapsto ρ|_{\N}$, and the second one is induced by a conditional expectation.

\begin{definition}[Conditional expectation and the induced map on states]\label{def:conditional-expectation}
    Let $\N \subset \M$ be a von Neumann subalgebra. A \emph{conditional expectation} from $\M$ onto $\N$ is a linear map $E: \M \to \N$ that is positive, unital and $\N$-bimodular:
    \begin{equation}\label{eq:conditional-expectation}
        E(x^*x) \geq 0\,, \qquad E(\IdentityMatrix) = \IdentityMatrix\,, \qquad E(a x b) = a\, E(x)\, b \qquad \text{for all } x \in \M \text{ and } a, b \in \N\,.
    \end{equation}
    It is \emph{normal} if $E(x_α) \nearrow E(x)$ for every bounded increasing net $x_α \nearrow x$ in $\M_+$, and \emph{faithful} if $E(x^*x) = 0$ implies $x = 0$. If $E$ is normal, the \emph{induced map on states} is
    \begin{equation}\label{eq:induced-map-on-states}
        E_*: \Ns \to \Ms\,, \qquad E_*(σ) \coloneqq σ \circ E\,, \qquad \text{i.e.} \quad E_*(σ)(x) = σ(E(x)) \quad \text{for } x \in \M\,.
    \end{equation}
\end{definition}

Taking $x = \IdentityMatrix$ in \eqref{eq:conditional-expectation} gives $E(a) = a$ for $a \in \N$, so $E$ is a projection of $\M$ onto $\N$, and $\norm{E} = \norm{E(\IdentityMatrix)} = 1$. Positivity and bimodularity together imply that $E$ is completely positive, and conversely every projection of norm one from $\M$ onto $\N$ is a conditional expectation by Tomiyama's theorem \cite{tomiyama_projection_1957} (see \cite[Thm.~5.5]{hiai_lectures_2021}). In the language of quantum information theory, a normal conditional expectation is thus a quantum channel in the Heisenberg picture (a normal unital completely positive map) that maps the observables of $\M$ to observables of $\N$ and leaves the latter unchanged, and its predual map $E_*$ is the same channel in the Schr\"odinger picture.

Normality of $E$ guarantees that $σ \circ E$ is again normal. The map $E_*$ is linear and positive, and since $E$ is unital and acts as the identity on $\N$, we have $E_*(σ)(\IdentityMatrix) = σ(\IdentityMatrix)$ and $E_*(σ)|_{\N} = σ$ for all $σ \in \Ns$. Thus $E_*$ maps $\N_*^+$ into $\M_*^+$ and $\S(\N)$ into $\S(\M)$, and it is a right inverse of the restriction map: among all extensions of a state $σ$ of $\N$ to a state of $\M$, the conditional expectation singles out the extension $E_*(σ)$. A state $ρ \in \S(\M)$ is of this form iff it is $E$-invariant: $ρ = E_*(σ)$ implies $ρ \circ E = σ \circ E \circ E = ρ$, and $ρ \circ E = ρ$ implies $ρ = E_*(ρ|_{\N})$. Hence
\begin{equation}\label{eq:invariant-states}
    E_*(\S(\N)) = \Set{ρ \in \S(\M) | ρ \circ E = ρ}\,.
\end{equation}

While in finite dimensions every conditional expectation is normal, and there is a faithful conditional expectation onto every subalgebra (for instance the trace-preserving one), a normal conditional expectation onto a given von Neumann subalgebra need not exist in general. By Takesaki's theorem \cite{takesaki_conditional_1972} (see \cite[Thm.~5.6]{hiai_lectures_2021}), for a faithful normal state $ω$ of $\M$ there is a normal conditional expectation $E: \M \to \N$ with $ω \circ E = ω$ iff $\N$ is invariant under the modular automorphism group of $ω$. If $\M$ has a faithful normal tracial state, then the trace-preserving conditional expectation onto $\N$ always exists \cite{umegaki_conditional_1954} (see \cite[Thm.~5.4]{hiai_lectures_2021}). In this paper a normal faithful conditional expectation $E: \M \to \N$ is always assumed. Dropping unitality and boundedness leads to the operator-valued weights of \Autoref{sec:operator_valued_weights}.

\subsection{The spatial derivative}\label{sec:prelim-spatial}

In finite dimensions the relative entropy $D(ρ\|σ) = \Tr\big(ρ(\log ρ - \log σ)\big)$ and its \renyi relatives are built from density matrices, i.e.\ from the trace. Since type III algebras have no trace, we need a substitute. The substitute is not the density matrix itself but the \emph{ratio} $ρ/σ$, understood as an operator on the representation space: in the subsystem picture $\M = \BO[\HS_A] \otimes \IdentityMatrix_C$ on $\HS_A \otimes \HS_C$, with commutant $\M' = \IdentityMatrix_A \otimes \BO[\HS_C]$, a state $φ$ of $\M$ is a density matrix $ρ_A$, a state $ψ$ of $\M'$ is a density matrix $σ_C$, and the operator
\begin{equation}\label{eq:spatial-derivative-finite-dim}
    \Delta(φ/ψ) = ρ_A \otimes{σ}_C^{-1}
\end{equation}
carries all information about $φ$ and $ψ$ and still makes sense when the two algebras involved are not tensor factors: this is Connes' \emph{spatial derivative} \cite{connes_spatial_1980}. Note that the \enquote{numerator} $φ$ lives on $\M$ and the \enquote{denominator} $ψ$ on the commutant $\M'$; in the standard form $\Delta(φ/ψ)$ is Araki's relative modular operator. We recall the construction following \cite[Sec.~11.1]{hiai_lectures_2021}, where the second argument is a faithful weight, and \cite[Ch.~4]{ohya_quantum_1993}, where both arguments are positive normal functionals that need not be faithful. The latter is the case that occurs in the duality relations.

\begin{definition}[Spatial derivative]\label{def:spatial-derivative}
    Let $\M \subset \BO$ be a von Neumann algebra, and let $ψ$ be a positive normal functional on $\M'$, which need not be faithful, or an n.f.s.\ weight on $\M'$. Let $e \coloneqq s(ψ) \in \M'$ be the support of $ψ$, let $\HS_ψ$ be the GNS Hilbert space of $ψ$, i.e.\ the completion of $\mathfrak{n}_ψ \coloneqq \Set{y \in \M' | ψ(y^*y) < \infty}$ with respect to $\ip{y}{z}_ψ \coloneqq ψ(y^*z)$, after dividing out the null space $\Set{y \in \M' | ψ(y^*y) = 0} = \M'(\IdentityMatrix - e)$, and let $η_ψ: \mathfrak{n}_ψ \to \HS_ψ$ be the canonical map. A vector $ξ \in \HS$ is called \emph{$ψ$-bounded} if there is a constant $C$ with $\norm{yξ}^2 \leq C\, ψ(y^*y)$ for all $y \in \mathfrak{n}_ψ$. The $ψ$-bounded vectors form an $\M$-invariant subspace $D(\HS, ψ)$ with closure $e\HS$, and for $ξ \in D(\HS, ψ)$ the map $η_ψ(y) \mapsto yξ$ extends to a bounded operator $R^ψ(ξ): \HS_ψ \to \HS$ with $R^ψ(xξ) = xR^ψ(ξ)$ for $x \in \M$, for which
    \begin{equation}
        θ^ψ(ξ, ξ) \coloneqq R^ψ(ξ) R^ψ(ξ)^* \in \M_+\,, \qquad θ^ψ(xξ, xξ) = x\, θ^ψ(ξ, ξ)\, x^* \quad \text{for } x \in \M
    \end{equation}
    \cite[Lemma~11.2]{hiai_lectures_2021}, \cite[Eqs.~(4.1)--(4.3)]{ohya_quantum_1993}. For $φ \in \M_*^+$ (more generally, for any semifinite normal weight $φ$ on $\M$) the quadratic form $q_φ: ξ \mapsto φ(θ^ψ(ξ, ξ))$ on $D(\HS, ψ)$ is lower semicontinuous and closable \cite[Lemma~11.3]{hiai_lectures_2021}, \cite[p.~68]{ohya_quantum_1993}. The \emph{spatial derivative} $\Delta(φ/ψ) \in \widehat{\BO_+}$ (Connes' notation is $\dd{φ}/\dd{ψ}$) is the generalized positive operator that on $e\HS$ is the positive self-adjoint operator associated with the closure $\overline{q_φ}$ of this form, and that is $+\infty$ on $(\IdentityMatrix - e)\HS$. Thus
    \begin{equation}\label{eq:spatial-derivative-form}
        \ip{ξ}{\Delta(φ/ψ)\, ξ} = φ\big(θ^ψ(ξ, ξ)\big)\,, \qquad{ξ} \in D(\HS, ψ)\,,
    \end{equation}
    the form domain $\Set{ξ \in \HS | \ip{ξ}{\Delta(φ/ψ)ξ} < \infty}$ is the domain of $\overline{q_φ}$, and $D(\HS, ψ)$ is dense in it with respect to the form norm $ξ \mapsto \big(\norm{ξ}^2 + \ip{ξ}{\Delta(φ/ψ)ξ}\big)^{1/2}$. If $ψ$ is faithful, then $e = \IdentityMatrix$ and $\Delta(φ/ψ)$ is a positive self-adjoint operator on $\HS$, for which $D(\HS, ψ)$ is a core of $\Delta(φ/ψ)^{1/2}$ \cite[Def.~11.4]{hiai_lectures_2021}.
\end{definition}

Some remarks help to decode this definition. The operator $θ^ψ(ξ,ξ) \in \M_+$ should be thought of as \enquote{$\ketbra{ξ}{ξ}$ divided by $ψ$}: if $\M = \BO$, so that $\M' = \complex\IdentityMatrix$ and $ψ = c \cdot (\text{the trivial state})$ for some $c > 0$, then every vector is $ψ$-bounded and $θ^ψ(ξ, ξ) = c^{-1} \ketbra{ξ}{ξ}$, so that $\ip{ξ}{\Delta(φ/ψ)ξ} = c^{-1}φ(\ketbra{ξ}{ξ})$ and $\Delta(φ/ψ) = c^{-1} \hat{φ}$ is (a multiple of) the density operator of $φ$ \cite[Ex.~11.5]{hiai_lectures_2021}. In the subsystem picture above one obtains \eqref{eq:spatial-derivative-finite-dim} in the same way. Since $\M'' = \M$, \Autoref{def:spatial-derivative} applies verbatim with the roles of $\M$ and $\M'$ interchanged: for $φ' \in (\M')_*^+$ and $ψ \in \M_*^+$ we obtain $\Delta(φ'/ψ)$. We always use the convention that the first argument of $\Delta(\cdot/\cdot)$ lives on one of the two algebras $\M$, $\M'$ and the second argument on the other one; the primes on the arguments indicate which is which. In the definition of the entropic quantities in \Autoref{sec:prelim-divergences} the first argument lives on the commutant.

\paragraph{Non-faithful second argument.}
For a faithful second argument, \Autoref{def:spatial-derivative} is Connes' definition \cite{connes_spatial_1980}, see \cite[Sec.~11.1]{hiai_lectures_2021}; there $ψ$ may be an n.f.s.\ weight, which is the generality needed for Haagerup's identity in \Autoref{sec:operator_valued_weights}. In the duality relations the second argument is a positive normal functional that is not faithful in general. For such $ψ$ the spatial derivative is defined in the same way in \cite[Ch.~4]{ohya_quantum_1993}, see also \cite[Sec.~2]{petz_quasientropies_1985} and \cite[Sec.~2.2]{berta_renyi_2018}, as well as \cite[App.~C]{araki_positive_1982} and \cite[Sec.~10.1]{hiai_lectures_2021} for the relative modular operators of non-faithful functionals in the standard form. There is one difference: in these references the operator is extended by $0$ to $(\IdentityMatrix - s(ψ))\HS$, whereas we extend it by $+\infty$. In the subsystem picture our convention gives \eqref{eq:spatial-derivative-finite-dim} with $σ_C^{-1}$ the generalized inverse, i.e.\ $\Delta(φ/ψ) = +\infty$ on the whole of $\HS_A \otimes \ker{σ}_C$, including $\ker{ρ}_A \otimes \ker{σ}_C$ where one might have expected $0 \cdot \infty$. This is the convention under which the divergences with non-faithful second argument take the values familiar from finite dimensions, see \Autoref{rem:divergence-conventions}. Equivalently, the restriction of $\Delta(φ/ψ)$ to $s(ψ)\HS$ is Connes' spatial derivative with respect to the faithful restriction of $ψ$ to the reduced algebra $s(ψ)\M's(ψ)$; this description, by which statements about non-faithful functionals are reduced to the faithful case, is made precise in \Autoref{lem:bounded-vectors-nonfaithful} in \Cref{app:bounded-vectors-proof}.

\begin{proposition}[Properties of the spatial derivative]\label{prop:spatial-derivative-properties}
    Let $\M \subset \BO$ be a von Neumann algebra, $φ, φ_1, φ_2 \in \M_*^+$, $ψ, ψ_1, ψ_2 \in (\M')_*^+$, $λ, μ > 0$ and $a \in \M$.
    \begin{enumerate}[label=(\alph*)]
        \item\label{item:sd-support} \emph{Supports.} With respect to the orthogonal decomposition
        \begin{equation}\label{eq:support-decomposition}
            \HS = s(φ)s(ψ)\HS \;\oplus\; (\IdentityMatrix - s(φ))s(ψ)\HS \;\oplus\; (\IdentityMatrix - s(ψ))\HS\,,
        \end{equation}
        which reduces $\Delta(φ/ψ)$, the operator $\Delta(φ/ψ)$ is positive with trivial kernel on the first summand, $0$ on the second and $+\infty$ on the third. In particular, $\ip{ξ}{\Delta(φ/ψ)^{s} ξ}$ for $s < 0$ is infinite if $ξ$ has a nonzero component in $(\IdentityMatrix - s(φ))s(ψ)\HS$, and $\ip{ξ}{\Delta(φ/ψ)^{s} ξ}$ for $s > 0$ is infinite if $ξ$ has a nonzero component in $(\IdentityMatrix - s(ψ))\HS$.
        \item\label{item:sd-inverse} \emph{Inversion.} The generalized positive operators $\Delta(φ/ψ)^{-1}$ and $\Delta(ψ/φ)$ (the latter with the roles of $\M$ and $\M'$ interchanged) coincide on the subspace $\big(\IdentityMatrix - (\IdentityMatrix - s(φ))(\IdentityMatrix-s(ψ))\big)\HS$, which reduces both. On its orthogonal complement $(\IdentityMatrix - s(φ))(\IdentityMatrix-s(ψ))\HS$ they differ: there $\Delta(φ/ψ)$ and $\Delta(ψ/φ)$ are both $+\infty$, so that $\Delta(φ/ψ)^{-1}$ vanishes. In particular $\Delta(φ/ψ)^{-1} \leq \Delta(ψ/φ)$, and $μ_ξ^{\Delta(φ/ψ)^{-1}} = μ_ξ^{\Delta(ψ/φ)}$ for every $ξ \in s(φ)\HS \cup s(ψ)\HS$. If $φ$ or $ψ$ is faithful, then $\Delta(φ/ψ)^{-1} = \Delta(ψ/φ)$.
        \item\label{item:sd-linearity} \emph{Linearity and monotonicity.} $\Delta(λφ/μψ) = \frac{λ}{μ}\Delta(φ/ψ)$, $\Delta((φ_1 + φ_2)/ψ) = \Delta(φ_1/ψ) + \Delta(φ_2/ψ)$ (as a sum of forms), $φ_1 \leq{φ}_2$ implies $\Delta(φ_1/ψ) \leq \Delta(φ_2/ψ)$, $φ_i \nearrow{φ}$ implies $\Delta(φ_i/ψ) \nearrow \Delta(φ/ψ)$, and $ψ_1 \leq{ψ}_2$ implies $\Delta(φ/ψ_1) \geq \Delta(φ/ψ_2)$.
        \item\label{item:sd-covariance} \emph{Covariance.} The operators $\Delta(aφa^*/ψ)$ and $a\,\Delta(φ/ψ)\,a^*$ have the same restriction to $s(ψ)\HS$, where $(aφa^*)(x) \coloneqq{φ}(a^*xa)$ and $a\,\Delta(φ/ψ)\,a^*$ is the generalized positive operator with form $ξ \mapsto \ip{a^*ξ}{\Delta(φ/ψ)\, a^*ξ}$. If $ψ$ is faithful, the two operators coincide.
        \item\label{item:sd-standard-form} \emph{Standard form.} If $(\M, \HS, J, \mathcal{P})$ is a standard form, $σ \in \M_*^+$ is faithful and $σ' \coloneqq{ω}_{ξ_σ}|_{\M'}$ is the functional on $\M'$ with the same vector representative, then $\Delta(φ/σ') = \Delta_{φ, σ}$ is Araki's relative modular operator \cite{araki_relative_1977}, i.e.\ $\Delta_{φ,σ}^{1/2}$ is the positive part in the polar decomposition of the closure of $xξ_σ \mapsto x^*ξ_φ$, $x \in \M$ \cite[Prop.~11.6]{hiai_lectures_2021}.
    \end{enumerate}
    If $ψ$ is an n.f.s.\ weight on $\M'$, then \ref{item:sd-support}, \ref{item:sd-linearity} and \ref{item:sd-covariance} hold with $s(ψ) = \IdentityMatrix$, also for semifinite normal weights $φ$ on $\M$, and $\Delta(φ/ψ)^{-1} = \Delta(ψ/φ)$ if $φ$ is faithful.
\end{proposition}
\begin{proof}
    For a faithful second argument these are results of Connes \cite{connes_spatial_1980}, see \cite[Thms.~11.7 and~11.8, Props.~11.6, 11.17 and~11.19]{hiai_lectures_2021}. For general arguments, \ref{item:sd-support}, \ref{item:sd-inverse} and the additivity in \ref{item:sd-linearity} can be found in \cite[Eqs.~(4.4)--(4.7)]{ohya_quantum_1993}, with the convention that the spatial derivative vanishes off the support of its second argument; the translation into our convention and the reduction to the faithful case are given in \Cref{app:spatial-derivative-general-proof}.
\end{proof}

Item \ref{item:sd-standard-form} shows that in the standard form the spatial derivative reduces to the relative modular operator, so that all entropic quantities defined below through spatial derivatives agree with their definitions in terms of relative modular operators used in \cite{araki_relative_1977, ohya_quantum_1993, hiai_quantum_2021, fawzi_asymptotic_2023}. The advantage of the spatial derivative is that it is defined in every representation of $\M$, in particular in the representation $\M \subset \N' \subset \BO$ in which both $\M$ and the commutant $\N'$ of a subalgebra act on the same Hilbert space; this is essential for the duality relations.

\begin{remark}\label{rem:inversion-degenerate-subspace}
    Part \ref{item:sd-inverse} of \Autoref{prop:spatial-derivative-properties} says that the inversion formula $\Delta(φ/ψ)^{-1} = \Delta(ψ/φ)$ holds in general except on the subspace where both $φ$ and $ψ$ vanish. There the convention \enquote{$\Delta$ is $+\infty$ off the support of the second argument} of \Autoref{def:spatial-derivative} makes both $\Delta(φ/ψ)$ and $\Delta(ψ/φ)$ equal to $+\infty$, instead of assigning them a pair $(0, \infty)$ of mutually inverse values, so that the two sides of the inversion formula are $0$ and $+\infty$, respectively. For the entropic quantities this subspace never matters: they are always evaluated on a vector lying in the support of one of the two arguments (see \Autoref{lem:sandwiched-support-restriction} for the only case where this requires an argument), and then the inversion formula allows us to switch freely between $\Delta(φ/ψ)$ and $\Delta(ψ/φ)^{-1}$.
\end{remark}

\subsection{Entropic quantities}\label{sec:prelim-divergences}

\emph{Standing assumptions.} Throughout this subsection $\M \subset \BO$ is a von Neumann algebra, $ρ \in \S(\M)$ is a normal state with a fixed vector representative $\ket{\v{ρ}} \in \HS$, $ρ' = ρ|_{\M'} = ω_{\v{ρ}}|_{\M'} \in \S(\M')$ is the induced state on the commutant, and $σ \in \M_*^+$ is a positive normal functional on $\M$, which is assumed neither to be normalized nor to be faithful. All spatial derivatives below are of the form $\Delta(ω'/σ)$ with $ω' \in (\M')_*^+$ in the first and $σ$ in the second argument; they are the generalized positive operators of \Autoref{def:spatial-derivative} with the roles of $\M$ and $\M'$ interchanged. By \Autoref{prop:spatial-derivative-properties}\ref{item:sd-support}, $\Delta(ω'/σ)$ is $+\infty$ on $(\IdentityMatrix - s(σ))\HS$, it vanishes on $(\IdentityMatrix - s(ω'))s(σ)\HS$, and it is positive with trivial kernel on $s(ω')s(σ)\HS$. In the subsystem picture $\M = \BO[\HS_A] \otimes \IdentityMatrix_C$, $\M' = \IdentityMatrix_A \otimes \BO[\HS_C]$, with $σ \leftrightarrow{σ}_A$ and $ω' \leftrightarrow{ω}_C$, this operator is $\Delta(ω'/σ) = σ_A^{-1} \otimes{ω}_C$. All expressions of the form $\bra{\v{ρ}} f(\Delta) \ket{\v{ρ}}$ are spectral integrals as in \Autoref{sec:prelim-vn}, i.e.\ $\bra{\v{ρ}} f(\Delta) \ket{\v{ρ}} = \int_{[0,\infty]} f \, \dd{μ}^{\Delta}_{\v{ρ}}$, with the conventions $\log 0 = -\infty$, $\log \infty = +\infty$, $0^{-s} = \infty^{s} = \infty$ and $0^{s} = \infty^{-s} = 0$ for $s > 0$. In particular, the subspace $(\IdentityMatrix - s(σ))\HS$, on which every $\Delta(ω'/σ)$ is $+\infty$, contributes $f(\infty)\, ρ(\IdentityMatrix - s(σ))$ to $\bra{\v{ρ}} f(\Delta(ω'/σ))\ket{\v{ρ}}$, since $\norm{(\IdentityMatrix - s(σ))\v{ρ}}^2 = ρ(\IdentityMatrix - s(σ))$; this contribution vanishes if $s(ρ) \leq s(σ)$, and otherwise it is $f(\infty) \cdot (\text{positive number})$.

\begin{definition}[Relative entropy and \renyi divergences]\label{def:divergences}
    In the setting above we define
    \begin{align}
        D(ρ\|σ) &\coloneqq \bra{\v{ρ}} \log \Delta(ρ'/σ) \ket{\v{ρ}}\,, \label{eq:def-relative-entropy}\\
        D_α(ρ\|σ) &\coloneqq \frac{1}{α - 1} \log \bra{\v{ρ}} \Delta(ρ'/σ)^{α - 1} \ket{\v{ρ}}\,, \qquad{α} \in [0, 1) \cup (1, \infty)\,, \label{eq:def-petz-renyi}\\
        \widetilde{D}_α(ρ\|σ) &\coloneqq \sup_{ω' \in \S(\M')} \frac{1}{α'} \log \bra{\v{ρ}} \Delta(ω'/σ)^{α'} \ket{\v{ρ}}\,, \qquad{α}' \coloneqq \frac{α - 1}{α}\,, \quad{α} \in [\tfrac12, 1) \cup (1, \infty]\,, \label{eq:def-sandwiched-renyi}
    \end{align}
    called the (Umegaki--Araki) \emph{relative entropy}, the \emph{Petz--\renyi divergence} and the \emph{sandwiched \renyi divergence} of order $α$, respectively. Here $α' \coloneqq 1$ for $α = \infty$, so that the two extreme orders $α = \tfrac12$ and $α = \infty$ correspond to $α' = -1$ and $α' = 1$; $\widetilde{D}_\infty$ is the max-relative entropy $D_{\max}$ (\Autoref{lem:dmax-variational}). Since $α' > 0$ for $α > 1$ and $α' < 0$ for $α < 1$, the supremum in \eqref{eq:def-sandwiched-renyi} reads explicitly
    \begin{equation}\label{eq:def-sandwiched-renyi-explicit}
        \widetilde{D}_α(ρ\|σ) = \begin{cases} \dfrac{1}{α'} \log\, \sup_{ω' \in \S(\M')} \bra{\v{ρ}} \Delta(ω'/σ)^{α'} \ket{\v{ρ}} & \text{for } α > 1\,,\\[2ex] \dfrac{1}{α'} \log\, \inf_{ω' \in \S(\M')} \bra{\v{ρ}} \Delta(ω'/σ)^{α'} \ket{\v{ρ}} & \text{for } α < 1\,. \end{cases}
    \end{equation}
    By \Autoref{prop:representation-independence} below, these quantities do not depend on the choice of the representation $\M \subset \BO$ and of the vector representative $\ket{\v{ρ}}$, and \Autoref{lem:relative-entropy-well-defined} shows that $D(ρ\|σ)$ is a well-defined element of $[-\log{σ}(\IdentityMatrix), +\infty]$, while the other two quantities take values in $[-\infty, +\infty]$ by definition. We also set $D_1 \coloneqq \widetilde{D}_1 \coloneqq D$, which is justified by \Autoref{prop:divergence-properties}\ref{item:dp-limits}.
\end{definition}

\begin{remark}[Finite dimensions, and the role of the supports]\label{rem:divergence-conventions}
    Let $\M = \BO[\HS_A] \otimes \IdentityMatrix_C$ act on $\HS_A \otimes \HS_C$, and let $\ket{\v{ρ}}$ be a purification of the density matrix $ρ$, so that $ρ' \leftrightarrow{ρ}_C$ is the reduced state on $C$. Then $\Delta(ρ'/σ)^{s} = σ^{-s} \otimes{ρ}_C^{s}$ and, for the canonical purification $\ket{\v{ρ}} = \sum_i (ρ^{1/2} \otimes \IdentityMatrix)\ket{i}\otimes\ket{\bar\imath}$, one has $\bra{\v{ρ}} X \otimes Y \ket{\v{ρ}} = \Tr(ρ^{1/2} X ρ^{1/2} Y^{T})$. A short computation, using Hölder's inequality for the optimization over $ω_C$ (see \cite[Sec.~3.1]{berta_renyi_2018} and \cite[Lemma~4]{muller-lennert_quantum_2013}), then recovers the familiar expressions
    \begin{align}\label{eq:finite-dim-divergences}
        D(ρ\|σ) &= \Tr\big(ρ(\log ρ - \log σ)\big),\\
        D_α(ρ\|σ) &= \frac{\log \Tr(ρ^α σ^{1-α})}{α - 1} \\
        \widetilde{D}_α(ρ\|σ) &= \frac{\log \Tr \big(σ^{\frac{1-α}{2α}} ρ\, σ^{\frac{1-α}{2α}}\big)^α}{α - 1}
    \end{align}
    of \cite{umegaki_conditional_1962, petz_quasientropies_1986, muller-lennert_quantum_2013, wilde_strong_2014}, including the usual conventions for non-faithful $σ$: negative powers of $σ$ are taken on the support of $σ$ (this is $\infty^{-s} = 0$), and the expressions with a positive power of $σ^{-1}$ are $+\infty$ unless $s(ρ) \leq s(σ)$ (this is $\infty^{s} = \infty$). The general picture is the same: by the standing assumptions,
    \begin{itemize}
        \item $D(ρ\|σ)$, $D_α(ρ\|σ)$ for $α > 1$ and $\widetilde{D}_α(ρ\|σ)$ for $α \in (1, \infty]$ are $+\infty$ unless $s(ρ) \leq s(σ)$ (the $+\infty$ part of $\Delta$ contributes);
        \item $D_α(ρ\|σ)$ for $α \in [0,1)$ and $\widetilde{D}_α(ρ\|σ)$ for $α \in [\tfrac12, 1)$ do not see $(\IdentityMatrix - s(σ))\HS$ at all; they are finite unless $ρ$ and $σ$ are orthogonal, i.e.\ unless $ρ(s(σ)) = 0$, in which case they are $+\infty$ (see \Autoref{lem:first-moment} and \Autoref{lem:sandwiched-support-restriction}).
    \end{itemize}
    For $α < 1$ there is a second support condition, on the optimization variable $ω'$ in \eqref{eq:def-sandwiched-renyi}: since $\Delta(ω'/σ)$ vanishes on $(\IdentityMatrix - s(ω'))s(σ)\HS$ and $0^{α'} = \infty$, the value $\bra{\v{ρ}}\Delta(ω'/σ)^{α'}\ket{\v{ρ}}$ is $+\infty$ unless $s(σ)\v{ρ} \in s(ω')\HS$. Such $ω'$ are irrelevant for the infimum, and in fact one may restrict the infimum to those $ω'$ whose support dominates $s(ρ') = [\M\v{ρ}]$, see \Autoref{lem:sandwiched-support-restriction}. This is the support restriction imposed in the definitions of \cite[Def.~2]{berta_renyi_2018} and \cite[Sec.~2]{jencova_renyi_2021}, where for $α < 1$ the infimum runs only over states $ω'$ with $s(ω') \geq s(ρ')$ (for vector functionals $ω' = ω_{ζ}|_{\M'}$, this is the condition $ω_ζ|_{\M'} \gg{ω}_{\v{ρ}}|_{\M'}$ of \cite{berta_renyi_2018}). With our conventions for the generalized positive operator $\Delta(ω'/σ)$ no such restriction is necessary; for $α > 1$ there was none to begin with, as states $ω'$ with small support simply give small values. Note also that the restriction on the optimization is what makes the two ways of writing the sandwiched divergence, with $\Delta(ω'/σ)^{α'}$ or with $\Delta(σ/ω')^{-α'}$, agree: for $ω'$ with $s(ω') \geq s(ρ')$ the vector $\v{ρ}$ lies in $s(ω')\HS$, so the two spectral measures coincide by \Autoref{prop:spatial-derivative-properties}\ref{item:sd-inverse}.
\end{remark}

\begin{lemma}[Well-definedness of the relative entropy]\label{lem:relative-entropy-well-defined}
    The relative entropy $D(ρ\|σ)$ of \Autoref{def:divergences} is well defined with values in $(-\infty, +\infty]$, and $D(ρ\|σ) \geq -\log{σ}(s(ρ)) \geq -\log{σ}(\IdentityMatrix)$.
\end{lemma}
\begin{proof}
    This is the classical argument showing that Araki's relative entropy is well defined and bounded from below, see e.g.\ \cite{araki_relative_1977} and \cite[Eqs.~(5.1)--(5.2), Prop.~5.1]{ohya_quantum_1993}; how the supports enter under our conventions is explained in \Cref{app:divergence-proofs} (\Autoref{lem:first-moment}).
\end{proof}

\begin{lemma}[Support restriction in the sandwiched divergence]\label{lem:sandwiched-support-restriction}
    Let $α \in [\tfrac12, 1)$. For $ω' \in \S(\M')$ with $s(σ)\v{ρ} \notin s(ω')\HS$ the value $\bra{\v{ρ}}\Delta(ω'/σ)^{α'}\ket{\v{ρ}}$ is $+\infty$, and
    \begin{equation}\label{eq:sandwiched-support-restriction}
        \inf_{ω' \in \S(\M')} \bra{\v{ρ}} \Delta(ω'/σ)^{α'} \ket{\v{ρ}} = \inf_{\substack{ω' \in \S(\M') \\ s(ω') \geq s(ρ')}} \bra{\v{ρ}} \Delta(ω'/σ)^{α'} \ket{\v{ρ}} = \inf_{\substack{ω' \in \S(\M') \\ s(ω') \geq s(ρ')}} \bra{\v{ρ}} \Delta(σ/ω')^{-α'} \ket{\v{ρ}}\,.
    \end{equation}
    Moreover, if $ρ(s(σ)) = 0$ then all three infima vanish, i.e.\ $\widetilde{D}_α(ρ\|σ) = +\infty$.
\end{lemma}
\begin{proof}
    The first statement was already observed in \Autoref{rem:divergence-conventions}, and the restriction of the infimum is obtained by mixing $ω'$ with a small multiple of $ρ'$; the details are given in \Cref{app:divergence-proofs}.
\end{proof}

\begin{proposition}[Independence of the representation]\label{prop:representation-independence}
    Let $\tilde{\HS}$ be a Hilbert space carrying a normal faithful unital $*$-representation $\tilde{π}: \M \to \BO[\tilde\HS]$, let $\ket{\tilde{\v{ρ}}} \in \tilde\HS$ be a vector representative of $ρ$ in it, i.e.\ $\bra{\tilde{\v{ρ}}} \tilde{π}(x) \ket{\tilde{\v{ρ}}} = ρ(x)$ for $x \in \M$, and write $\tilde{ρ}' \coloneqq{ω}_{\tilde{\v{ρ}}}|_{\tilde{π}(\M)'}$ and $\tilde{σ} \coloneqq{σ} \circ \tilde{π}^{-1} \in \tilde{π}(\M)_*^+$. Then:
    \begin{enumerate}[label=(\roman*)]
        \item\label{item:ri-unoptimized} the spectral measure of $\v{ρ}$ with respect to $\Delta(ρ'/σ)$ coincides with the spectral measure of $\tilde{\v{ρ}}$ with respect to $\Delta(\tilde{ρ}'/\tilde{σ})$;
        \item\label{item:ri-sandwiched} for every $α \in [\tfrac12, 1) \cup (1, \infty]$,
        \begin{equation}
            \sup_{ω' \in \S(\M')} \bra{\v{ρ}} \Delta(ω'/σ)^{α'} \ket{\v{ρ}} = \sup_{\tilde{ω}' \in \S(\tilde{π}(\M)')} \bra{\tilde{\v{ρ}}} \Delta(\tilde{ω}'/\tilde{σ})^{α'} \ket{\tilde{\v{ρ}}}\,,
        \end{equation}
        and the same holds with $\sup$ replaced by $\inf$ on both sides.
    \end{enumerate}
    Consequently $D(ρ\|σ)$, $D_α(ρ\|σ)$ and $\widetilde{D}_α(ρ\|σ)$ depend neither on the representation of $\M$ nor on the vector representative of $ρ$. Moreover, in a representation in which every normal state of $\M'$ is a vector state (for instance the standard form), the optimization in \eqref{eq:def-sandwiched-renyi} can equivalently be performed over unit vectors $ζ \in \HS$, with $ω' = ω_ζ|_{\M'}$; hence $\widetilde{D}_α$ coincides with the Araki--Masuda divergence of \cite[Def.~7]{berta_renyi_2018} (see also \cite{araki_positive_1982}) and with the sandwiched \renyi divergence of \cite{jencova_renyi_2018, jencova_renyi_2021}.
\end{proposition}
\begin{proof}
The proof, which adapts the argument of \cite[Lemma~3]{berta_renyi_2018} to our conventions, is given in \Cref{app:divergence-proofs}.
\end{proof}

We collect the properties of the three divergences that are used in this paper. Since all quantities are independent of the representation, results proved in the literature for the standard form apply verbatim.

\begin{proposition}[Basic properties]\label{prop:divergence-properties}
    Let $ρ \in \S(\M)$ and $σ, σ_1, σ_2 \in \M_*^+$.
    \begin{enumerate}[label=(\roman*)]
        \item\label{item:dp-scaling} \emph{Scaling.} For $λ > 0$ and $\D \in \set{D, D_α, \widetilde{D}_α}$, $\D(ρ\|λσ) = \D(ρ\|σ) - \log{λ}$.
        \item\label{item:dp-monotone} \emph{Anti-monotonicity in the second argument.} If $σ_1 \leq{σ}_2$, then $\D(ρ\|σ_1) \geq \D(ρ\|σ_2)$ for $\D = D$, for $\D = D_α$ with $α \in [0, 2]$, and for $\D = \widetilde{D}_α$ with $α \in [\tfrac12, \infty]$.
        \item\label{item:dp-dpi} \emph{Data processing.} Let $\N$ be a von Neumann algebra and $\Phi: \N \to \M$ a normal unital completely positive map, with predual $\Phi_*: \Ms \to \Ns$, $\Phi_*(ω) = ω \circ \Phi$. Then $\D(\Phi_*(ρ)\|\Phi_*(σ)) \leq \D(ρ\|σ)$ for $\D = D$ \cite{uhlmann_relative_1977} (see \cite[Thm.~5.3]{ohya_quantum_1993}), for $\D = D_α$ with $α \in [0,2]$ \cite{petz_quasientropies_1985, hiai_quantum_2018}, and for $\D = \widetilde{D}_α$ with $α \in [\tfrac12, \infty)$ \cite[Thm.~14]{berta_renyi_2018}, \cite[Thm.~4.1]{jencova_renyi_2021}; for $\widetilde{D}_\infty = D_{\max}$ it is immediate from \Autoref{def:one-shot}\ref{item:os-dmax}, since $ρ \leq{λ}σ$ implies $\Phi_*(ρ) \leq{λ}\Phi_*(σ)$.
        \item\label{item:dp-limits} \emph{Dependence on $α$.} The map $α \mapsto D_α(ρ\|σ)$ is nondecreasing on $[0, 2]$, and $α \mapsto \widetilde{D}_α(ρ\|σ)$ is nondecreasing on $[\tfrac12, 1) \cup (1, \infty]$ \cite[Prop.~11, Lemma~8]{berta_renyi_2018}. One has $\lim_{α \nearrow 1} D_α(ρ\|σ) = \lim_{α \nearrow 1} \widetilde{D}_α(ρ\|σ) = D(ρ\|σ)$. The limits from above require a finiteness condition: $\lim_{α \searrow 1} D_α(ρ\|σ) = \lim_{α \searrow 1} \widetilde{D}_α(ρ\|σ) = D(ρ\|σ)$ provided $D_α(ρ\|σ) < \infty$ for some $α \in (1, 2)$ or, equivalently, $\widetilde{D}_α(ρ\|σ) < \infty$ for some $α > 1$ \cite[Prop.~11, Thm.~13]{berta_renyi_2018}. In infinite dimensions the support condition $s(ρ) \leq s(σ)$ is not sufficient for this, since $D(ρ\|σ)$ can be finite while $D_α(ρ\|σ) = \widetilde{D}_α(ρ\|σ) = \infty$ for all $α > 1$. Moreover, for $α \in [\tfrac12, 1) \cup (1, \infty)$,
        \begin{equation}\label{eq:alt-inequalities}
            D_{2 - 1/α}(ρ\|σ) \leq \widetilde{D}_α(ρ\|σ) \leq D_α(ρ\|σ)\,,
        \end{equation}
        the second inequality being the Araki--Lieb--Thirring inequality \cite[Thm.~12]{berta_renyi_2018} and the first one due to \cite[Thm.~3.3]{jencova_renyi_2021}. In particular, for $α \in [\tfrac12, 1)$, $\widetilde{D}_α(ρ\|σ) < \infty$ iff $D_α(ρ\|σ) < \infty$ iff $ρ(s(σ)) > 0$, by \Autoref{lem:first-moment}.
        \item\label{item:dp-extreme} \emph{Extreme orders.} $D_0(ρ\|σ) = -\log{σ}(s(ρ))$ (\Autoref{lem:first-moment}), $\widetilde{D}_{1/2}(ρ\|σ) = -\log F(ρ, σ)$ with the fidelity $F$ of \Autoref{def:one-shot}\ref{item:os-fidelity}, and $\widetilde{D}_\infty(ρ\|σ) = D_{\max}(ρ\|σ) = \lim_{α \to \infty} \widetilde{D}_α(ρ\|σ)$ with $D_{\max}$ as in \Autoref{def:one-shot}\ref{item:os-dmax} (\Autoref{lem:dmax-variational} and \cite[Lemma~9]{berta_renyi_2018}).
    \end{enumerate}
\end{proposition}
\begin{proof}
    Items \ref{item:dp-scaling} and \ref{item:dp-monotone} follow from the corresponding properties of the spatial derivative (\Autoref{prop:spatial-derivative-properties}\ref{item:sd-linearity}), and the remaining items are taken from the references given in the statement; the details are given in \Cref{app:divergence-proofs}.
\end{proof}

Finally we define the one-shot quantities entering \Autoref{thm:gqsl_subalgebras}; in von Neumann algebras these were used e.g.\ in \cite{berta_renyi_2018, fawzi_asymptotic_2023}, to which we refer for further properties. In particular, all quantities in the following definition satisfy the data processing inequality of \Autoref{prop:divergence-properties}\ref{item:dp-dpi}, either directly from their definition or by \cite[Sec.~2.4]{fawzi_asymptotic_2023}.

\begin{definition}[Max-relative entropy, fidelity, hypothesis testing and smoothing]\label{def:one-shot}
    Let $ρ \in \S(\M)$, $σ \in \M_*^+$ and $ε \in [0, 1)$.
    \begin{enumerate}[label=(\alph*)]
        \item\label{item:os-dmax} The \emph{max-relative entropy} \cite{datta_min_2009} is $D_{\max}(ρ\|σ) \coloneqq \log \inf\Set{λ > 0 | ρ \leq{λ} σ}$, with $\inf \emptyset = +\infty$. It coincides with the sandwiched \renyi divergence $\widetilde{D}_\infty$ of order $\infty$ from \Autoref{def:divergences}, see \Autoref{lem:dmax-variational} below.
        \item\label{item:os-fidelity} The \emph{fidelity} (Uhlmann's transition probability \cite{uhlmann_transition_1976}) is
        \begin{equation}
            F(ρ, σ) \coloneqq \exp\big(-\widetilde{D}_{1/2}(ρ\|σ)\big) = \inf_{ω' \in \S(\M')} \bra{\v{ρ}} \Delta(ω'/σ)^{-1} \ket{\v{ρ}} \in [0, σ(\IdentityMatrix)]\,;
        \end{equation}
        by \cite[Lemma~9]{berta_renyi_2018} this agrees with Uhlmann's fidelity in the form of Alberti, $F(ρ,σ) = \sup_U \abs{\ip{\v{ρ}}{U ξ_σ}}^2$. Here $ξ_σ \in \HS'$ is a vector representative of $σ$ in a normal representation $π': \M \to \BO[\HS']$, and the supremum runs over all contractions $U: \HS' \to \HS$ that intertwine the two representations, i.e.\ $Uπ'(x) = xU$ for all $x \in \M$. If $\HS$ itself contains a vector representative $ξ_σ$ of $σ$, these are the contractions in $\M'$, so that $F(ρ,σ) = \sup\Set{\abs{\ip{\v{ρ}}{u' ξ_σ}}^2 | u' \in \M',\ \norm{u'} \leq 1}$. Moreover, $F$ has the variational expression $F(ρ,σ) = \inf\Set{ρ(x)\,σ(x^{-1}) | x \in \M \text{ positive and invertible}}$ \cite{alberti_stochastic_1983}. In finite dimensions $F(ρ,σ) = \norm{ρ^{1/2}σ^{1/2}}_1^2$. For states $ρ, σ \in \S(\M)$ the \emph{purified distance} \cite{gilchrist_distance_2005, tomamichel_quantum_2016} is $P(ρ, σ) \coloneqq \sqrt{1 - F(ρ, σ)}$; it is a metric on $\S(\M)$ (see \cite[Sec.~2.2]{fawzi_asymptotic_2023}).
        \item\label{item:os-hypothesis-testing} The \emph{hypothesis testing relative entropy} \cite{buscemi_quantum_2010, wang_oneshot_2012} is
        \begin{equation}
            D_H^ε(ρ\|σ) \coloneqq - \log \inf\Set{σ(T) | T \in \M,\ 0 \leq T \leq \IdentityMatrix,\ ρ(T) \geq 1 - ε}\,.
        \end{equation}
        \item\label{item:os-smoothing} The \emph{smoothed} max- and min-relative entropies are
        \begin{equation}
            D_{\max}^ε(ρ\|σ) \coloneqq \inf_{\substack{\tilde{ρ} \in \S(\M) \\ P(\tilde{ρ}, ρ) \leq{ε}}} D_{\max}(\tilde{ρ}\|σ)\,, \qquad
            D_{\min}^ε(ρ\|σ) \coloneqq \sup_{\substack{\tilde{ρ} \in \S(\M) \\ P(\tilde{ρ}, ρ) \leq{ε}}} \widetilde{D}_{1/2}(\tilde{ρ}\|σ)\,.
        \end{equation}
        In particular $D_{\min}^0 = \widetilde{D}_{1/2} = -\log F$ and $D_{\max}^0 = D_{\max}$.\footnote{In \cite{fawzi_asymptotic_2023} (and in much of the finite-dimensional literature, e.g.\ \cite{tomamichel_quantum_2016}) the smoothing is performed over sub-normalized states with respect to the purified distance defined through the generalized fidelity. We smooth over normalized states throughout.}
        \item\label{item:os-sets} For a set $S \subset \M_*^+$ and any of the unsmoothed quantities $\D \in \set{D, D_α, \widetilde{D}_α, D_{\max}, D_H^ε}$ we write $\D(ρ\|S) \coloneqq \inf_{σ \in S} \D(ρ\|σ)$, and for the smoothed ones
        \begin{align}
            D_{\max}^ε(ρ\|S) &\coloneqq \inf_{\substack{\tilde{ρ} \in \S(\M) \\ P(\tilde{ρ}, ρ) \leq{ε}}} \inf_{σ \in S} D_{\max}(\tilde{ρ}\|σ) = \inf_{σ \in S} D_{\max}^ε(ρ\|σ)\,, \\
            D_{\min}^ε(ρ\|S) &\coloneqq \sup_{\substack{\tilde{ρ} \in \S(\M) \\ P(\tilde{ρ}, ρ) \leq{ε}}} \inf_{σ \in S} \widetilde{D}_{1/2}(\tilde{ρ}\|σ)\,,
        \end{align}
        where in the second definition the order of the optimizations matters. Given a subalgebra $\N \subset \M$ with a map $E_*: \Ns \to \Ms$ as in \autoref{def:conditional-expectation}, we abbreviate $\D(ρ\|E_*(\Ns)) \coloneqq \D(ρ\|\Set{E_*(σ) | σ \in \S(\N)})$, which is the notation used in \Autoref{sec:intro-mainthms}.
    \end{enumerate}
\end{definition}

The max-relative entropy is defined through the order relation between $ρ$ and $σ$, without any reference to a representation. For the duality relations we need to know that it is nevertheless given by the same kind of expression as the sandwiched \renyi divergences of finite order, namely by \eqref{eq:def-sandwiched-renyi} with $α' = 1$:

\begin{lemma}[The max-relative entropy as the sandwiched divergence of order $\infty$]\label{lem:dmax-variational}
    In the setting of \Autoref{def:divergences}, let $λ > 0$. Then $ρ \leq{λ}σ$ if and only if $\bra{\v{ρ}} \Delta(ω'/σ) \ket{\v{ρ}} \leq{λ}$ for all $ω' \in \S(\M')$. Consequently,
    \begin{equation}\label{eq:dmax-variational}
        D_{\max}(ρ\|σ) = \log \sup_{ω' \in \S(\M')} \bra{\v{ρ}} \Delta(ω'/σ) \ket{\v{ρ}} = \widetilde{D}_\infty(ρ\|σ)\,.
    \end{equation}
\end{lemma}
\begin{proof}
    For the optimization over vector functionals this is \cite[Lemma~9]{berta_renyi_2018}; the adaptation to our conventions is given in \Cref{app:divergence-proofs}.
\end{proof}

In the finite-dimensional setting of \Autoref{rem:divergence-conventions}, and for $s(ρ) \leq s(σ)$, the supremum in \eqref{eq:dmax-variational} is
\begin{equation}
    \sup_{ω_C} \Tr\big(ρ^{1/2} σ^{-1} ρ^{1/2}\, ω_C^{T}\big) = \norm{ρ^{1/2} σ^{-1} ρ^{1/2}}_\infty = \norm{σ^{-1/2} ρ\, σ^{-1/2}}_\infty\,,
\end{equation}
which is the familiar expression $D_{\max}(ρ\|σ) = \log \norm{σ^{-1/2} ρ\, σ^{-1/2}}_\infty$ of \cite{datta_min_2009}; for $s(ρ) \not\leq s(σ)$ both sides of \eqref{eq:dmax-variational} are $+\infty$.

Relations between $D_H^ε$ and $D_{\max}^ε$ in von Neumann algebras, of the kind needed to pass from one to the other in \Autoref{thm:gqsl_subalgebras}, are proved in \cite[Lemma~3.5]{fawzi_asymptotic_2023}; the relation between $D_{\min}^ε$ and $D_{\max}^ε$ that we need is \Autoref{lem:smoothed_min_max_trivial_relation} below.

\subsection{Operator-valued Weights} \label{sec:operator_valued_weights}
Operator-valued weights can be thought of as (potentially) unbounded conditional expectations, in particular every conditional expectation is an operator-valued weight. This unboundedness is formalized by operator-valued weights taking values in what is called the extended positive part of $\M$. This extended positive part of $\M$, written as $\widehat{\M_+}$, can be thought of as the set of weights on $\Ms$, i.e. it is the set of all maps $x: \M_*^+ \to [0, \infty]$ that are linear and lower semi-continuous. This set inherits many natural properties from $\M$, for details see e.g.\ \cite{hiai_lectures_2021}. An operator-valued weight is then a linear map $T: \M_+ \to \widehat{\N_+}$ that also satisfies $T(x^* A x) = x^* T(A)x$ for all $x \in \N$. Note that unlike for conditional expectations we do not require that $T(\IdentityMatrix) = \IdentityMatrix$. We call $T$ normal if $x_α \nearrow x \in \M_+$ implies $T(x_α) \nearrow T(x)$, and semifinite if the subset of $x \in \M$ such that $T(x^*x) \in \N_+$ (rather than $ \in \widehat{\N_+}$) is weakly dense in $\M$. In particular then, if $T$ is normal and semi-finite and $\phi$ is a normal semi-finite weight on $\N$, then $\phi$ has a natural extension to $\widehat{\N_+}$ and $T_*(\phi) = \phi \circ T$ is a normal semi-finite weight on $\M$. In particular also, every normal conditional expectation is a normal semifinite operator valued weight. We write $P(\M)$ for the set of normal semi-finite weights on $\M$, and $P(\M, \N)$ for the set of all normal semi-finite operator valued weights $\M_+ \to \widehat{\N_+}$.

Let $\N \subset \M \subset \BO$ be a subalgebra with normal faithful semi-finite operator valued weight $E: \M_+ \to \widehat{\N_+}$ (for example $E$ could be a conditional expectation $\M \to \N$). Let $\N'$ and $\M'$ be the commutants within $\BO$ of $\N$ and $\M$ respectively. Then, by \cite[Remark 11.26]{hiai_lectures_2021} (see also \cite[Thms. 5.9,
6.13]{haagerup_operator_1979a}, \cite{kosaki_extension_1986}) there exists a normal faithful semifinite operator-valued weight $E^{-1}: \N'_+ \to \widehat{\M'_+}$ which satisfies that for every n.f.s. weights $ρ$ on $\N$ and $σ$ on $\M'$
\begin{equation}\label{eq:inverse_weight}
    \Delta(E_*(ρ)/ σ) = \Delta(ρ / E^{-1}_*(σ))\,.
\end{equation}
The existence of $E^{-1}$ is due to Haagerup \cite[Thms.~5.9, 6.13]{haagerup_operator_1979a}, who obtained a bijection between the normal faithful semifinite operator-valued weights from $\M$ to $\N$ and those from $\N'$ to $\M'$ and formulated it in terms of modular automorphism groups and cocycle derivatives. The formulation \eqref{eq:inverse_weight} in terms of Connes' spatial derivative \cite{connes_spatial_1980}, by which $E^{-1}$ is uniquely determined, is the one given in \cite[Rem.~11.26]{hiai_lectures_2021}. We refer to \eqref{eq:inverse_weight} as \emph{Haagerup's identity} and to $E^{-1}$ as Haagerup's dual operator-valued weight.
Note that $E^{-1}$ is usually not a conditional expectation, even if $E$ is a conditional expectation.

\paragraph{Why positive functionals suffice: finite index.}
In the duality relations of \Autoref{thm:duality_relation} the second arguments are of the form $E_*(σ)$ and $E^{-1}_*(ω')$, where $E: \M \to \N$ is a conditional expectation from $\M$ onto a von Neumann subalgebra $\N$ (see \Autoref{def:conditional-expectation}), $E_*(σ) \coloneqq{σ} \circ E$ for $σ \in \N_*^+$, $E^{-1}: \N'_+ \to \widehat{\M'_+}$ is Haagerup's dual operator-valued weight \cite[Thms.~5.9, 6.13]{haagerup_operator_1979a} (see \cite[Rem.~11.26]{hiai_lectures_2021}), and $E^{-1}_*(ω') \coloneqq{ω}' \circ E^{-1}$ for $ω' \in (\M')_*^+$. A priori $E^{-1}_*(ω')$ is only a normal weight on $\N'$, and our definitions do not cover this case. Following Kosaki \cite{kosaki_extension_1986} and, for non-factors, \cite{baillet_indice_1988} (see also \cite{longo_index_1989}), the \emph{index} of $E$ is
\begin{equation}\label{eq:index-definition}
    \Ind E \coloneqq E^{-1}(\IdentityMatrix) \in \widehat{\M'_+}\,;
\end{equation}
since $E^{-1}(u' x u'^*) = u' E^{-1}(x) u'^*$ for unitaries $u' \in \M'$, the element $\Ind E$ is invariant under conjugation by $\M'$ and hence belongs to the extended positive part of the center $Z(\M) = \M \cap \M'$. We say that $E$ has \emph{finite index} if $\Ind E$ is a bounded operator, i.e.\ $\Ind E \in Z(\M)_+$, and that the inclusion $\N \subset \M$ has finite index if it admits a normal faithful conditional expectation of finite index. If $\M$ is a factor, then $\Ind E \in [1, \infty]$ is a scalar, whether or not $\N$ is a factor. For a conditional expectation $E$ between factors it is Kosaki's index, which extends the Jones index $[\M:\N]$ \cite{jones_index_1983} of a subfactor of a type II$_1$ factor, to which it reduces for the trace-preserving expectation \cite{kosaki_extension_1986, pimsner_entropy_1986}, and which satisfies the Pimsner--Popa inequality $E(x) \geq (\Ind E)^{-1} x$ for $x \in \M_+$ \cite{pimsner_entropy_1986, kosaki_extension_1986, longo_index_1989}. In the subsystem case $\N = \IdentityMatrix_A \otimes \BO[\HS_B] \subset \M = \BO[\HS_A \otimes \HS_B]$ with $E = \frac{1}{d_A}\Tr_A \otimes \id$ one finds $\Ind E = d_A^2$, which is the index correction $2\log d_A$ in \eqref{eq:conditional entropy with dim.}.\footnote{Note that for reducible inclusions the optimal constant in the Pimsner--Popa inequality can be strictly larger than $(\Ind E)^{-1}$; for instance, in the subsystem case with $d_B < d_A$ the optimal constant is $(d_A \min\set{d_A, d_B})^{-1} = (d_A d_B)^{-1}$, cf.\ \cite[Sec.~4]{gao_generalized_2024}, whereas $\Ind E = d_A^2$. Throughout this paper \enquote{index} refers to \eqref{eq:index-definition}.} We do not use the notation $[\M:\N]$ in this paper, since we do not restrict to factors: in general the index depends on $E$ and not only on the inclusion, and if $\M$ has a nontrivial center it is an element of $Z(\M)$ that need not be a scalar (\Autoref{rem:non-scalar-index}). Likewise, finite index is a property of $E$: if the center of $\M$ is infinite-dimensional, the same inclusion can admit normal faithful conditional expectations of finite and of infinite index, for instance an infinite direct sum of finite-dimensional inclusions with conditional expectations whose indices are unbounded. In finite dimensions every faithful conditional expectation has finite index. The following lemma shows that in the finite-index situation all second arguments appearing in \Autoref{thm:duality_relation} are positive normal functionals, so that \Autoref{def:divergences} applies.

\begin{lemma}[Finite index makes the second arguments bounded]\label{lem:finite_index_bounded_inverse_weight}
    Let $\N \subset \M \subset \BO$ be von Neumann algebras and $E: \M_+ \to \widehat{\N_+}$ a normal faithful semifinite operator-valued weight with dual $E^{-1}: \N'_+ \to \widehat{\M'_+}$.
    \begin{enumerate}[label=(\roman*)]
        \item\label{item:fi-forward} If $E(\IdentityMatrix) \in \N_+$ is bounded (for instance if $E$ is a conditional expectation), then $E$ extends uniquely to a bounded, positive and normal linear map $E: \M \to \N$ with $\norm{E} = \norm{E(\IdentityMatrix)}$, and $E_*(σ) = σ \circ E \in \M_*^+$ for every $σ \in \N_*^+$, with $E_*(σ)(\IdentityMatrix) = σ(E(\IdentityMatrix))$.
        \item\label{item:fi-dual} If $E$ has finite index, then $E^{-1}$ extends uniquely to a bounded, positive and normal linear map $E^{-1}: \N' \to \M'$ with $\norm{E^{-1}} = \norm{\Ind E}$, and $E^{-1}_*(ω') = ω' \circ E^{-1} \in (\N')_*^+$ for every $ω' \in (\M')_*^+$, with $E^{-1}_*(ω')(\IdentityMatrix) = ω'(\Ind E) \leq \norm{\Ind E}\, ω'(\IdentityMatrix)$.
    \end{enumerate}
    In particular, if $E$ is a normal faithful conditional expectation of finite index, then for every $ρ \in \S(\M)$ with vector representative $\ket{\v{ρ}}$ both $E_*(ρ|_{\N}) \in \M_*^+$ and $E^{-1}_*(ρ|_{\M'}) \in (\N')_*^+$ are positive normal functionals.
\end{lemma}
\begin{proof}
    A positive map is bounded by its value at the identity, see e.g.\ \cite[Cor.~2.9]{paulsen_completely_2003}; the details for operator-valued weights are given in \Cref{app:finite-index-proof}.
\end{proof}

\subsection{Conditional expectations and their inverse weights in finite dimensions}\label{sec:finite-dim-cond-exp}

Before turning to the duality relations, we make the objects of \Autoref{sec:operator_valued_weights} explicit in the simplest situation covered by \Autoref{thm:finite_dim_duality}: $\HS$ is finite-dimensional, $\N \subset \M = \BO$ is an arbitrary unital $*$-subalgebra, and so
\begin{equation}\label{eq:finite-dim-setting}
	\N \subset \M \coloneqq \BO \otimes \IdentityMatrix_{\HS} \subset \BO[\HS \otimes \HS]\,, \qquad \M' = \IdentityMatrix_{\HS} \otimes \BO\,,
\end{equation}
where the second tensor factor is the purifying system: every state of $\M$ has a vector representative in $\HS \otimes \HS$ (cf.\ the standard form of $\BO$ in \Autoref{sec:prelim-vn}). We identify $\M$ with $\BO$ whenever the second factor plays no role.

\paragraph{The structure of $\N$ and of its commutant.}
When considering $\N$ as a subset of $\BO$ we get by \eqref{eq:finite-dim-structure}, after a unitary change of basis,
\begin{equation}\label{eq:finite-dim-inclusion}
	\HS = \bigoplus_{j \in J} \mathcal{L}_j \otimes \mathcal{V}_j\,, \qquad \N = \bigoplus_{j} \BO[\mathcal{L}_j] \otimes \IdentityMatrix_{\mathcal{V}_j}\,.
\end{equation}
    Considering everything inside the \enquote{representation space} $\BO[\HS \otimes \HS]$ we get:
\begin{equation}\label{eq:finite-dim-inclusion-algebras}
	\N = \bigoplus_{j} \BO[\mathcal{L}_j] \otimes \IdentityMatrix_{\mathcal{V}_j \otimes \HS}\,, \qquad
	\N' = \bigoplus_{j} \IdentityMatrix_{\mathcal{L}_j} \otimes \BO[\mathcal{V}_j \otimes \HS]\,, \qquad
	\N' \cap \M = \bigoplus_{j} \IdentityMatrix_{\mathcal{L}_j} \otimes \BO[\mathcal{V}_j] \otimes \IdentityMatrix_{\HS}\,.
\end{equation}

    Examples of this setting include the following: (1) The subsystem case: $\HS = \HS_A \otimes \HS_B$ and $\N = \IdentityMatrix_A \otimes \BO[\HS_B]$, a single summand with $\mathcal{L} = \HS_B$, $\mathcal{V} = \HS_A$. (2) Fixed-point algebras: let $g \mapsto U_g$ be a unitary representation of a finite group $G$ on $\HS$ with isotypic decomposition $\HS = \bigoplus_j \mathcal{L}_j \otimes \mathcal{V}_j$, $U_g = \bigoplus_j \IdentityMatrix_{\mathcal{L}_j} \otimes u_j(g)$, where the $u_j$ are pairwise inequivalent irreducible representations of dimensions $d_j = \dim \mathcal{V}_j$ and the $\mathcal{L}_j$ are their multiplicity spaces. Then the fixed-point algebra $\N = \Set{x \in \BO | U_g x U_g^* = x \text{ for all } g \in G}$ of \Autoref{ex:examples-group-invariant} is of the form \eqref{eq:finite-dim-inclusion}. (3) If all $\mathcal{V}_j$ and $\mathcal{L}_j$ are one-dimensional then $\N$ is a maximal abelian subalgebra, the operators diagonal in a fixed orthonormal basis $\set{\ket{a}}$ of $\HS$, and $\N' = \Set{\sum_a \ketbra{a}{a} \otimes y_a | y_a \in \BO}$ is the algebra of block-diagonal operators on $\HS \otimes \HS$.

    We will find that a conditional expectation $E: \M \to \N$ is uniquely specified by a density matrix $τ_j$ on each multiplicity space of $\N$, and that $E^{-1}$ is a map of the same form built from the inverses $τ_j^{-1}$ (which explains the notation), and that $\Ind E = \sum_j \Tr τ_j^{-1}$.

    \paragraph{Density matrices and spatial derivatives.}
For an algebra $\mathcal{A} = \bigoplus_i \BO[\mathcal{X}_i] \otimes \IdentityMatrix_{\mathcal{Y}_i}$ on $\bigoplus_i \mathcal{X}_i \otimes \mathcal{Y}_i$ in the normal form \eqref{eq:finite-dim-structure} every $φ \in \mathcal{A}_*^+$ is of the form
\begin{equation}\label{eq:finite-dim-density-matrices}
	φ\Big(\bigoplus_i x_i \otimes \IdentityMatrix_{\mathcal{Y}_i}\Big) = \sum_i \Tr(\hat{φ}_i x_i)
\end{equation}
for a unique family of positive operators $\hat{φ}_i \in \BO[\mathcal{X}_i]_+$, which we call the \emph{density matrices} of $φ$. Likewise $ψ \in (\mathcal{A}')_*^+$ has density matrices $\hat{ψ}_i \in \BO[\mathcal{Y}_i]_+$, and the spatial derivative of a faithful $ψ$ is
\begin{equation}\label{eq:finite-dim-spatial-derivative}
	\Delta(φ/ψ) = \bigoplus_i \hat{φ}_i \otimes \hat{ψ}_i^{-1}\,.
\end{equation}
For a single block this is \eqref{eq:spatial-derivative-finite-dim}.

\begin{proposition}[Conditional expectations onto $\N \subset \BO$]\label{prop:finite-dim-cond-exp}
	Let $\N \subset \BO$ be as in \eqref{eq:finite-dim-inclusion}, let $q_j$ be the projection onto $\mathcal{L}_j \otimes \mathcal{V}_j$, and for $x \in \BO$ let $x_j \coloneqq q_j x q_j \in \BO[\mathcal{L}_j \otimes \mathcal{V}_j]$ be its $j$-th diagonal block.
	\begin{enumerate}[label=(\roman*)]
		\item\label{item:fdce-form} The positive linear $\N$-bimodule maps $T: \BO \to \N$, equivalently the normal semifinite operator-valued weights from $\BO$ to $\N$, are exactly the maps
		\begin{equation}\label{eq:finite-dim-cond-exp}
			T(x) = \bigoplus_{j} \Tr_{\mathcal{V}_j}\big[(\IdentityMatrix_{\mathcal{L}_j} \otimes τ_j)\, x_j\big] \otimes \IdentityMatrix_{\mathcal{V}_j}\,, \qquad x \in \BO\,,
		\end{equation}
		for a family $(τ_j)_{j \in J}$ of positive operators $τ_j \in \BO[\mathcal{V}_j]_+$, which is uniquely determined by $T$. One has $T(\IdentityMatrix) = \sum_j (\Tr τ_j)\, q_j$. Thus $T$ is a conditional expectation iff every $τ_j$ is a density matrix on $\mathcal{V}_j$, and $T$ is faithful iff every $τ_j$ is invertible.
		\item\label{item:fdce-predual} If $σ \in \N_*^+$ has density matrices $\hat{σ}_j \in \BO[\mathcal{L}_j]_+$, then $T_*(σ) = σ \circ T \in \BO_*^+$ has the density matrix
		\begin{equation}\label{eq:finite-dim-cond-exp-predual}
			\widehat{T_*(σ)} = \bigoplus_{j} \hat{σ}_j \otimes τ_j\,.
		\end{equation}
		\item\label{item:fdce-trace} The trace-preserving conditional expectation, i.e.\ the unique $E$ with $\Tr \circ E = \Tr$, is the one with $τ_j = \IdentityMatrix_{\mathcal{V}_j} / \dim \mathcal{V}_j$ for all $j$.
	\end{enumerate}
\end{proposition}
\begin{proof}
	For conditional expectations this is well known, see e.g.\ \cite[Lemma~1 and App.~A]{magan_quantum_2021}; the modifications for operator-valued weights are given in \Cref{app:finite-dim-proofs}.
\end{proof}

A similar structure also holds for the inverse weight. Recall from \Autoref{sec:operator_valued_weights} that for a faithful conditional expectation $E: \M \to \N$ Haagerup's theorem provides a normal faithful semifinite operator-valued weight $E^{-1}: \N'_+ \to \widehat{\M'_+}$, characterized by the identity $\Delta(E_*(σ)/ω') = \Delta(σ/E^{-1}_*(ω'))$ for all n.f.s.\ weights $σ$ on $\N$ and $ω'$ on $\M'$, see \eqref{eq:inverse_weight}.

\begin{proposition}[The inverse weight]\label{prop:finite-dim-inverse-weight}
	In the setting \eqref{eq:finite-dim-setting}--\eqref{eq:finite-dim-inclusion-algebras}, let $E: \M \to \N$ be a faithful conditional expectation, given by \eqref{eq:finite-dim-cond-exp} with invertible density matrices $τ_j$. Write $y \in \N'$ as $y = \bigoplus_j \IdentityMatrix_{\mathcal{L}_j} \otimes y_j$ with $y_j \in \BO[\mathcal{V}_j \otimes \HS]$.
	\begin{enumerate}[label=(\roman*)]
		\item\label{item:fdiw-form} Haagerup's operator-valued weight $E^{-1}$ is the bounded, positive $\M'$-bimodule map $\N' \to \M'$ given by
		\begin{equation}\label{eq:finite-dim-inverse-weight}
			E^{-1}(y) = \IdentityMatrix_{\HS} \otimes \sum_{j} \Tr_{\mathcal{V}_j}\big[(τ_j^{-1} \otimes \IdentityMatrix_{\HS})\, y_j\big]\,, \qquad y \in \N'\,.
		\end{equation}
		\item\label{item:fdiw-predual} If $ω' \in (\M')_*^+$ has density matrix $\hat{ω}' \in \BO_+$, then $E^{-1}_*(ω') = ω' \circ E^{-1} \in (\N')_*^+$ has density matrices
		\begin{equation}\label{eq:finite-dim-inverse-weight-predual}
			\widehat{E^{-1}_*(ω')}_j = τ_j^{-1} \otimes \hat{ω}' \in \BO[\mathcal{V}_j \otimes \HS]_+\,,
		\end{equation}
		and for $σ \in \N_*^+$ with density matrices $\hat{σ}_j$ and faithful $ω'$, Haagerup's identity \eqref{eq:inverse_weight} takes the form
		\begin{equation}\label{eq:finite-dim-haagerup-identity}
			\Delta(E_*(σ)/ω') = \bigoplus_{j} \hat{σ}_j \otimes τ_j \otimes (\hat{ω}')^{-1} = \Delta(σ/E^{-1}_*(ω'))\,.
		\end{equation}
		\item\label{item:fdiw-index} The index of $E$ is the scalar
		\begin{equation}\label{eq:finite-dim-index}
			\Ind E = E^{-1}(\IdentityMatrix) = \sum_{j} \Tr\, τ_j^{-1} \;\geq\; \sum_j (\dim \mathcal{V}_j)^2\,,
		\end{equation}
		with equality iff $E$ is the trace-preserving expectation, and $E' \coloneqq (\Ind E)^{-1} E^{-1}$ is a faithful conditional expectation $\N' \to \M'$, so that $E^{-1} = (\Ind E)\, E'$.
	\end{enumerate}
\end{proposition}
\begin{proof}
	This is again an explicit computation, which evaluates Haagerup's identity \eqref{eq:inverse_weight} using the explicit form \eqref{eq:finite-dim-spatial-derivative} of the spatial derivative, and is given in \Cref{app:finite-dim-proofs}.
\end{proof}

\begin{example}[Subsystems]\label{ex:finite-dim-subsystem}
	For $\N = \IdentityMatrix_A \otimes \BO[\HS_B] \subset \BO[\HS_A \otimes \HS_B]$ there is a single summand, with $\mathcal{L} = \HS_B$ and $\mathcal{V} = \HS_A$, so the conditional expectations onto $\N$ are the maps $E(x_{AB}) = \IdentityMatrix_A \otimes \Tr_A[(τ_A \otimes \IdentityMatrix_B)\, x_{AB}]$ with $τ_A$ a density matrix on $\HS_A$. For the trace-preserving one, $τ_A = \IdentityMatrix_A / d_A$, and writing $\HS_C$ for the second (purifying) copy of $\HS$,
	\begin{equation}
		E^{-1}(y_{AC} \otimes \IdentityMatrix_B) = d_A\, \IdentityMatrix_{AB} \otimes \Tr_A(y_{AC})\,, \qquad \Ind E = d_A^2\,,
	\end{equation}
	in accordance with \Autoref{cor:finite_dim_duality_subfactor} and the discussion of \eqref{eq:index-definition}.
\end{example}

\begin{example}[Fixed-point algebras]\label{ex:finite-dim-fixed-point}
	Let $\N = \Set{x \in \BO | U_g x U_g^* = x \text{ for all } g \in G}$ be the fixed-point algebra of a unitary representation $U$ of a finite group $G$, with isotypic decomposition $\HS = \bigoplus_j \mathcal{L}_j \otimes \mathcal{V}_j$ as above and $d_j = \dim \mathcal{V}_j$. The twirl $E(x) = \frac{1}{|G|} \sum_{g \in G} U_g x U_g^*$ is the trace-preserving conditional expectation onto $\N$, so $τ_j = \IdentityMatrix_{\mathcal{V}_j} / d_j$ (this is Schur's lemma), and
	\begin{equation}
		E^{-1}\Big(\bigoplus_j \IdentityMatrix_{\mathcal{L}_j} \otimes y_j\Big) = \IdentityMatrix_{\HS} \otimes \sum_j d_j \Tr_{\mathcal{V}_j}(y_j)\,, \qquad \Ind E = \sum_j d_j^2\,,
	\end{equation}
	where the sum runs over the irreducible representations occurring in $U$; for the regular representation, $\Ind E = |G|$.
\end{example}

\begin{example}[Pinching/Coherence]\label{ex:finite-dim-pinching}
	Let $\N$ be the algebra of operators diagonal in a fixed orthonormal basis $\set{\ket{a}}_{a=1}^{d}$ of $\HS$, $d = \dim \HS$, so that $J = \set{1, \dots, d}$, $\mathcal{L}_a = \mathcal{V}_a = \complex$, and $\N' = \Set{\sum_a \ketbra{a}{a} \otimes y_a | y_a \in \BO}$. Since all $\mathcal{V}_a$ are one-dimensional, \Autoref{prop:finite-dim-cond-exp} leaves no freedom: the only conditional expectation onto $\N$ is the pinching $E(x) = \sum_a \ketbra{a}{a}\, x\, \ketbra{a}{a}$, and \Autoref{prop:finite-dim-inverse-weight} gives
	\begin{equation}
		E^{-1}\Big(\sum_a \ketbra{a}{a} \otimes y_a\Big) = \IdentityMatrix_{\HS} \otimes \sum_a y_a\,, \qquad \Ind E = d\,, \qquad E' = \frac{1}{d}\, E^{-1}\,.
	\end{equation}
	To connect with the notation of \Autoref{sec:intro-duality-origin}, write $\HS \otimes \HS = \HS_A \otimes \HS_C$. For a state $ρ \in \S(\M)$ with vector representative $\ket{\v{ρ}} \in \HS_A \otimes \HS_C$, the state $ρ|_{\M}$ is the reduced density matrix $ρ_A$, and $ρ|_{\N'}$ has density matrices $\hat{ρ}_a = (\bra{a} \otimes \IdentityMatrix_C)\ketbra{\v{ρ}}{\v{ρ}}(\ket{a} \otimes \IdentityMatrix_C) \in \BO[\HS_C]_+$, i.e.\ it is the classical-quantum state $\sum_a \ketbra{a}{a} \otimes \hat{ρ}_a = (E \otimes \id_C)(\ketbra{\v{ρ}}{\v{ρ}})$ obtained by measuring $A$ in the basis $\set{\ket{a}}$. Since $E_*(\S(\N))$ is the set of states diagonal in this basis, and by \eqref{eq:finite-dim-inverse-weight-predual} $E^{-1}_*(\S(\M'))$ consists of the restrictions to $\N'$ of the functionals $\Tr[(\IdentityMatrix_A \otimes ω_C)\, \cdot\,]$ with $ω_C \in \S(\BO[\HS_C])$, \Autoref{thm:finite_dim_duality} reads in this case
	\begin{equation}
		\min_{σ_A \text{ diagonal}} \widetilde{D}_α(ρ_A \| σ_A) = \widetilde{H}_β(A|C)_{(E \otimes \id_C)(\ketbra{\v{ρ}}{\v{ρ}})}\,, \qquad \tfrac{1}{α} + \tfrac{1}{β} = 2\,,
	\end{equation}
	with $\widetilde{H}_β(A|C)$ as in \Autoref{sec:intro-duality-origin}: the sandwiched \renyi relative entropy of coherence of $ρ_A$ equals the conditional \renyi entropy of the state obtained by measuring a purification of $ρ_A$ in the incoherent basis. For $α = β = 1$ this is the familiar identity between the relative entropy of coherence $H(E(ρ_A)) - H(ρ_A)$ (see \cite{winter_operational_2016}) and the conditional entropy $H(A|C)$ after the measurement.
\end{example}

\section{Duality Relations}\label{sec:duality-relations}
\paragraph{Standing assumptions.}
Throughout this section, $\N \subset \M \subset \BO$ are von Neumann algebras with commutants $\M' \subset \N'$, and $E: \M \to \N$ is a normal faithful conditional expectation of \emph{finite index}, i.e.\ $\Ind E = E^{-1}(\IdentityMatrix)$ is a bounded operator (see \eqref{eq:index-definition}), where $E^{-1}: \N'_+ \to \widehat{\M'_+}$ is Haagerup's dual operator-valued weight (\Autoref{sec:operator_valued_weights}). In fact, everything below remains valid if $E$ is merely a normal faithful semifinite operator-valued weight with $E(\IdentityMatrix)$ and $\Ind E$ bounded. By \Autoref{lem:finite_index_bounded_inverse_weight}, $E$ and $E^{-1}$ are bounded positive normal linear maps, and
\begin{equation}\label{eq:finite-index-preduals}
    E_*(σ) \coloneqq σ \circ E \in \M_*^+ \quad \text{for } σ \in \N_*^+\,, \qquad\qquad E^{-1}_*(ω') \coloneqq ω' \circ E^{-1} \in (\N')_*^+ \quad \text{for } ω' \in (\M')_*^+\,,
\end{equation}
so that all divergences appearing below have positive normal functionals as second arguments and are the ones of \Autoref{def:divergences} and \Autoref{def:one-shot}; in particular, $\D(ρ\|E_*(\Ns))$ and $\D(ρ\|E^{-1}_*(\M'_*))$ are understood as in \Autoref{def:one-shot}\ref{item:os-sets}. Finally, $ρ \in \S(\M)$ is a normal state with a fixed vector representative $\ket{\v{ρ}} \in \HS$, and $ρ|_{\mathcal{A}} \coloneqq ω_{\v{ρ}}|_{\mathcal{A}}$ denotes the induced state on $\mathcal{A} \in \set{\M, \N, \M', \N'}$ (\Autoref{sec:prelim-vn}); as before, $ρ' = ρ|_{\M'}$.

The technical engine behind all duality relations is Haagerup's identity \eqref{eq:inverse_weight}, which characterizes $E^{-1}$ through the spatial derivative. In the present setting it takes the following form.

\begin{lemma}[Haagerup's identity for non-faithful functionals]\label{lem:haagerup-identity-finite-index}
    Let $σ \in \N_*^+$ and $ω' \in (\M')_*^+$. Then $s(E_*(σ)) = s(σ)$ and $s(E^{-1}_*(ω')) = s(ω')$, and
    \begin{equation}\label{eq:haagerup-identity-finite-index}
        \Delta(E_*(σ)/ω') = \Delta(σ/E^{-1}_*(ω'))\,, \qquad \text{equivalently} \qquad \Delta(ω'/E_*(σ)) = \Delta(E^{-1}_*(ω')/σ)\,,
    \end{equation}
    as identities between the generalized positive operators of \Autoref{def:spatial-derivative}. If $σ$ and $ω'$ are faithful, then so are $E_*(σ)$ and $E^{-1}_*(ω')$, all four operators in \eqref{eq:haagerup-identity-finite-index} are positive self-adjoint operators on $\HS$, and \eqref{eq:haagerup-identity-finite-index} is Haagerup's identity \eqref{eq:inverse_weight}.
\end{lemma}
\begin{proof}
The proof, which reduces the general case to \eqref{eq:inverse_weight} by completing the non-faithful functionals to faithful weights, is given in \Cref{app:haagerup-proof}.
\end{proof}

\subsection{Petz and Umegaki}
The framework of Haagerup's operator-valued weights, and in particular the identity \eqref{eq:haagerup-identity-finite-index}, very naturally gives rise to conditional-entropy duality relations. The following duality relation for the Petz-relative entropy follows very straightforwardly: 
\begin{theorem}[Petz-\renyi duality relation]\label{thm:petz_duality_relation}
    Under the standing assumptions of this section, let $α \in [0, 2]$ and $β = 2 - α$. Then
    \begin{equation}
        D_α(ρ|_{\M}\|E_*(ρ|_{\N})) = - D_β(ρ|_{\N'} \|E^{-1}_{*}(ρ|_{\M'}))\,,
    \end{equation}
    where for $α = β = 1$ both sides are relative entropies, $D_1 = D$ (\Autoref{def:divergences}).
\end{theorem}
\begin{proof}
    The case $α = 1$ is \Autoref{cor:relative_entropy_unoptimized_duality} below, so let $α \neq 1$. By \eqref{eq:def-petz-renyi},
    \begin{align}
        D_α(ρ|_{\M}\|E_*(ρ|_{\N})) &= {1 \over α - 1} \log \bra{\v{ρ}} \Delta(ρ|_{\M'}/E_*(ρ|_{\N}))^{α - 1} \ket{\v{ρ}} \\
        &= {1 \over α - 1} \log \bra{\v{ρ}} \Delta(E_*(ρ|_{\N})/ρ|_{\M'})^{1 - α} \ket{\v{ρ}} \\
        &= {1 \over α - 1} \log \bra{\v{ρ}} \Delta(ρ|_{\N}/E^{-1}_*(ρ|_{\M'}))^{1 - α} \ket{\v{ρ}} \\
        &= - {1 \over β - 1} \log \bra{\v{ρ}} \Delta(ρ|_{\N}/E^{-1}_*(ρ|_{\M'}))^{β - 1} \ket{\v{ρ}} \\
        &= - D_β(ρ|_{\N'} \|E^{-1}_{*}(ρ|_{\M'}))\,,
    \end{align}
    where the second equality is the inversion formula $\Delta(ρ'/σ)= \Delta(σ/ρ')^{-1}$ (\Autoref{prop:spatial-derivative-properties}\ref{item:sd-inverse}), the third is Haagerup's identity in the form of \Autoref{lem:haagerup-identity-finite-index}, the fourth uses $β - 1 = 1 - α$, and the last one is \eqref{eq:def-petz-renyi} applied to the algebra $\N'$, whose commutant is $\N$, so that the first argument of the spatial derivative is the state $ρ|_{\N}$ induced by $\ket{\v{ρ}}$ on the commutant of $\N'$.
\end{proof}

Since $E_*(ρ|_{\N})$ and $E^{-1}_*(ρ|_{\M'})$ are positive normal functionals by \eqref{eq:finite-index-preduals}, we can take the limit $α \to 1$ using \Autoref{prop:divergence-properties}\ref{item:dp-limits} and obtain the following statement. Note that a similar statement was also already previously obtained in \cite{xu_relative_2019, longo_note_2022} (under the stronger assumption that $\N \subset \M$ is an inclusion of factors).

\begin{corollary}[Relative entropy duality relation]\label{cor:relative_entropy_unoptimized_duality}
    Under the standing assumptions of this section,
    \begin{equation}
        D(ρ|_{\M}\|E_*(ρ|_{\N}))= - D(ρ|_{\N'} \|E^{-1}_{*}(ρ|_{\M'}))\,.
    \end{equation}
\end{corollary}
\begin{proof}
    For $α \in [0, 1)$, \Autoref{thm:petz_duality_relation} gives $D_α(ρ|_{\M}\|E_*(ρ|_{\N})) = - D_{2 - α}(ρ|_{\N'}\|E^{-1}_*(ρ|_{\M'}))$. As $α \nearrow 1$ the left hand side converges to $D(ρ|_{\M}\|E_*(ρ|_{\N}))$ by \Autoref{prop:divergence-properties}\ref{item:dp-limits}. Since the left hand side is bounded below (\Autoref{lem:first-moment}), the identity shows that $D_{2-α}(ρ|_{\N'}\|E^{-1}_*(ρ|_{\M'}))$ is finite for $α \in [0, 1)$, so that the right hand side converges to $-D(ρ|_{\N'}\|E^{-1}_*(ρ|_{\M'}))$ as $α \nearrow 1$, again by \Autoref{prop:divergence-properties}\ref{item:dp-limits} (limit of $D_β$ as $β \searrow 1$).
\end{proof}

Additionally, one can consider duality relations in the setting where one optimizes over all states in the second argument. In the usual bipartite setting it is well-known that for the relative entropy (i.e. $α = 1$) this is equivalent to just conditioning on the marginal, and this holds true even in more general settings. We need this fact not only for the conditional expectation $E$ but also for its dual $E^{-1}$, which is not a conditional expectation, and therefore formulate it for a general operator-valued weight $T$:
\begin{proposition}[Pythagorean identity for operator-valued weights]\label{prop:pythagorean_identity}
    Let $\N \subset \M \subset \BO$ be von Neumann algebras, and let $T: \M_+ \to \widehat{\N_+}$ be a normal faithful semifinite operator-valued weight, not necessarily a conditional expectation, such that $T(\IdentityMatrix)$ and $\Ind T = T^{-1}(\IdentityMatrix)$ are bounded, so that $T_*(σ) \coloneqq σ \circ T \in \M_*^+$ for $σ \in \N_*^+$ (\Autoref{lem:finite_index_bounded_inverse_weight}). Let $ρ \in \S(\M)$ be a state with a vector representative in $\HS$, and let $σ \in \S(\N)$. Then, as an identity in $(-\infty, +\infty]$,
    \begin{equation}\label{eq:pythagorean-identity}
        D(ρ\|T_*(σ)) = D(ρ\|T_*(ρ|_{\N})) + D(ρ|_{\N}\|σ)\,,
    \end{equation}
    and in particular $\inf_{σ \in \S(\N)} D(ρ\|T_*(σ)) = D(ρ\|T_*(ρ|_{\N}))$.
\end{proposition}
\begin{proof}
    If $T$ is a normal faithful conditional expectation, then \eqref{eq:pythagorean-identity} is a theorem of Petz \cite[Thm.~5.15]{ohya_quantum_1993}. The general case is reduced to this one in \Cref{app:pythagorean-proof}.
\end{proof}

Together with \Autoref{cor:relative_entropy_unoptimized_duality}, this shows that the duality relation for the relative entropy also holds for the optimized quantities,
\begin{equation}\label{eq:optimized-relative-entropy-duality}
    D(ρ|_{\M}\|E_*(\Ns)) = D(ρ|_{\M}\|E_*(ρ|_{\N})) = - D(ρ|_{\N'}\|E^{-1}_*(ρ|_{\M'})) = - D(ρ|_{\N'}\|E^{-1}_*(\M'_*))\,.
\end{equation}

\subsection{Sandwiched \renyi relative entropy}\label{sec:duality-relation-sandwiched}
\newcommand{\wt}{\widetilde}
Proving the sandwiched duality relation is significantly more tricky, since we need to invoke a minimax theorem, which is a lot less straightforward in infinite dimensions compared to finite-dimensional settings. We start by developing all the necessary analytical tools to prove the duality relation (which is \autoref{thm:sandwiched_duality_relation} at the end of this section). The idea is that we want to apply the following minimax theorem:
\begin{lemma}[{Generalized minimax theorem~\cite[Theorem 5.2]{farkas_potential_2006}}]
    \label{lem:gen_minimax}
    Let $X$ be a compact and convex subset of a Hausdorff topological vector space and let $Y$ be a convex subset of a linear space. Let $f : X \times Y \to  \reals \cup \{\infty\}$ be lower semi-continuous on $X$ for fixed $y \in Y$, convex in x and concave in y. Then 
    \begin{equation}
        \sup_{y \in Y} \inf_{x \in X} f(x, y) =  \inf_{x \in X} \sup_{y \in Y}f(x,y)\,.
    \end{equation}
\end{lemma}
Unfortunately, the space of normal states $\S(\M)$ is not weak-$*$ compact in general, so we will have to extend the state space to $\S^*(\M) = \Set{ρ \in (\M^*)_+ | ρ(1) = 1}$ which also includes non-normal functionals, and which is weak-$*$ compact by the Banach-Alaoglu Theorem.
It is well-known (see e.g. \cite[Theorem III.2.14]{takesaki_theory_1979}) that every functional $\widehat{φ} \in \S^*(\M)$ can be decomposed into a normal and singular part $\widehat{φ} = φ_n + φ_s$, where the singular part satisfies the following properties:
\begin{lemma}[Singular functionals]\label{lem:singular functional properties}
Let $ω$ be a positive (not necessarily normal) linear functional on \M. Then, the following are equivalent:
\begin{enumerate}[(i)]
\item $ω$ is singular, i.e.\ if $μ \in \M_*^+$ is a positive normal linear functional such that $0 \leq μ \leq ω$ then $μ = 0$.
\item For every nonzero projection $p \in \M$ there exists a nonzero projection $q \leq p$ such that $ω(q) = 0$
\item There exists an increasing net of projections $(e_i)_i \subset \M$ such that $e_i \nearrow 1$ strongly and $ω(e_i) = 0$ for all $i$.
\end{enumerate}
\end{lemma}
\begin{proof}
The equivalence (i)$\Leftrightarrow$(ii) is due to \cite[Theorem III.3.8]{takesaki_theory_1979}. For the proof of (ii)$\Rightarrow$(iii) consider the following collection of all pairwise orthogonal sets of projections $p$ with $ω(p) = 0$.
\begin{equation}
    C \coloneqq \Set{ F \subset \mathrm{Proj}(\M) | \forall p \in F:ω(p) = 0, \forall p, p' \in F \text{ s.t. }p \neq p': pp' = 0 \, }
\end{equation}
The elements of $C$ are ordered by set inclusion, and every chain $(F_λ)_λ \subset C$ has an upper bound, namely the union $F \coloneqq \bigcup_λ F_λ$, which is also in $C$ (since every pair of elements has to lie in some $F_λ$, which is in $C$). Thus, by Zorn's lemma $C$ has a maximal element $F_0 = \Set{ q_j | j \in J} \in C$. Consider the supremum of all the projections $q \coloneqq \bigvee_{j \in J} q_j$ (i.e. the projection onto the union of all the subspaces projected onto). If $q \neq 1$, then we can apply (ii) to $1 - q$ and obtain a nonzero projection $q' \leq 1 - q$ such that $ω(q') = 0$, and $q'$ is orthogonal to $q$ which contradicts the maximality of $F_0$, hence $q = 1$. Now for each finite subset $I \subset J$ define
\begin{equation}
    e_I \coloneqq \sum_{j \in I} q_j\,.
\end{equation}
By ordering the sets $I$ by inclusion, this defines a net $(e_I)_I$, and $q = 1$ implies that $e_I \nearrow 1$ strongly. Additionally, since every $I$ is finite $ω(e_I) = \sum_{j \in I} ω(q_j) = 0$, which proves (iii).

For the proof of (iii)$\Rightarrow$(i) let $μ \in \M_*^+$ be a positive normal linear functional with $0 \leq μ \leq ω$. Then $0 \leq μ(e_i) \leq ω(e_i) = 0$ for all $i$, and since $μ$ is normal and $e_i \nearrow 1$, we get $μ(1) = \lim_i μ(e_i) = 0$. Since $\norm{μ} = μ(1)$ for a positive functional, this means $μ = 0$, i.e.\ $ω$ is singular.
\end{proof}

The point of the next Lemma is to extend the definition of the spatial derivative to the case where the first argument is not normal. 
\renewcommand{\wp}{\widehat{\varphi}}
\newcommand{\f}{f_{\v{\xi}}}
\begin{lemma}[Spatial derivative extension]\label{lemma:spatial derivate non-normal extension}
Let $\M \subset \BO$, let $\wp \in (\M^*)_+$ be a (not necessarily normal) positive linear functional on $\M$, let $ψ \in (\M')_*^+$ be a (not necessarily faithful) positive normal functional on $\M'$, so that $D(\HS, ψ)$, $θ^ψ$ and $\Delta(\cdot/ψ)$ are as in \Autoref{sec:prelim-spatial} (\Autoref{def:spatial-derivative}), let $\v{\xi} \in \HS$ and $α \in [-1, 0)$.
\begin{enumerate}[(i)]
\item The function
\[φ \in \M_*^+ \mapsto \bra{\v{\xi}} \Delta(φ/ψ)^α \ket{\v{ξ}}\] has a convex and weak-$*$ lower semi-continuous extension to $(\M^*)_+$ via
\begin{equation}\label{eq:definition f}
\wp \in (\M^*)_+ \mapsto \f(\wp, \psi) \coloneqq \begin{cases} \displaystyle{\abs{\sin(\pi α)} \over \pi} \int_{0}^{\infty} t^α g^{(t)}_{\v{\xi}}(\wp, \psi)\, dt & \text{for } α \in (-1, 0)\,, \\[2ex] \displaystyle\sup_{t > 0}\, g^{(t)}_{\v{\xi}}(\wp, \psi) & \text{for } α = -1\,, \end{cases}
\end{equation}
where \begin{equation}\label{eq:definition g}
g_{\v{\xi}}^{(t)}(\wp, \psi) \coloneqq \sup_{\v{\eta} \in D(\HS, \psi)} \qty[2 \Re \ip{\v{\xi}}{\v{\eta}} - t \norm{\v{\eta}}^2 - \wp(θ^ψ(\v{\eta}, \v{\eta}))]\,.
\end{equation} Note that weak-$*$ lower semi-continuity also implies lower semi-continuity in any stronger topology such as the strong or norm topology.
\item If $\wp = φ_n + φ_s$ is the decomposition into normal and singular part, then $\f(\wp, ψ) = \f(φ_n, ψ)$.
\item For fixed $\wp$, the function $ψ \mapsto \f(\wp, ψ)$ is concave on $(\M')_*^+$.
\end{enumerate}
\end{lemma}
\begin{proof}\mbox{}\\
    \begin{enumerate}[(i)]
        \item We start by showing that for any normal positive linear functional $φ \in \M_*^+$ we have $\bra{\v{\xi}} \Delta(φ/ψ)^α \ket{\v{ξ}} = f_{\xi}(φ, \psi)$. For this, we start with the integral representation of $x^α$ for $α \in (-1, 0)$:
        \begin{equation}
        x^α = {\abs{\sin(\pi α)} \over \pi} \int_{0}^\infty t^α (t + x)^{-1} dt\,,
        \end{equation}
        which holds for all $x \in [0, \infty]$ with our conventions. Integrating against the spectral measure $μ^{\Delta(φ/ψ)}_{\v{ξ}}$ on $[0, \infty]$ and using Tonelli's theorem, this implies
        \begin{equation}
        \bra{\v{\xi}} \Delta(φ/ψ)^α \ket{\v{ξ}} = {\abs{\sin(\pi α)} \over \pi}  \int_{0}^\infty t^α \bra{\v{\xi}} (t + \Delta(φ/ψ))^{-1} \ket{\v{ξ}} dt\,,
        \end{equation}
        where $(t + \Delta(φ/ψ))^{-1}$ is a bounded operator that vanishes on $(\IdentityMatrix - s(ψ))\HS$, the $+\infty$ part of $\Delta(φ/ψ)$. For $α = -1$ the integral representation is not available (the integral diverges), and we use instead that $(t + x)^{-1}$ increases to $x^{-1}$ as $t$ decreases to $0$, for all $x \in [0, \infty]$. By monotone convergence this implies
        \begin{equation}
        \bra{\v{\xi}} \Delta(φ/ψ)^{-1} \ket{\v{ξ}} = \sup_{t > 0}\, \bra{\v{\xi}} (t + \Delta(φ/ψ))^{-1} \ket{\v{ξ}}\,.
        \end{equation}
        In both cases it remains to express the resolvent through $g^{(t)}_{\v{\xi}}$.

        By \autoref{lem:resolvent-variational-formula} applied to the finite part of $\Delta(φ/ψ)$ (see \eqref{eq:resolvent-variational-generalized}) we have
        \begin{equation}\label{eq:rel-mod-op-variational-formula}
        \bra{\v{\xi}} (t + \Delta(φ/ψ))^{-1} \ket{\v{ξ}} = \sup_{\v{\eta} \in \mathcal{D}(\overline{q_{φ}})} \qty[2 \Re \ip{\v{\xi}}{\v{\eta}} - t \norm{\v{\eta}}^2 - \overline{q_{φ}}(\v{\eta})]\,,
        \end{equation}
        where, by \Autoref{def:spatial-derivative}, $\overline{q_{φ}}(\v{η}) = \ip{\v{η}}{\Delta(φ/ψ)\v{η}}$ is the closure of the quadratic form
        \begin{equation}
        \v\eta \in D(\HS, ψ) \mapsto q_{φ}(\v{\eta}) = φ(θ^ψ(\v{\eta}, \v{\eta}))
        \end{equation}
        and $\mathcal{D}(\overline{q_{φ}}) = \Set{\v{η} \in \HS | \ip{\v{η}}{\Delta(φ/ψ)\v{η}} < \infty}$ is its domain. Since $D(\HS, ψ)$ is dense in $\mathcal{D}(\overline{q_{φ}})$ with respect to the form norm $(\norm{\v{η}}^2 + \overline{q_φ}(\v{η}))^{1/2}$, and the objective function in \eqref{eq:rel-mod-op-variational-formula} is continuous with respect to this norm, the supremum does not change if it is restricted to $\v\eta \in D(\HS, ψ)$. This gives $\bra{\v{\xi}} (t + \Delta(φ/ψ))^{-1} \ket{\v{ξ}} = g_{\v{\xi}}^{(t)}(φ, ψ)$ for all $t > 0$, and hence $\bra{\v{\xi}} \Delta(φ/ψ)^α \ket{\v{ξ}} = \f(φ, ψ)$ for all $α \in [-1, 0)$.

        It remains to show that $\f$ is weak-$*$ lower semi-continuous and convex in $\wp$. For this note that the objective function in \eqref{eq:definition g} is linear and weak-$*$ continuous in $\wp$, and hence $g_{\v{\xi}}^{(t)}$ is convex and weak-$*$ lower semi-continuous in $\wp$. For $α = -1$, $\f$ is the supremum of the functions $g_{\v{\xi}}^{(t)}$, $t > 0$, and therefore also convex and weak-$*$ lower semi-continuous. For $α \in (-1, 0)$, note that for fixed $\wp$ the integrand $t \mapsto t^α g_{\v{\xi}}^{(t)}(\wp, ψ)$ in \eqref{eq:definition f} is nonincreasing in $t$ and takes values in $[0, t^{α - 1}\norm{\v{ξ}}^2]$ (since $\v\eta = 0$ is always a valid choice in \eqref{eq:definition g}, and since $2 \Re \ip{\v{\xi}}{\v{\eta}} - t \norm{\v{\eta}}^2 \leq \norm{\v{ξ}}^2/t$). Its integral can therefore be approximated from below by Riemann sums, i.e.\ it is the supremum of finite linear combinations, with positive coefficients, of the functions $g_{\v{\xi}}^{(t_i)}(\wp, ψ)$ with $t_i > 0$. Hence also $\f$ is convex and weak-$*$ lower semi-continuous, as a supremum of convex and weak-$*$ lower semi-continuous functions (Fatou's lemma is not available here, since the weak-$*$ topology is not metrizable in general and lower semi-continuity has to be checked on nets).
        \item By \autoref{lem:singular functional properties}~(iii) applied to $φ_s$, there exists an increasing net of projections $(e_i)_i \subset \M$ such that $(e_i)_{i \in I} \nearrow 1$ strongly and $φ_s(e_i) = 0$ for all $i$. Take $\v\eta \in D(\HS, ψ)$ and write $a = θ^ψ(\v\eta, \v\eta)$. Then, also $e_i \v\eta \in D(\HS, ψ)$ and $θ^ψ(e_i \v\eta, e_i \v\eta) = e_i a e_i$, since $R^ψ(x\v{η}) = x R^ψ(\v{η})$ for $x \in \M$ (\Autoref{def:spatial-derivative}). Since operator multiplication is strongly continuous on bounded sets (and both the $e_i$ and $a$ are bounded) we get $e_i a e_i \to a$ strongly. Now by \cite[Theorem II.2.6(ii)]{takesaki_theory_1979} $φ_n$ being normal also implies it being strongly continuous on any bounded set, and so $φ_n(e_i a e_i) \to φ_n(a)$. Additionally, since $0 \leq e_i a e_i \leq \norm{a} e_i$ we still have that $φ_s(e_i a e_i) = 0$ for all $i$. The claim is now that, for every $t > 0$, the optimization in \eqref{eq:definition g} can be restricted to $\Set{e_i \v{\eta} | \v{\eta} \in D(\HS, ψ),\, i \in I}$, which implies $g^{(t)}_{\v{\xi}}(\wp, ψ) = g^{(t)}_{\v{\xi}}(φ_n, ψ)$ and hence the desired statement (in both cases of \eqref{eq:definition f}), since $\wp(θ^ψ(e_i \v\eta, e_i \v\eta)) = φ_n(θ^ψ(e_i \v{\eta}, e_i \v{\eta}))$. To show the claim we only have to show that the value taken at any $\v\eta$ is upper bounded by the limit of the values taken at $e_i\v{\eta}$. Crucially, $\norm{e_i \v\eta}^2 \to \norm{\eta}^2$ and $\ip{\v\xi}{e_i \v{\eta}} \to \ip{\v\xi}{\v\eta}$, and so the only discontinuity can come from the $\wp$ term. We have $\wp(θ^ψ(e_i\v{\eta}, e_i\v{\eta})) = φ_n(e_i a e_i) + φ_s(e_i a e_i) = φ_n(e_i a e_i) \to φ_n(a)$, as well as $\wp(θ^ψ(\v{\eta}, \v{\eta})) = φ_n(a) + φ_s(a) \geq φ_n(a)$ which gives the desired inequality since $\wp(θ^ψ(\v\eta, \v\eta))$ appears with a minus sign in the objective function.
        \item By part (ii) we may assume that $\wp = φ \in \M_*^+$ is normal, so that $\f(φ, ψ) = \bra{\v{\xi}} \Delta(φ/ψ)^α \ket{\v{ξ}}$ with $\Delta(φ/ψ)$ as in \Autoref{def:spatial-derivative}. Let $e \coloneqq s(φ) \in \M$ and $e' \coloneqq s(ψ) \in \M'$, two commuting projections. By \Autoref{prop:spatial-derivative-properties}\ref{item:sd-support}, $\Delta(φ/ψ)$ is $+\infty$ on $(\IdentityMatrix - e')\HS$, $0$ on $(\IdentityMatrix - e)e'\HS$ and an injective positive operator $A$ on $ee'\HS$, and by \Autoref{prop:spatial-derivative-properties}\ref{item:sd-inverse} the generalized positive operator $\Delta(ψ/φ)$ (\Autoref{def:spatial-derivative} with the roles of $\M$ and $\M'$ interchanged) is $+\infty$ on $(\IdentityMatrix - e)\HS$, $0$ on $(\IdentityMatrix - e')e\HS$ and equal to $A^{-1}$ on $ee'\HS$. Since $α < 0$ and $-α \in (0, 1]$, we obtain
        \begin{equation}\label{eq:concavity-decomposition}
            \f(φ, ψ) = \infty \cdot \norm{(\IdentityMatrix - e)e'\v{\xi}}^2 + \bra{ee'\v{\xi}} A^{α} \ket{ee'\v{\xi}} = \infty \cdot \norm{(\IdentityMatrix - e)s(ψ)\v{\xi}}^2 + \bra{e\v{\xi}} \Delta(ψ/φ)^{-α} \ket{e\v{\xi}}\,,
        \end{equation}
        with $\infty \cdot 0 = 0$: the $+\infty$ part of $\Delta(φ/ψ)$ contributes $\infty^{α} = 0$, its kernel contributes $0^{α} = \infty$, and $\Delta(ψ/φ)^{-α}$ vanishes on $(\IdentityMatrix - e')e\HS$ and equals $A^{α}$ on $ee'\HS$, while $e\v{\xi}$ has no component in $(\IdentityMatrix - e)\HS$. Both terms on the right hand side of \eqref{eq:concavity-decomposition} are concave functions of $ψ \in (\M')_*^+$ with values in $[0, \infty]$. For the first term let $ψ = μψ_1 + (1 - μ)ψ_2$ with $μ \in (0,1)$; then $s(ψ) = s(ψ_1) \vee s(ψ_2)$, so $(\IdentityMatrix - e)s(ψ) \geq (\IdentityMatrix - e)s(ψ_i)$ as commuting projections, and the first term is $+\infty$ at $ψ$ as soon as it is $+\infty$ at $ψ_1$ or at $ψ_2$, which is the required concavity inequality. For the second term note that the vector $e\v{\xi}$ does not depend on $ψ$, that $ψ \mapsto \Delta(ψ/φ)$ is affine by \Autoref{prop:spatial-derivative-properties}\ref{item:sd-linearity} (additivity and homogeneity in the first argument), and that $C \mapsto \ip{ζ}{C^{-α}ζ}$ is concave on $\widehat{\BO_+}$ for $-α \in (0, 1]$ by \Autoref{lem:compression-inequality}\ref{item:ci-concavity} (for $-α = 1$ it is affine). Since a sum of concave $[0, \infty]$-valued functions is concave, (iii) follows.
    \end{enumerate}
\end{proof}

Similar to the previous sections, the proof of the duality relation will rely crucially on the inversion formula \Autoref{prop:spatial-derivative-properties}\ref{item:sd-inverse}, which however does not hold outside of the support of both arguments. The following lemma shows that this does not matter for the optimizations we are interested in.

\begin{lemma}[Supports in the optimizations]\label{lem:support-bookkeeping}
    Under the standing assumptions of this section, let $α \in [1/2, 1) \cup (1, \infty]$, $α' = (α - 1)/α$ (with $α' = 1$ for $α = \infty$, as in \Autoref{def:divergences}) and $p \coloneqq s(ρ') = [\M\v{ρ}] \in \M'$. For $ω' \in \S(\M')$ and $σ \in \S(\N)$ put
    \begin{align}\label{eq:XY}
        F_α(ω', σ) &\coloneqq \bra{\v{ρ}} \Delta(ω'/E_*(σ))^{α'} \ket{\v{ρ}} = \bra{\v{ρ}} \Delta(E^{-1}_*(ω')/σ)^{α'} \ket{\v{ρ}}\,\\
        G_α(ω', σ) &\coloneqq \bra{\v{ρ}} \Delta(σ/E^{-1}_*(ω'))^{-α'} \ket{\v{ρ}}\,,
    \end{align}
    where the two expressions for $F_α$ agree by \Autoref{lem:haagerup-identity-finite-index}. Then:
    \begin{enumerate}[label=(\alph*)]
        \item\label{item:sb-pointwise} $F_α(ω', σ) = G_α(ω', σ)$ whenever $(\IdentityMatrix - s(σ))(\IdentityMatrix - s(ω'))\v{ρ} = 0$, in particular whenever $s(ω') \geq p$. In general $F_α \geq G_α$ if $α > 1$ and $F_α \leq G_α$ if $α < 1$.
        \item\label{item:sb-optimizations} Let $ω'$ range over $\S(\M')$ and $σ$ over $\S(\N)$. If $α > 1$, then $\sup_{ω'} F_α(ω', σ) = \sup_{ω'} G_α(ω', σ)$ for every $σ$, and $\sup_{ω'} \inf_{σ} F_α = \sup_{ω'} \inf_{σ} G_α$. If $α < 1$, then $\inf_{ω'} F_α(ω', σ) = \inf_{ω'} G_α(ω', σ)$ for every $σ$, and $\inf_{ω'} \sup_{σ} F_α = \inf_{ω'} \sup_{σ} G_α$. In all four identities the optimization over $ω'$ may equivalently be restricted to the states with $s(ω') \geq p$.
    \end{enumerate}
\end{lemma}
\begin{proof}
    Fix $ω'$ and $σ$, and let $P \coloneqq s(ω') \in \M'$ and $Q \coloneqq s(σ) \in \N \subset \M$, two commuting projections. By \Autoref{lem:haagerup-identity-finite-index}, $s(E_*(σ)) = Q$ and $s(E^{-1}_*(ω')) = P$, so \Autoref{prop:spatial-derivative-properties}\ref{item:sd-support} shows that $\Delta(ω'/E_*(σ))$ is $+\infty$ on $(\IdentityMatrix - Q)\HS$, $0$ on $(\IdentityMatrix - P)Q\HS$ and an injective positive operator $A$ on $PQ\HS$, whereas $\Delta(σ/E^{-1}_*(ω'))$ is $+\infty$ on $(\IdentityMatrix - P)\HS$, $0$ on $(\IdentityMatrix - Q)P\HS$ and an injective positive operator $B$ on $PQ\HS$. Moreover $B = A^{-1}$: by \Autoref{lem:haagerup-identity-finite-index} and \Autoref{prop:spatial-derivative-properties}\ref{item:sd-inverse}, applied to the pair $(σ, E^{-1}_*(ω'))$, the generalized positive operators $\Delta(σ/E^{-1}_*(ω'))^{-1}$ and $\Delta(E^{-1}_*(ω')/σ) = \Delta(ω'/E_*(σ))$ coincide on $PQ\HS$. With $Φ(ω', σ) \coloneqq \bra{PQ\v{ρ}} A^{α'} \ket{PQ\v{ρ}} \in [0, \infty]$ and the conventions $0^{t} = \infty^{-t} = 0$, $0^{-t} = \infty^{t} = \infty$ for $t > 0$ and $\infty \cdot 0 = 0$, we therefore obtain
    \begin{align}
        α > 1: \qquad F_α(ω', σ) &= \infty \cdot \norm{(\IdentityMatrix - Q)\v{ρ}}^2 + Φ(ω', σ)\,, & G_α(ω', σ) &= \infty \cdot \norm{(\IdentityMatrix - Q)P\v{ρ}}^2 + Φ(ω', σ)\,, \label{eq:XY-decomposition-large}\\
        α < 1: \qquad F_α(ω', σ) &= \infty \cdot \norm{(\IdentityMatrix - P)Q\v{ρ}}^2 + Φ(ω', σ)\,, & G_α(ω', σ) &= \infty \cdot \norm{(\IdentityMatrix - P)\v{ρ}}^2 + Φ(ω', σ)\,. \label{eq:XY-decomposition-small}
    \end{align}
    Since $\norm{(\IdentityMatrix - Q)\v{ρ}}^2 = \norm{(\IdentityMatrix - Q)P\v{ρ}}^2 + \norm{(\IdentityMatrix - Q)(\IdentityMatrix - P)\v{ρ}}^2$ and $\norm{(\IdentityMatrix - P)\v{ρ}}^2 = \norm{(\IdentityMatrix - P)Q\v{ρ}}^2 + \norm{(\IdentityMatrix - P)(\IdentityMatrix - Q)\v{ρ}}^2$, this proves \ref{item:sb-pointwise}; note that $s(ω') \geq p$ means $P\v{ρ} = \v{ρ}$.

    \ref{item:sb-optimizations} We first show that all four optimizations over $ω'$ can be restricted to the states with $s(ω') \geq p$. For $ω' \in \S(\M')$ and $ε > 0$ put $ω'_ε \coloneqq (ω' + ερ')/(1+ε) \in \S(\M')$. Then $s(ω'_ε) = s(ω') \vee p \geq p$, and by \Autoref{prop:spatial-derivative-properties}\ref{item:sd-linearity} (additivity and homogeneity in the first argument, applied with the roles of $\M$ and $\M'$ interchanged)
    \begin{equation}
        \Delta(ω'_ε/E_*(σ)) = \frac{1}{1+ε}\Big(\Delta(ω'/E_*(σ)) + ε\,\Delta(ρ'/E_*(σ))\Big) \geq \frac{1}{1+ε}\,\Delta(ω'/E_*(σ))
    \end{equation}
    in $\widehat{\BO_+}$, so that \Autoref{lem:compression-inequality}\ref{item:ci-monotone} with $s = \abs{α'} \in (0, 1]$ and the homogeneity of the powers give, for all $σ \in \S(\N)$,
    \begin{equation}\label{eq:sb-mixing}
        F_α(ω'_ε, σ) \geq (1+ε)^{-α'} F_α(ω', σ) \quad \text{if } α > 1\,, \qquad\qquad F_α(ω'_ε, σ) \leq (1+ε)^{-α'} F_α(ω', σ) \quad \text{if } α < 1\,.
    \end{equation}
    Letting $ε \to 0$ in \eqref{eq:sb-mixing} shows that for $α > 1$ the suprema over $ω'$ of $F_α(\cdot, σ)$ and of $\inf_σ F_α(\cdot, σ)$, and for $α < 1$ the infima over $ω'$ of $F_α(\cdot, σ)$ and of $\sup_σ F_α(\cdot, σ)$, do not change when $ω'$ is restricted to the states with $s(ω') \geq p$.

    Now let $α > 1$. For fixed $σ$, either $(\IdentityMatrix - Q)\v{ρ} \neq 0$, in which case $F_α(ω', σ) = \infty$ for all $ω'$ by \eqref{eq:XY-decomposition-large} and also $G_α(ρ', σ) = \infty$ because $(\IdentityMatrix - Q)p\v{ρ} = (\IdentityMatrix - Q)\v{ρ}$; or $(\IdentityMatrix - Q)\v{ρ} = 0$, in which case $F_α(\cdot, σ) = G_α(\cdot, σ)$. Hence $\sup_{ω'} F_α(ω', σ) = \sup_{ω'} G_α(ω', σ)$. Next, $F_α \geq G_α$ gives $\sup_{ω'} \inf_σ F_α \geq \sup_{ω'} \inf_σ G_α$, while by the restriction just proved and \ref{item:sb-pointwise}, $\sup_{ω'} \inf_σ F_α = \sup_{s(ω') \geq p} \inf_σ F_α = \sup_{s(ω') \geq p} \inf_σ G_α \leq \sup_{ω'} \inf_σ G_α$. This proves the claims for $α > 1$; since $F_α = G_α$ on the states with $s(ω') \geq p$, the identities also hold with the optimizations over $ω'$ restricted to these states. For $α < 1$ the argument is the same with all inequalities reversed: $F_α \leq G_α$, both the infimum of $F_α(\cdot, σ)$ and the infimum of $\sup_σ F_α(\cdot, σ)$ over $ω'$ can be restricted to $s(ω') \geq p$, where $F_α = G_α$, and the restricted infima of $G_α$ dominate the unrestricted ones, which in turn dominate those of $F_α$.
\end{proof}

\begin{proposition}[Minimax for sandwiched relative modular operator]\label{prop:sandwiched_minimax}
    Under the standing assumptions of this section, let $α \in [1/2, 1) \cup (1, \infty]$ and $α' = {α - 1 \over α}$ (with $α' = 1$ for $α = \infty$). Then,
    \begin{equation}
        \sup_{ω' \in \S(\M')} \inf_{σ \in \S(\N)} {1 \over α'} \log \bra{\v{ρ}} \Delta(ω'/E_*(σ))^{α'} \ket{\v{ρ}} = 
        \inf_{σ \in \S(\N)} \sup_{ω' \in \S(\M')} {1 \over α'} \log \bra{\v{ρ}} \Delta(ω'/E_*(σ))^{α'} \ket{\v{ρ}}
    \end{equation}
\end{proposition}
\begin{proof}
We first treat the case $α \in [1/2, 1)$, so that $α' \in [-1, 0)$ and the prefactor $1/α'$ is negative; we thus have to prove that
\begin{equation}
    \inf_{ω' \in \S(\M')} \sup_{σ \in \S(\N)} \log \bra{\v{ρ}} \Delta(ω'/E_*(σ))^{α'} \ket{\v{ρ}} = \sup_{σ \in \S(\N)}\inf_{ω' \in \S(\M')}  \log \bra{\v{ρ}} \Delta(ω'/E_*(σ))^{α'} \ket{\v{ρ}}
\end{equation}

Let $X = \S(\M') \subset \S^*(\M') \coloneqq \Set{ω \in (\M'^{*})_+ | ω(\IdentityMatrix) = 1}$, the set of all (not necessarily normal) states of $\M'$, let $Y = \Set{E_*(σ) | σ \in \S(\N)}$, a convex subset of $\M_*^+$ by \eqref{eq:finite-index-preduals}, and let $f: X \times Y \to [0, \infty]$ be defined as
\begin{equation}
    f(x, y) = f_{\v{ρ}}(x, y) = \bra{\v{\rho}} \Delta(x/y)^{α'} \ket{\v{ρ}}
\end{equation}
We want to show that 
\begin{equation}
    \sup_{y \in Y} \inf_{x \in X} f(x,y) = \inf_{x \in X} \sup_{y \in Y} f(x,y)\,.
\end{equation}

By \autoref{lemma:spatial derivate non-normal extension}, $f$ can be extended to a weak-$*$ lower semi-continuous and convex function on $\overline{X} = \S^*(\M')$. Additionally, for every non-normal state $\wp \in \S^*(\M')$ the function value $f(\wp, y)$ is identical to $f(φ_n, y)$ where $φ_n$ is the normal part of $\wp$, which satisfies $φ_n(1) \leq 1$. We claim that there always is a normal state $x \in X$ such that $f(x, y) \leq f(\wp, y)$ for all $y \in Y$. If $φ_n \neq 0$, then $c φ_n \in \S(\M')$ for $c \coloneqq φ_n(1)^{-1} \geq 1$. Since $\Delta(x/y)$ is linear in the first argument and $α' \in [-1, 0)$ we have that
\begin{equation} \bra{\v{\rho}} \Delta(c φ_n/y)^{α'} \ket{\v{ρ}} = c^{α'} f(φ_n, y) = c^{α'} f(\wp, y) \leq f(\wp, y)\,,  \end{equation}
so $x = c φ_n$ has the claimed property. If $φ_n = 0$, i.e.\ if $\wp$ is singular, then every $x \in X$ has it: $0 \leq x$ implies $\Delta(0/y) \leq \Delta(x/y)$ by \Autoref{prop:spatial-derivative-properties}\ref{item:sd-linearity} (with the roles of $\M$ and $\M'$ interchanged), and since $α' < 0$, \Autoref{lem:compression-inequality}\ref{item:ci-monotone} gives $f(x, y) \leq f(0, y) = f(\wp, y)$. (The operator $\Delta(0/y)$ vanishes on $s(y)\HS$ and is $+\infty$ on its orthogonal complement, so that $f(0, y) = +\infty$ unless $s(y)\v{ρ} = 0$.) Hence the optimizations over $x \in X$ or $x \in \overline{X}$ are equivalent (also if there is another optimization over $y \in Y$ inside).

Crucially, $\overline{X}$ is weak-$*$ compact in $(\M')^*$, and 
by \autoref{lemma:spatial derivate non-normal extension} $f$ is weak-$*$ lower semi-continuous in $x$, convex in $x$ and concave in $y$. Thus, we get that
\begin{equation}
      \sup_{y \in Y} \inf_{x \in X} f(x,y) = \sup_{y \in Y} \inf_{x \in \overline{X}} f(x, y) = \inf_{x \in \overline{X}} \sup_{y \in Y} f(x, y) =  \inf_{x \in X} \sup_{y \in Y} f(x, y)
\end{equation}
where the middle equality is \autoref{lem:gen_minimax}.

Finally let $α \in (1, \infty]$, so that $α' \in (0, 1]$ and the claim reads $\sup_{ω'} \inf_σ F_α(ω', σ) = \inf_σ \sup_{ω'} F_α(ω', σ)$ in the notation \eqref{eq:XY}. By \Autoref{lem:support-bookkeeping}\ref{item:sb-optimizations}, $F_α$ may be replaced on both sides by $G_α(ω', σ) = \bra{\v{ρ}} \Delta(σ/E^{-1}_*(ω'))^{-α'} \ket{\v{ρ}}$. The argument above therefore applies identically with the roles of $(\M, \M', ω', E_*(σ))$ and $(\N, \N', σ, E^{-1}_*(ω'))$ interchanged, i.e.\ with $X = \S(\N) \subset \overline{X} = \S^*(\N)$ and $Y = E^{-1}_*(\S(\M'))$ in \autoref{lem:gen_minimax}, and yields $\sup_{ω'} \inf_σ G_α(ω', σ) = \inf_σ \sup_{ω'} G_α(ω', σ)$.
\end{proof}

\begin{theorem}[Sandwiched duality relation]\label{thm:sandwiched_duality_relation}
    Under the standing assumptions of this section, let $α \in [1/2, \infty]$ and ${1 \over α} + {1 \over β} = 2$. Then
    \begin{equation}\label{eq:sandwiched-duality-relation}
        \widetilde{D}_α(ρ|_{\M}\|E_*(\N_*)) = - \widetilde{D}_β(ρ|_{\N'}\|E^{-1}_*(\M'_*))\,.
    \end{equation}
    For $α = β = 1$, where $\widetilde{D}_1 = D$ is the relative entropy (\Autoref{def:divergences}), this is \eqref{eq:optimized-relative-entropy-duality}. Moreover, since $\widetilde{D}_\infty = D_{\max}$ is the max-relative entropy (\Autoref{lem:dmax-variational}) and $\widetilde{D}_{1/2} = -\log F$ is given by the fidelity (\Autoref{def:one-shot}\ref{item:os-fidelity}), the two extreme cases of \eqref{eq:sandwiched-duality-relation} read
    \begin{equation}\label{eq:max-min-duality}
        \begin{aligned}
            D_{\max}(ρ|_{\M}\|E_*(\N_*)) &= - \widetilde{D}_{1/2}(ρ|_{\N'}\|E^{-1}_*(\M'_*))\,, \\
            \widetilde{D}_{1/2}(ρ|_{\M}\|E_*(\N_*)) &= - D_{\max}(ρ|_{\N'}\|E^{-1}_*(\M'_*))\,.
        \end{aligned}
    \end{equation}
\end{theorem}
\begin{proof}
For $α = β = 1$ there is nothing to prove, since \eqref{eq:optimized-relative-entropy-duality} was already established in the previous subsection. So let $α \neq 1$, and let $α' \coloneqq {α - 1 \over α} = {1 - β \over β} \eqqcolon - β'$ (with $α' = 1$, $β' = -1$ for $α = \infty$, and $α' = -1$, $β' = 1$ for $β = \infty$), and let $F_α(ω', σ)$ and $G_α(ω', σ)$ be as in \eqref{eq:XY}, where $ω'$ ranges over $\S(\M')$ and $σ$ over $\S(\N)$. By \eqref{eq:def-sandwiched-renyi} and \Autoref{def:one-shot}\ref{item:os-sets}, $\widetilde{D}_α(ρ|_{\M}\|E_*(\N_*)) = \inf_{σ} \sup_{ω'} {1 \over α'} \log F_α(ω', σ)$, and likewise, applying \eqref{eq:def-sandwiched-renyi} to the algebra $\N'$ (whose commutant is $\N$) with the parameter $β' = -α'$, $\widetilde{D}_β(ρ|_{\N'}\|E^{-1}_*(\M'_*)) = \inf_{ω'} \sup_{σ} {1 \over β'} \log G_α(ω', σ)$. Hence
\begin{align}
\widetilde{D}_α(ρ|_{\M}\|E_*(\N_*)) &= \inf_{σ}\sup_{ω'} {1 \over α'} \log F_α(ω', σ) = \sup_{ω'} \inf_{σ} {1 \over α'} \log F_α(ω', σ) \\
    &= \sup_{ω'} \inf_{σ} {1 \over α'} \log G_α(ω', σ) = - \inf_{ω'} \sup_{σ} {1 \over β'} \log G_α(ω', σ) \\
    &= - \widetilde{D}_{β}(ρ|_{\N'}\|E^{-1}_*(\M'_*))\,,
\end{align}
where the second equality is \Autoref{prop:sandwiched_minimax}, and the third one is \Autoref{lem:support-bookkeeping}\ref{item:sb-optimizations}: for $α > 1$ the function ${1 \over α'} \log$ is increasing, so that $\sup_{ω'} \inf_σ {1 \over α'} \log F_α = {1 \over α'} \log \sup_{ω'} \inf_σ F_α$ and likewise for $G_α$, whereas for $α < 1$ it is decreasing, so that $\sup_{ω'} \inf_σ {1 \over α'} \log F_α = {1 \over α'} \log \inf_{ω'} \sup_σ F_α$ and likewise for $G_α$. The remaining equalities are the identities recalled at the beginning of the proof and $β' = -α'$. Note that the identification of the two expressions for $F_α$ in \eqref{eq:XY} is Haagerup's identity, \Autoref{lem:haagerup-identity-finite-index}. The extreme orders require no separate argument, and in particular no continuity in $α$: all steps above hold for $\abs{α'} = 1$, and \eqref{eq:max-min-duality} is \eqref{eq:sandwiched-duality-relation} for $(α, β) = (\infty, \tfrac12)$ and $(α, β) = (\tfrac12, \infty)$, with $\widetilde{D}_\infty = D_{\max}$ on the algebras $\M$ and $\N'$ by \Autoref{lem:dmax-variational}.
\end{proof}

\subsection{The dual conditional expectation and the additive form}\label{sec:duality-additive-form}

In the duality relations above, the second arguments on the side of the commutants are built from the operator-valued weight $E^{-1}$, which is not a conditional expectation. It is one up to the index, and if the index is a scalar this brings the duality relations into the additive form of \Autoref{cor:finite_index_duality_cond_exp}. For this it is irrelevant whether $\N$ is a factor; what matters is the center of $\M$, in which the index takes its values.

\begin{lemma}[Dual conditional expectation]\label{lem:dual-conditional-expectation}
    Let $E: \M \to \N$ be a normal faithful conditional expectation of finite index. Then $\Ind E \in Z(\M)$ satisfies $\Ind E \geq \IdentityMatrix$, and
    \begin{equation}
        E' \coloneqq (\Ind E)^{-1} E^{-1}: \N' \to \M'
    \end{equation}
    is a normal faithful conditional expectation. If $\M$ is a factor, then $\Ind E$ is a scalar and $E^{-1}_*(ω') = (\Ind E)\, E'_*(ω')$ for all $ω' \in (\M')_*^+$.
\end{lemma}
\begin{proof}
    The index is central by the discussion following \eqref{eq:index-definition}, and $\Ind E \geq \IdentityMatrix$ is the lower bound $T(\IdentityMatrix) \geq \norm{\Ind T}^{-1}\IdentityMatrix$ established in Step~1 of the proof of \Autoref{prop:pythagorean_identity} in \Cref{app:pythagorean-proof}, applied to $T = E^{-1}$, for which $T(\IdentityMatrix) = \Ind E$ and $\Ind T = E(\IdentityMatrix) = \IdentityMatrix$. Hence $h \coloneqq \Ind E$ is invertible and $h^{-1} \in Z(\M) = Z(\M')$. By \Autoref{lem:finite_index_bounded_inverse_weight}, $E^{-1}$ is a bounded positive normal linear map from $\N'$ to $\M'$, and polarizing $E^{-1}(a^* y a) = a^* E^{-1}(y) a$ in $a \in \M'$ shows that it is an $\M'$-bimodule map. For $y \in \N'_+$ the positive operators $h^{-1}$ and $E^{-1}(y)$ commute, so that $E'(y) = h^{-1}E^{-1}(y) \geq 0$; moreover $E'(\IdentityMatrix) = h^{-1}h = \IdentityMatrix$, and $E'(ayb) = h^{-1} a E^{-1}(y) b = a E'(y) b$ for $a, b \in \M'$ because $h^{-1}$ is central in $\M'$. Thus $E'$ is a conditional expectation (\Autoref{def:conditional-expectation}). It is normal because $E^{-1}$ is, and faithful because $E'(y^*y) = 0$ implies $E^{-1}(y^*y) = h E'(y^*y) = 0$ and hence $y = 0$. If $\M$ is a factor, then $Z(\M) = \complex\IdentityMatrix$, so $h$ is a scalar and $ω' \circ E^{-1} = h\, ω' \circ E'$.
\end{proof}

\begin{proof}[Proof of Corollary~\ref{cor:finite_index_duality_cond_exp}]
    Let $λ \coloneqq \Ind E \in [1, \infty)$. As in the proof of \Autoref{lem:dual-conditional-expectation}, $E^{-1}_*(ω') = λ\, E'_*(ω')$ for $ω' \in \S(\M')$, so that the scaling property of \Autoref{prop:divergence-properties}\ref{item:dp-scaling} gives $\D(ρ|_{\N'}\|E^{-1}_*(ω')) = \D(ρ|_{\N'}\|E'_*(ω')) - \log λ$ for $\D \in \set{D, D_α, \widetilde{D}_α}$, including the extreme orders $\widetilde{D}_{1/2}$ and $\widetilde{D}_\infty = D_{\max}$. Taking the infimum over $ω' \in \S(\M')$, respectively choosing $ω' = ρ|_{\M'}$, and inserting the result into \Autoref{thm:sandwiched_duality_relation} and \eqref{eq:optimized-relative-entropy-duality} (for $α = β = 1$), into \Autoref{thm:petz_duality_relation} and into \Autoref{cor:relative_entropy_unoptimized_duality} gives the three parts of the corollary. For type II$_1$ factors and the trace-preserving conditional expectation, $\Ind E$ is the Jones index by \cite{kosaki_extension_1986}.
\end{proof}

\begin{remark}[Non-scalar index]\label{rem:non-scalar-index}
    If $\M$ has a nontrivial center, $\Ind E$ need not be a scalar; the simplest example is a direct sum $\N_1 \oplus \N_2 \subset \M_1 \oplus \M_2$ of two inclusions with conditional expectations of different indices. The duality relations of \Autoref{sec:duality-relations}, which are formulated in terms of $E^{-1}$, are not affected by this, but in their additive form the index correction becomes state dependent. Indeed, with the central element $h \coloneqq \Ind E$ we have $E^{-1}_*(ω') = ω'\big(h E'(\,\cdot\,)\big)$ and therefore $\norm{h^{-1}}^{-1} E'_*(ω') \leq E^{-1}_*(ω') \leq \norm{h}\, E'_*(ω')$ for $ω' \in (\M')_*^+$, so that \Autoref{prop:divergence-properties}\ref{item:dp-scaling} and \ref{item:dp-monotone} give, for $α, β$ as in \Autoref{cor:finite_index_duality_cond_exp},
    \begin{equation}
        \log \norm{(\Ind E)^{-1}}^{-1} \leq \widetilde{D}_α(ρ|_{\M}\|E_*(\Ns)) + \widetilde{D}_β(ρ|_{\N'}\|E'_*(\M'_*)) \leq \log \norm{\Ind E}\,,
    \end{equation}
    and likewise for the unoptimized Petz--\renyi divergences and the relative entropy. For the relative entropy the correction can be computed exactly: the last step of the proof of \Autoref{prop:pythagorean_identity} in \Cref{app:pythagorean-proof}, applied to $T = E^{-1} = hE'$, shows that $D(ρ|_{\N'}\|E^{-1}_*(ρ|_{\M'})) = D(ρ|_{\N'}\|E'_*(ρ|_{\M'})) - ρ(\log \Ind E)$, and hence, by \Autoref{cor:relative_entropy_unoptimized_duality},
    \begin{equation}
        D(ρ|_{\M}\|E_*(ρ|_{\N})) + D(ρ|_{\N'}\|E'_*(ρ|_{\M'})) = ρ(\log \Ind E)\,.
    \end{equation}
\end{remark}

\subsection{Smoothed Entropies}
\newcommand{\rt}{\tilde{ρ}}
We use the smoothed quantities $D_{\max}^ε(ρ\|S)$ and $D_{\min}^ε(ρ\|S)$ of \Autoref{def:one-shot}\ref{item:os-sets}, for $ρ \in \S(\M)$ and $S \subset \M_*^+$, the purified distance $P$ of \Autoref{def:one-shot}\ref{item:os-fidelity}, and the abbreviation $D_{\min} \coloneqq D_{\min}^0 = \widetilde{D}_{1/2}$.

The duality relation for the smoothed quantities is obtained by applying \Autoref{thm:sandwiched_duality_relation} to the purifications of all states $\tilde{ρ}$ in a ball around $ρ$. For this we need that the supremum in Uhlmann's theorem is attained for a fixed purification of $ρ$, which is the case for the vector representative in the natural cone of the standard form (\Autoref{sec:prelim-vn}).

\begin{lemma}[{Uhlmann's theorem with attainment \cite{araki_bures_1972, alberti_bures_2000}}]\label{lem:uhlmann-attainment}
    Let $(\M, \HS, J, \mathcal{P})$ be a standard form, let $ρ \in \S(\M)$, and let $ξ_ρ \in \mathcal{P}$ be the vector representative of $ρ$ in the natural cone. Then every $\tilde{ρ} \in \S(\M)$ has a vector representative $ξ \in \HS$ such that
    \begin{equation}\label{eq:uhlmann-attainment}
        \abs{\ip{ξ_ρ}{ξ}}^2 = F(\tilde{ρ}, ρ)\,.
    \end{equation}
    Consequently, for every $ε \in [0, 1)$,
    \begin{equation}\label{eq:purified-distance-balls}
        \Set{\tilde{ρ} \in \S(\M) | P(\tilde{ρ}, ρ) \leq ε} = \Set{ω_ξ\vert_{\M} | ξ \in \HS,\ \norm{ξ} = 1,\ P(ω_ξ, ω_{ξ_ρ}) \leq ε}\,,
    \end{equation}
    where $ω_ξ$ and $ω_{ξ_ρ}$ are the pure vector states on $\BO$, for which $P(ω_ξ, ω_{ξ_ρ})^2 = 1 - \abs{\ip{ξ_ρ}{ξ}}^2$.
\end{lemma}
\begin{proof}
    The statement is usually formulated for the Bures distance $d(ρ, \tilde{ρ})$, the infimum of $\norm{ξ' - ξ}$ over all vector representatives $ξ'$ of $ρ$ and $ξ$ of $\tilde{ρ}$ in all representations of $\M$, which satisfies $d(ρ, \tilde{ρ})^2 = 2 - 2 F(\tilde{ρ}, ρ)^{1/2}$ \cite{uhlmann_transition_1976}; thus a vector representative $ξ$ of $\tilde{ρ}$ satisfies \eqref{eq:uhlmann-attainment} iff $\norm{ξ_ρ - e^{iθ}ξ} = d(ρ, \tilde{ρ})$ for some phase $θ$. For standard forms with a cyclic and separating vector this is \cite[Thm.~2.7]{alberti_bures_2000}, see also \cite[Lemma~2.15]{alberti_bures_2020}. In general, by \cite[Thm.~3(1)]{araki_bures_1972} a vector representative $ξ$ of $\tilde{ρ}$ with $\norm{ξ_ρ - ξ} = d(ρ, \tilde{ρ})$ exists iff the largest $\tilde{ρ}_2 \in \M_*^+$ with $\tilde{ρ}_2 \leq \tilde{ρ}$ and $s(\tilde{ρ}_2) \perp s(ρ)$ has a vector representative in $(\IdentityMatrix - E_1)\HS$, for a certain projection $E_1 \leq [\M ξ_ρ]$. Since $Jξ = ξ$ and hence $[\Mξ] = J[\M'ξ]J = J s(ω_ξ|_{\M}) J$ for all $ξ \in \mathcal{P}$, the vector representative of $\tilde{ρ}_2$ in $\mathcal{P}$ lies in $J s(\tilde{ρ}_2) J \HS \subset (\IdentityMatrix - [\Mξ_ρ])\HS$. Finally, \eqref{eq:uhlmann-attainment} gives the inclusion \enquote{$\subset$} in \eqref{eq:purified-distance-balls}, and the converse inclusion is the data processing inequality for the fidelity (\Autoref{prop:divergence-properties}\ref{item:dp-dpi}) under the restriction from $\BO$ to $\M$.
\end{proof}

\begin{corollary}[Duality for smoothed entropies]\label{cor:smoothed_duality_relation}
    Under the standing assumptions of this section, let additionally $(\M, \HS, J, \mathcal{P})$ be a standard form and let $\ket{\v{ρ}} = ξ_ρ \in \mathcal{P}$ be the vector representative of $ρ$ in the natural cone. Let $α \in [1/2, 1) \cup (1, \infty]$ and ${1 \over α} + {1 \over β} = 2$. Then, for every $ε \in [0, 1)$,
    \begin{equation}\label{eq:smoothed-duality}
        \inf_{\substack{\tilde{\rho} \in \S(\M) \\ P(\tilde{ρ}, ρ) \leq ε}} \widetilde{D}_α(\tilde{ρ}\|E_*(\N_*)) \geq - \sup_{\substack{\tilde{\rho}' \in \S(\N') \\ P(\tilde{ρ}', ρ|_{\N'}) \leq ε}} \widetilde{D}_β(\tilde{ρ}'\|E^{-1}_*(\M'_*))\,.
    \end{equation}
    In particular, for $α = \infty$ and $β = 1/2$,
    \begin{equation}\label{eq:smoothed-max-min-duality}
        D_{\max}^ε(ρ\|E_*(\N_*)) \geq - D_{\min}^ε(ρ|_{\N'}\|E^{-1}_*(\M'_*))\,.
    \end{equation}
\end{corollary}
\begin{proof}
    For a unit vector $\v{\tilde{ρ}} \in \HS$ write $\tilde{ρ}|_{\mathcal{A}} \coloneqq ω_{\v{\tilde{ρ}}}|_{\mathcal{A}}$ for the induced state on a von Neumann algebra $\mathcal{A} \subset \BO$ (\Autoref{sec:prelim-vn}), and let $ω_{\v{\tilde{ρ}}}, ω_{\v{ρ}}$ denote the (pure) vector states on $\BO$. Then
    \begin{align}
        \inf_{\substack{\tilde{\rho} \in \S(\M) \\ P(\tilde{ρ}, ρ) \leq ε}} \widetilde{D}_α(\tilde{ρ}\|E_*(\N_*)) &= \inf_{\substack{\v{\tilde{\rho}} \in \HS,\ \norm{\v{\tilde{ρ}}} = 1 \\ P(ω_{\v{\tilde{ρ}}}, ω_{\v{ρ}}) \leq ε}} \widetilde{D}_α(\tilde{ρ}|_{\M}\|E_*(\N_*)) \\
        &= \inf_{\substack{\v{\tilde{\rho}} \in \HS,\ \norm{\v{\tilde{ρ}}} = 1 \\ P(ω_{\v{\tilde{ρ}}}, ω_{\v{ρ}}) \leq ε}}  -\widetilde{D}_β(\tilde{ρ}|_{\N'}\|E^{-1}_*(\M'_*)) \\
        &\geq \inf_{\substack{\tilde{\rho}' \in \S(\N') \\ P(\tilde{ρ}', ρ|_{\N'}) \leq ε}} -\widetilde{D}_β(\tilde{ρ}'\|E^{-1}_*(\M'_*))\,.
    \end{align}
    The first equality is \eqref{eq:purified-distance-balls}: by \Autoref{lem:uhlmann-attainment} the states in the $ε$-ball around $ρ$ are exactly the states induced on $\M$ by the unit vectors in the $ε$-ball around $\ket{\v{ρ}} = ξ_ρ$, and the divergence of such a state does not depend on the vector representing it (\Autoref{prop:representation-independence}). The second equality is \autoref{thm:sandwiched_duality_relation}, applied to each of the vectors $\ket{\v{\tilde{ρ}}}$ in place of $\ket{\v{ρ}}$. For the last step note that $P(\tilde{ρ}|_{\N'}, ρ|_{\N'}) \leq P(ω_{\v{\tilde{ρ}}}, ω_{\v{ρ}}) \leq ε$ by the data processing inequality for the fidelity (\Autoref{prop:divergence-properties}\ref{item:dp-dpi} for $\widetilde{D}_{1/2} = -\log F$) under the restriction from $\BO$ to $\N'$, so that the infimum in the last line runs over a larger set of states. This proves \eqref{eq:smoothed-duality}, and \eqref{eq:smoothed-max-min-duality} is the case $α = \infty$, $β = 1/2$, in which the two sides of \eqref{eq:smoothed-duality} are the smoothed quantities of \Autoref{def:one-shot}\ref{item:os-sets}, because $\widetilde{D}_\infty = D_{\max}$ (\Autoref{lem:dmax-variational}).
\end{proof}

\begin{remark}[On the choice of the vector representative]\label{rem:natural-cone-needed}
    The right hand sides of \eqref{eq:smoothed-duality} and \eqref{eq:smoothed-max-min-duality} depend on the vector representative $\ket{\v{ρ}}$ through the induced state $ρ|_{\N'}$. In finite dimensions, and for finite von Neumann algebras in standard form, \Autoref{lem:uhlmann-attainment} holds for every vector representative of $ρ$ in place of $ξ_ρ$ \cite[Rem.~2.2]{alberti_bures_2000}. In general it does not: if $s(ρ) \leq s(\tilde{ρ})$ and $\tilde{ρ}$ has a nonzero part orthogonal to $ρ$, the supremum in Uhlmann's theorem is not attained for any vector representative of $ρ$ that is cyclic for $\M$ \cite[Lemma~2.22]{alberti_bures_2000}, and such examples exist in every von Neumann algebra that is not finite \cite{alberti_bures_2000}.
\end{remark}

For the next section we will also need the following relation between smoothed min and max relative entropy, which is analogous to \cite[Proposition 6.3]{tomamichel_quantum_2016}:
\begin{lemma}\label{lem:smoothed_min_max_trivial_relation}
Let $ρ \in \S(\M)$, let $S \subset \M_*^+$ be any subset, and let $ε, ε' \geq 0$ such that $ε + ε' < 1$. Then,
    \begin{equation}
        D_{\min}^{ε'}(ρ\|S) \leq D_{\max}^ε(ρ\|S) + \log(1 \over 1 - (ε + ε')^2)
    \end{equation}
\end{lemma}
\begin{proof}
    This proof follows the lines of the proof of \cite[Proposition 6.3]{tomamichel_quantum_2016}.
    If $D_{\max}^{ε}(ρ\|S) = +\infty$ there is nothing to prove. Otherwise let $λ > 0$, $σ \in S$ and $\tilde{ρ} \in \S(\M)$ with $P(ρ, \tilde{ρ}) \leq ε$ be such that
    \begin{equation}
        \rt \leq λ\, σ\,.
    \end{equation}
    By \Autoref{def:one-shot}\ref{item:os-dmax} and \ref{item:os-sets}, $D_{\max}^{ε}(ρ\|S)$ is the infimum of $\log λ$ over all such triples $(λ, σ, \rt)$. Let furthermore $\widehat{ρ} \in \S(\M)$ with $P(\widehat{ρ}, ρ) \leq ε'$.
    Since $D_{\min} = \widetilde{D}_{1/2}$ is anti-monotone in its second argument and satisfies $D_{\min}(\widehat{ρ}\|λ^{-1}\rt) = D_{\min}(\widehat{ρ}\|\rt) + \log λ$ (\Autoref{prop:divergence-properties}\ref{item:dp-monotone} and \ref{item:dp-scaling}), we get from $σ \geq λ^{-1}\rt$ that \[\inf_{σ_0 \in S} D_{\min}(\widehat{ρ}\|σ_0) \leq D_{\min}(\widehat{ρ}\|σ) \leq D_{\min}(\widehat{ρ} \| λ^{-1}\rt) = \log λ + D_{\min}(\widehat{ρ}\|\rt)\,.\]
    Now, \[D_{\min}(\widehat{ρ}\|\rt) = -\log F(\widehat{ρ}, \rt) = -\log(1 - P(\widehat{ρ}, \rt)^2) \leq - \log(1 - (ε + ε')^2)\]
    by \Autoref{def:one-shot}\ref{item:os-fidelity} and the triangle inequality for the purified distance, $P(\widehat{ρ}, \rt) \leq P(\widehat{ρ}, ρ) + P(ρ, \rt) \leq ε + ε' < 1$. Hence $\inf_{σ_0 \in S} D_{\min}(\widehat{ρ}\|σ_0) \leq \log λ - \log(1 - (ε + ε')^2)$. Taking the supremum over $\widehat{ρ}$ on the left hand side, which gives $D_{\min}^{ε'}(ρ\|S)$, and then the infimum over $(λ, σ, \rt)$ on the right hand side proves the statement.
\end{proof}

\section{Generalized Quantum Stein's Lemma for Finite-Index Subalgebras}\label{sec:gqsl}

Let $\mathcal{M}$ be a von Neumann algebra and let $\rho,\sigma\in\mathcal{S}(\mathcal{M})$. In i.i.d. quantum hypothesis testing, one is given $n$ copies of an unknown state, which is either the \emph{null hypothesis} $\rho^{\otimes n}$ or the \emph{alternative hypothesis} $\sigma^{\otimes n}$, viewed as normal states on the von Neumann tensor product $\mathcal{M}^{\otimes n}$. A \emph{test} is an element $T\in\mathcal{M}^{\otimes n}$ with $0\le T\le\IdentityMatrix$, i.e.\ a two-outcome measurement $\{T,\IdentityMatrix-T\}$, where the outcome $T$ is interpreted as accepting the null hypothesis. The test can fail in two ways:
\begin{align}
  \alpha_n(T) &:= \rho^{\otimes n}(\IdentityMatrix-T) && \text{(type-I error: rejecting $\rho$ although $\rho$ is true),} \\
  \beta_n(T)  &:= \sigma^{\otimes n}(T)            && \text{(type-II error: accepting $\rho$ although $\sigma$ is true).}
\end{align}
In general the two errors cannot be made small simultaneously at arbitrary rates, and in the \emph{asymmetric} setting one minimizes the type-II error subject to a fixed constraint $\varepsilon\in(0,1)$ on the type-I error,
\begin{equation}\label{eq:beta-simple}
  \beta_n(\varepsilon) := \inf\bigl\{\, \sigma^{\otimes n}(T) \;\big|\; T\in\mathcal{M}^{\otimes n},\ 0\le T\le \IdentityMatrix,\ \rho^{\otimes n}(T)\ge 1-\varepsilon \,\bigr\}
  = \exp\bigl(-D^{\varepsilon}_{H}(\rho^{\otimes n}\,\|\,\sigma^{\otimes n})\bigr),
\end{equation}
where $D^{\varepsilon}_{H}$ is the hypothesis-testing relative entropy of \Autoref{def:one-shot}\ref{item:os-hypothesis-testing}. Quantum Stein's lemma \cite{hiai_proper_1991, ogawa_strong_2000} identifies the optimal exponential decay rate of the type-II error as the Umegaki relative entropy:
\begin{equation}\label{eq:stein}
  \lim_{n\to\infty} -\frac{1}{n}\log\beta_n(\varepsilon)
  = \lim_{n\to\infty}\frac{1}{n} D^{\varepsilon}_{H}(\rho^{\otimes n}\,\|\,\sigma^{\otimes n})
  = D(\rho\|\sigma) \qquad \text{for every } \varepsilon\in(0,1).
\end{equation}
That is, $\beta_n(\varepsilon)$ decays as $\exp(-nD(\rho\|\sigma)+o(n))$, and the asymptotic decay rate $D(\rho\|\sigma)$ does not depend on $\varepsilon$. This gives the relative entropy its most prominent operational interpretation. In finite dimensions, \eqref{eq:stein} is due to Hiai--Petz
(achievability) and Ogawa--Nagaoka (strong converse)
\cite{hiai_proper_1991,ogawa_strong_2000}.
For von Neumann algebraic versions of Stein's lemma, see
\cite[Section~6.5]{jaksic_quantum_2012} and
\cite[Section~4]{pautrat_ke_2023}.
Related asymptotic equipartition results for smooth max-relative
entropy are established in \cite{fawzi_asymptotic_2023}.
For broader background on relative entropy and quantum information
in the von Neumann algebraic setting, see the monograph of
Hollands--Sanders \cite{hollands_sanders_entanglement_2018}
and van Luijk's thesis \cite[Section~4.10]{vanluijk_thesis_2025}.

\subsection{Composite alternatives and the Generalized Quantum Stein's Lemma}

In many applications the alternative hypothesis is not a single state but a whole family of states. For each $n$ let $\mathcal{F}_n\subseteq\mathcal{S}(\mathcal{M}^{\otimes n})$ be a set of states, and consider the the type-II error against \emph{every} state of $\mathcal{F}_n$:
\begin{equation}\label{eq:beta-composite}
  \beta_n(\varepsilon;\mathcal{F}_n) := \inf_{\substack{0\le T\le \IdentityMatrix\\ \rho^{\otimes n}(T)\ge 1-\varepsilon}}\ \sup_{\sigma_n\in\mathcal{F}_n} \sigma_n(T).
\end{equation}
If $\mathcal{F}_n$ is convex, the order of the optimizations in \eqref{eq:beta-composite} can be exchanged: the set of tests is convex and weak-$*$ compact in $\mathcal{M}^{\otimes n}$, and $(T,\sigma_n)\mapsto\sigma_n(T)$ is affine in both arguments and weak-$*$ continuous in $T$ for normal $\sigma_n$. Sion's minimax theorem therefore gives
\begin{equation}\label{eq:beta-composite-DH}
  -\log\beta_n(\varepsilon;\mathcal{F}_n) = \inf_{\sigma_n\in\mathcal{F}_n} D^{\varepsilon}_{H}(\rho^{\otimes n}\,\|\,\sigma_n) = D^{\varepsilon}_{H}(\rho^{\otimes n}\,\|\,\mathcal{F}_n),
\end{equation}
in the notation of \Autoref{def:one-shot}\ref{item:os-sets}.

The prime example comes from quantum resource theories, where $\mathcal{F}_n$ is the set of \emph{free states} on $n$ systems, for example separable states in entanglement theory, incoherent states in coherence theory, and symmetric states in the theory of asymmetry. Brandão and Plenio \cite{brandao_generalization_2010} formulated a set of natural axioms for such families: each $\mathcal{F}_n$ is convex and closed, contains a full-rank state, and the family is closed under partial traces, tensor products and permutations of the subsystems. They stated that under these axioms
\begin{equation}\label{eq:gqsl}
  \lim_{n\to\infty} -\frac{1}{n}\log\beta_n(\varepsilon;\mathcal{F}_n) = D^{\infty}(\rho\,\|\,\mathcal{F}) := \lim_{n\to\infty}\frac{1}{n} D(\rho^{\otimes n}\,\|\,\mathcal{F}_n) \qquad \text{for every } \varepsilon\in(0,1),
\end{equation}
i.e.\ that the optimal Stein exponent is the \emph{regularized relative entropy of resource}. This statement, the \emph{Generalized Quantum Stein's Lemma} (GQSL), is central to resource theories: it implies that the regularized relative entropy of resource governs the asymptotic rates of state conversion under asymptotically resource-non-generating operations, so that these theories are reversible, in analogy with the second law of thermodynamics. A gap in the original proof was identified in \cite{berta_gap_2023}, and \eqref{eq:gqsl} was subsequently established in finite dimensions by Hayashi and Yamasaki \cite{hayashi_generalized_2025} and, independently, by Lami \cite{lami_solution_2025}. All of these works rely in an essential way on finite dimensionality, for example through the method of types, the polynomial number of permutation-invariant states, or dimension-dependent continuity bounds.

\subsection{Subalgebras as alternative hypotheses}

A structurally simpler, but still rich, class of composite alternatives is obtained from a subalgebra. Let $\mathcal{N}\subseteq\mathcal{M}$ be a von Neumann subalgebra with a normal faithful conditional expectation $E\colon\mathcal{M}\to\mathcal{N}$. Then $E^{\otimes n}\colon\mathcal{M}^{\otimes n}\to\mathcal{N}^{\otimes n}$ is again a normal faithful conditional expectation, and we take as free states
\begin{equation}\label{eq:free-subalgebra}
  \mathcal{F}_n := E^{\otimes n}_*\bigl(\mathcal{S}(\mathcal{N}^{\otimes n})\bigr) = \bigl\{\, \sigma\circ E^{\otimes n} \;\big|\; \sigma\in\mathcal{S}(\mathcal{N}^{\otimes n}) \,\bigr\},
\end{equation}
which we abbreviate as $E^{\otimes n}_*(\mathcal{N}^{\otimes n}_*)$ as in \Autoref{def:one-shot}\ref{item:os-sets}. Note that $\sigma$ ranges over \emph{all} normal states of $\mathcal{N}^{\otimes n}$, including correlated ones, and not only over product states; in particular $\mathcal{F}_n$ is convex, so \eqref{eq:beta-composite-DH} applies. In finite dimensions this framework covers several resource theories of interest:
\begin{itemize}
  \item for the fixed-point algebra $\mathcal{N}$ of a group action and the twirl $E$, $\mathcal{F}_1$ consists of the symmetric states (resource theory of asymmetry, Example~\ref{ex:finite-dim-fixed-point});
  \item for the diagonal subalgebra $\mathcal{N}$ and the pinching map $E$, $\mathcal{F}_1$ consists of the incoherent states (resource theory of coherence, Example~\ref{ex:finite-dim-pinching});
  \item for $\mathcal{N}=\IdentityMatrix_A\otimes\mathcal{B}(\mathcal{H}_B)$ and $E=\frac{\IdentityMatrix_A}{d_A}\otimes\operatorname{Tr}_A$, $\mathcal{F}_1=\{\IdentityMatrix_A/d_A\otimes\sigma_B\}$, and the exponent is $D(\rho_{AB}\|E_*(\mathcal{N}_*)) = \log d_A - H(A|B)_\rho$. This gives the conditional entropy an operational meaning as a Stein exponent.
\end{itemize}
For such alternatives, Gao and Rahaman \cite{gao_generalized_2024} proved the GQSL in finite dimensions by a route independent of \cite{hayashi_generalized_2025, lami_solution_2025}. They also showed that the regularization in \eqref{eq:gqsl} collapses: the Stein exponent is the single-letter quantity $D(\rho\|E_*(\mathcal{N}_*)) = D(\rho\|E_*(\rho|_{\mathcal{N}}))$, where the second equality is the Pythagorean identity (Proposition~\ref{prop:pythagorean_identity}).

In this section we lift this result to arbitrary von Neumann subalgebras, including those of type II and type III, where none of the finite-dimensional tools mentioned above are available. The only structural assumption is that the subalgebra has finite index. The strategy is as follows. We prove an asymptotic equipartition property for the smoothed max-relative entropy with respect to the subalgebra, using the duality relations and the single-state asymptotic equipartition property of \cite{fawzi_asymptotic_2023}; the tight relations between $D^{\varepsilon}_{H}$ and $D_{\max}^\varepsilon$ then imply the GQSL for subalgebras. The following is the proof of \Autoref{thm:gqsl_subalgebras} from the introduction.

%The upper bound (achievability) follows from the single-state case \eqref{eq:stein}, applied to $\sigma=E_*(\rho|_{\mathcal{N}})\in\mathcal{F}_1$. For the lower bound (optimality), we use the duality relations of Section~\ref{sec:duality}, in the smoothed form of Corollary~\ref{cor:smoothed-duality}. They convert a lower bound on the smoothed max-relative entropy with respect to the composite set $E^{\otimes n}_*(\mathcal{N}^{\otimes n}_*)$ into an upper bound on a smoothed min-relative entropy on the commutant $(\mathcal{N}')^{\otimes n}$, with respect to the \emph{single} functional $(E^{-1}_*(\rho|_{\mathcal{M}'}))^{\otimes n}$. The finite-index assumption guarantees that $E^{-1}$ is bounded (Lemma~\ref{lem:finite-index-bounded}), so this functional is a positive normal functional. The single-state asymptotic equipartition property of \cite{FGR23} then applies on the commutant, and the relative-entropy duality of Corollary~\ref{cor:relent-duality} shows that the resulting upper and lower bounds coincide. Our main result is the following.
\begin{proof}[Proof of Theorem~\ref{thm:gqsl_subalgebras}]
    Neither the quantities in the statement nor the index of $E$ depend on the representation of $\M$, so we may assume that $\M \subset \BO$ is in standard form $(\M, \HS, J, \mathcal{P})$ (\Autoref{sec:prelim-vn}). From the definition it is clear that
    \begin{equation}
        D_{\max}^ε(ρ\n\|E\n_*(\N\n_*)) \leq D_{\max}^ε(ρ\n\|E\n_*(ρ\n|_{\N\n}))\,.
    \end{equation}
    Let $σ \coloneqq E_*(ρ|_{\N}) \in \S(\M)$ and $σ' \coloneqq E^{-1}_*(ρ|_{\M'}) \in (\N')_*^+$ (\Autoref{lem:finite_index_bounded_inverse_weight}). The asymptotic equipartition property of \cite[Thm.~3.1]{fawzi_asymptotic_2023} requires the relative entropy and the second and third moments of $\log λ$ with respect to the spectral measure $μ$ of \Autoref{lem:first-moment} to be finite. This is true because we assumed finite index: by \Autoref{lem:haagerup-identity-finite-index} and \eqref{eq:first-moment}, $$\bra{\v{ρ}} \Delta(ω'/σ) \ket{\v{ρ}} = \bra{\v{ρ}} \Delta(E^{-1}_*(ω')/ρ|_{\N}) \ket{\v{ρ}} = E^{-1}_*(ω')(s(ρ|_{\N'})) \leq ω'(\Ind E) \leq \norm{\Ind E}$$ for every $ω' \in \S(\M')$, so that $ρ \leq \norm{\Ind E}\, σ$ by \Autoref{lem:dmax-variational}, and likewise $\bra{\v{ρ}} \Delta(τ/σ') \ket{\v{ρ}} = E_*(τ)(s(ρ)) \leq 1$ for $τ \in \S(\N)$ gives $ρ|_{\N'} \leq σ'$. Hence $D_{\max}(ρ\|σ) \leq \log\norm{\Ind E}$ and $D_{\max}(ρ|_{\N'}\|σ') \leq 0$, which by \eqref{eq:alt-inequalities} bounds $\int λ^{β - 1}\,\dd μ = e^{(β - 1) D_β}$ for $β \in (1, 2)$, while $\int λ^{-1}\,\dd μ < \infty$ by \eqref{eq:first-moment}; thus $\log λ$ has finite exponential moments of both signs, and all its moments are finite. Moreover, \cite{fawzi_asymptotic_2023} smooth over sub-normalized states, whereas we smooth over states. This does not affect the limits: a sub-normalized $\tilde{ρ}$ with $P(\tilde{ρ}, ρ) \leq ε$ has $\tilde{ρ}(\IdentityMatrix) \geq F(\tilde{ρ}, ρ) \geq 1 - ε^2$, so that $\tilde{ρ}/\tilde{ρ}(\IdentityMatrix)$ is again $ε$-close to $ρ$ and its $D_{\max}$ is larger by at most $\log \frac{1}{1 - ε^2}$; hence the two smoothed max-relative entropies differ by a constant. Therefore
    \begin{equation}
        \lim_{n \to \infty} {1 \over n} D_{\max}^ε(ρ\n\|σ\n) = D(ρ\|σ) = D(ρ\|E_*(ρ|_{\N}))\,.
    \end{equation}
    By \autoref{prop:pythagorean_identity} we also have that $D(ρ\|E_*(ρ|_{\N})) = D(ρ\|E_*(\N_*))$, so that
    \begin{equation}
        \limsup_{n \to \infty} {1 \over n} D_{\max}^ε(ρ\n\|E\n_*(\N\n_*)) \leq D(ρ\|E_*(\N_*)) = D(ρ\|E_*(ρ|_{\N}))\,.
    \end{equation}

    For the lower bound we pass to the commutants, and for this we fix the vector representative $\ket{\v{ρ}} \coloneqq ξ_ρ \in \mathcal{P}$ of $ρ$ in the natural cone. The algebras $\N\n \subset \M\n$ act on $\HS\n$, the vector $\ket{\v{ρ}}\n$ is a vector representative of $ρ\n$, and by the commutation theorem for tensor products \cite[Sec.~1.3]{hiai_lectures_2021} the commutants are $(\M\n)' = (\M')\n$ and $(\N\n)' = (\N')\n$. Hence the states induced by $\ket{\v{ρ}}\n$ on the two commutants are $(ρ|_{\M'})\n$ and $(ρ|_{\N'})\n$. We use two more facts about tensor powers. First, $\M\n \subset \BO[\HS\n]$ is again in standard form, with the antiunitary involution $J\n$ and a natural cone that contains the tensor products of vectors in $\mathcal{P}$; in particular, $\ket{\v{ρ}}\n = ξ_ρ\n$ is the vector representative of $ρ\n$ in the natural cone of this standard form, see e.g.\ \cite[Thm.~3.9]{okayasu_haagerup_2015}. Secondly, $E\n$ has finite index, with dual operator-valued weight $(E\n)^{-1} = (E^{-1})\n$ and index $\Ind E\n = (\Ind E)\n$, see e.g.\ \cite[Prop.~1.7]{hiai_minimum_1990}. Therefore the inclusion $\N\n \subset \M\n$, the conditional expectation $E\n$ and the state $ρ\n$ with its vector representative $\ket{\v{ρ}}\n$ satisfy the assumptions of \autoref{cor:smoothed_duality_relation}. We abbreviate
    \begin{equation}
        σ' \coloneqq E^{-1}_*(ρ|_{\M'}) \in (\N')_*^+\,,
    \end{equation}
    which is a positive normal functional by \autoref{lem:finite_index_bounded_inverse_weight}, because $E$ has finite index. Then $(σ')\n = (E\n)^{-1}_*\big((ρ|_{\M'})\n\big)$ is one of the functionals in the set $(E\n)^{-1}_*\big((\M')\n_*\big)$ to which the duality relation for $E\n$ refers. Now let $ε' > 0$ be such that $ε + ε' < 1$. Then
    \begin{align}
         D_{\max}^ε(ρ\n\|E\n_*(\N\n_*)) &\geq - D_{\min}^ε\Big((ρ|_{\N'})\n\Big\|(E\n)^{-1}_*\big((\M')\n_*\big)\Big) \\
         &\geq - D_{\min}^ε\big((ρ|_{\N'})\n\big\|(σ')\n\big)\\
         &\geq - D_{\max}^{ε'}\big((ρ|_{\N'})\n\big\|(σ')\n\big) - \log{1 \over 1 - (ε + ε')^2}\,.
    \end{align}
    Here the first inequality is \eqref{eq:smoothed-max-min-duality} of \autoref{cor:smoothed_duality_relation}, applied to the inclusion $\N\n \subset \M\n$, the conditional expectation $E\n$ and the state $ρ\n$. The second one holds because the infimum over the second argument in the definition of $D_{\min}^ε$ (\Autoref{def:one-shot}\ref{item:os-sets}) can only increase if the set is replaced by one of its elements. The last one is \autoref{lem:smoothed_min_max_trivial_relation} for the algebra $(\N')\n$ and the set $S = \set{(σ')\n}$, with the roles of $ε$ and $ε'$ interchanged. Note that the smoothing parameter of the max-relative entropy on the side of the commutants is $ε'$ and not $ε$.

    Since the asymptotic equipartition property of \cite{fawzi_asymptotic_2023} holds for every smoothing parameter in $(0, 1)$, we can use it again, now on $\N'$ and with the parameter $ε'$:
    \begin{equation}
        \lim_{n \to \infty} {1 \over n} D_{\max}^{ε'}\big((ρ|_{\N'})\n\big\|(σ')\n\big) = D(ρ|_{\N'}\|σ')\,.
    \end{equation}
    Dividing the chain of inequalities above by $n$ and letting $n \to \infty$, where the correction term $\log{1 \over 1 - (ε + ε')^2}$ does not depend on $n$, we obtain
    \begin{equation}
        \begin{aligned}
            \liminf_{n \to \infty} {1 \over n} D_{\max}^ε(ρ\n\|E\n_*(\N\n_*)) &\geq - D(ρ|_{\N'}\|σ') = - D(ρ|_{\N'}\|E^{-1}_*(ρ|_{\M'})) \\
            &= D(ρ\|E_*(ρ|_{\N}))\,,
        \end{aligned}
    \end{equation}
    where the last equality is \autoref{cor:relative_entropy_unoptimized_duality}. Together with the upper bound this proves the first claim.

    For $D_H^\varepsilon$ we use the relation \cite[Lemma~3.5]{fawzi_asymptotic_2023} (with $\varepsilon$ replaced by $1 - \varepsilon$) between $D_H^\varepsilon$ and $D_{\max}^\varepsilon$: for a state $\rho$, a positive normal functional $\sigma$ and $\delta \in (0, 1 - \varepsilon)$ with $D_{\max}^{\sqrt{1-\varepsilon-\delta}}(\rho\|\sigma) < \infty$,
    \begin{align}
 D_{\max}^{\sqrt{1-\varepsilon}}(\rho\|\sigma )
 -\log\frac{1}{\varepsilon}
 \le D_H^\varepsilon(\rho\|\sigma )\le
 D_{\max}^{\sqrt{1-\varepsilon-\delta}}(\rho\|\sigma )
 +\log\frac{4(\varepsilon+\delta)}{\delta^2}\,.
 \label{eq:aep-dh-upper-pair}
\end{align}
    Both inequalities pass to the infimum over $σ_n \in E\n_*(\N\n_*)$. For the first one, a $σ_n$ violating the finiteness condition may be replaced by $(1-t)σ_n + tσ\n$ with $t \in (0, 1)$, which lies in the same convex set, satisfies $D_{\max}\big(ρ\n\big\|(1-t)σ_n + tσ\n\big) \leq n\log\norm{\Ind E} - \log t < \infty$ because $ρ\n \leq \norm{\Ind E}^n σ\n$, and for which $D_H^ε(ρ\n\|\,\cdot\,)$ is larger than for $σ_n$ by at most $-\log(1-t)$; letting $t \to 0$ gives the first inequality for the infima. Since the limit of $\frac{1}{n} D_{\max}^ε$ does not depend on $ε$, we obtain
    \begin{equation*}
         \lim_{n \to \infty} {1 \over n} D_{\max}^ε(ρ\n\|E\n_*(\N\n_*)) = \lim_{n \to \infty} {1 \over n} D_{H}^ε(ρ\n\|E\n_*(\N\n_*))  = D(ρ\|E_*(ρ|_{\N}))=D(ρ\|E_*(\N_*))\,.
    \end{equation*}
\end{proof}
\noindent {\bf Acknowledgments. }
The authors thank Anna Jenčová for making them aware of \autoref{lem:resolvent-variational-formula} and Lauritz van Luijk as well as Fumio Hiai for helpful discussions. BB is supported by the Engineering and Physical Sciences Research Council [Grant Ref: EP/Y028732/1]. LG is partially supported by the National Natural Science Foundation of China (grant No.~12401163) and the Department of Science and Technology of Hubei Province (Projects No.2025EHA041, No.2025AFA044).  MR was partially supported by
the Wallenberg Centre for Quantum Technology (WACQT) funded by the Knut and Alice
Wallenberg Foundation (KAW).

\noindent  \textbf{Statement on AI usage.} 
OpenAI GPT 5.6 Sol has provided ideas for the proofs of \autoref{lem:singular functional properties}, \autoref{lemma:spatial derivate non-normal extension}. Additionally, Claude Fable 5.1 was used to add explanations, improve the exposition of the paper, and write out some of the standard arguments used in the Appendix. All arguments have been verified by the authors.
\appendix
\section{Proofs and Lemmas}\label{app:prelim-proofs}

This appendix collects auxiliary lemmas and the proofs deferred from the main text; we use the notation and conventions of \Autoref{sec:prelim-vn}. Most of this material is standard, or a routine adaptation of known results to the conventions of this paper; where a statement can be found in the literature, we give the reference and only indicate the necessary modifications. \Cref{app:generalized-positive-operators} contains two standard facts about positive operators, which are used in \Autoref{sec:prelim-divergences} and \Autoref{sec:duality-relations}. \Cref{app:nonfaithful-proofs} collects the proofs that do little more than extend known results about spatial derivatives from faithful to non-faithful functionals: the reduction to the support of the second argument (\Autoref{lem:bounded-vectors-nonfaithful}), the proof of \Autoref{prop:spatial-derivative-properties}, and the proof of Haagerup's identity for non-faithful functionals (\Autoref{lem:haagerup-identity-finite-index}). \Cref{app:divergence-proofs} contains the proofs of the statements about the entropic quantities in \Autoref{sec:prelim-divergences} (\Autoref{lem:relative-entropy-well-defined}, \Autoref{lem:sandwiched-support-restriction}, \Autoref{prop:representation-independence}, \Autoref{prop:divergence-properties} and \Autoref{lem:dmax-variational}), together with a lemma on the spectral measures entering their definition (\Autoref{lem:first-moment}), \Cref{app:finite-index-proof} the proof of \Autoref{lem:finite_index_bounded_inverse_weight}, \Cref{app:finite-dim-proofs} the finite-dimensional computations behind \Autoref{prop:finite-dim-cond-exp} and \Autoref{prop:finite-dim-inverse-weight}, and \Cref{app:pythagorean-proof} the proof of the Pythagorean identity for operator-valued weights (\Autoref{prop:pythagorean_identity}), by reduction to the known case of a conditional expectation.

\subsection{Standard inequalities for generalized positive operators}\label[subappendix]{app:generalized-positive-operators}

The two lemmas of this subsection are the workhorse for everything concerning the optimization in the sandwiched divergence: a variational expression for the resolvent of a positive operator (\Autoref{lem:resolvent-variational-formula}) and its consequences for operator monotonicity, compressions and concavity (\Autoref{lem:compression-inequality}). Both are well known. For the first one we refer to the literature. For the second one we give a complete proof, because we need it as a statement about quadratic forms with values in $[0, \infty]$, for unbounded and for generalized positive operators, for negative powers and without any assumptions on domains, whereas the versions in the literature are mostly formulated for bounded operators.

The first lemma expresses the resolvent as the Legendre transform of the quadratic form of $t + A$. Substituting $h = tw$, $g = v - h$ and using $\ip{v}{A(t+A)^{-1}v} = \norm{v}^2 - t\ip{v}{(t+A)^{-1}v}$, one sees that it is equivalent to the identity
\begin{equation}\label{eq:resolvent-infimum-formula}
    \ip{v}{A(t + A)^{-1} v} = \inf\Set{\norm{g}^2 + t^{-1}\norm{A^{1/2}h}^2 | g + h = v,\ h \in \mathcal{D}(A^{1/2})}\,,
\end{equation}
which is the starting point of Kosaki's variational expression for the relative entropy \cite{kosaki_relative_1986}; see \cite[Lemma~5.8]{ohya_quantum_1993}, where \eqref{eq:resolvent-infimum-formula} is proved for $v$, $g$ and $h$ in a core of $A^{1/2}$, and, for generalized positive operators and $t = 1$, \cite[p.~400]{kosaki_parallel_2017} and \cite[p.~228]{hiai_lectures_2021}.

\begin{lemma}[Resolvent variational formula]\label{lem:resolvent-variational-formula}
    Let $\HS$ be a Hilbert space with vector $v \in \HS$. For any (potentially unbounded) positive self-adjoint operator $A \geq 0$ on $\HS$, and any $t > 0$, the following identity holds
    \begin{equation}
        \langle{v, (t + A)^{-1} v\rangle} = \sup_{w \in \mathcal{D}(A^{1/2})} 2 \Re\langle w, v\rangle - t \norm{w}^2 - \|A^{1/2} w\|^2
    \end{equation}
    where $\mathcal{D}(A^{1/2})$ is the domain of $A^{1/2}$.
\end{lemma}
\begin{proof}
    This is a standard fact about closed quadratic forms, see e.g.\ \cite[Thm.~10.9]{schmudgen_unbounded_2012}, applied to the form $w \mapsto \norm{A^{1/2}w}^2$ on $\mathcal{D}(A^{1/2})$. Explicitly, with $w_0 \coloneqq (t + A)^{-1}v \in \mathcal{D}(A)$ the expression in the supremum equals $\ip{v}{(t + A)^{-1}v} - \norm{(t + A)^{1/2}(w - w_0)}^2$, so that the supremum is attained exactly at $w = w_0$.
\end{proof}

For bounded operators the three parts of the next lemma are classical results from the theory of operator monotone and operator convex functions: part~\ref{item:ci-monotone} is the Löwner--Heinz inequality, together with the fact that inversion reverses the order, part~\ref{item:ci-compression} is the operator Jensen inequality for compressions \cite{hansen_jensens_1982}, and part~\ref{item:ci-concavity} is the operator concavity of $t \mapsto t^{s}$ and the operator convexity of $t \mapsto t^{-s}$ for $s \in (0, 1]$; see \cite[Ch.~V]{bhatia_matrix_1997} or \cite{carlen_trace_2010, simon_loewners_2019a}. For unbounded positive self-adjoint operators, the equivalence between the order of the quadratic forms and the reversed order of the resolvents, on which part~\ref{item:ci-monotone} rests, can be found in \cite[Thm.~VI.2.21]{kato_perturbation_1995}, \cite[Thm.~7.5.14]{simon_operator_2015} and, for generalized positive operators, in \cite[Lemma~A.1, Prop.~A.18]{hiai_lectures_2021}; a version of the operator Jensen inequality for unbounded operators is \cite[Lemma~20]{berta_renyi_2018}.

\begin{lemma}[Operator monotonicity, compressions and concavity of powers]\label{lem:compression-inequality}
    Let $A, B \in \widehat{\BO_+}$, let $q$ be a projection on $\HS$, and let $A_q \in \widehat{\BO[q\HS]_+}$ be the compression of $A$ to $q\HS$, i.e.\ the generalized positive operator on $q\HS$ whose form is $ζ \mapsto \ip{ζ}{Aζ}$, $ζ \in q\HS$. Then for all $s \in (0, 1]$:
    \begin{enumerate}[label=(\roman*)]
        \item\label{item:ci-monotone} If $A \leq B$, then $\ip{ζ}{A^{s}ζ} \leq \ip{ζ}{B^{s}ζ}$ and $\ip{ζ}{A^{-s}ζ} \geq \ip{ζ}{B^{-s}ζ}$ for all $ζ \in \HS$.
        \item\label{item:ci-compression} $\ip{ζ}{A_q^{s}ζ} \geq \ip{ζ}{A^{s}ζ}$ and $\ip{ζ}{A_q^{-s}ζ} \leq \ip{ζ}{A^{-s}ζ}$ for all $ζ \in q\HS$.
        \item\label{item:ci-concavity} For $μ \in [0, 1]$ let $μA + (1 - μ)B \in \widehat{\BO_+}$ be the generalized positive operator with form $ζ \mapsto μ\ip{ζ}{Aζ} + (1-μ)\ip{ζ}{Bζ}$. Then, for all $ζ \in \HS$,
        \begin{align}
            \ip{ζ}{(μA + (1-μ)B)^{s}ζ} &\geq μ\ip{ζ}{A^{s}ζ} + (1-μ)\ip{ζ}{B^{s}ζ}\, \\ 
            \ip{ζ}{(μA + (1-μ)B)^{-s}ζ} &\leq μ\ip{ζ}{A^{-s}ζ} + (1-μ)\ip{ζ}{B^{-s}ζ}\,,
        \end{align}
        i.e.\ $C \mapsto \ip{ζ}{C^{s}ζ}$ is concave and $C \mapsto \ip{ζ}{C^{-s}ζ}$ is convex on $\widehat{\BO_+}$.
    \end{enumerate}
\end{lemma}
\begin{proof}
    We follow the standard route (see e.g.\ \cite[Ch.~V]{bhatia_matrix_1997}): the integral representations \eqref{eq:power-integral-representations} of the power functions reduce all three statements to the corresponding statements about resolvents, and these can be read off from the variational formula of \Autoref{lem:resolvent-variational-formula}. For $C \in \widehat{\BO_+}$ and $λ > 0$, the generalized positive operator $λ + C$ has the same $+\infty$ part as $C$, so $(λ + C)^{-1}$ is a bounded operator, and \Autoref{lem:resolvent-variational-formula} (applied to the finite part of $C$) gives
    \begin{equation}\label{eq:resolvent-variational-generalized}
        \ip{ζ}{(λ + C)^{-1} ζ} = \sup\Set{ 2\Re\ip{w}{ζ} - λ\norm{w}^2 - \ip{w}{Cw} | w \in \HS,\ \ip{w}{Cw} < \infty}\,, \qquad{ζ} \in \HS\,.
    \end{equation}
    If $A \leq B$, the supremum for $B$ runs over a smaller set and has a smaller objective than the one for $A$, so $\ip{ζ}{(λ + B)^{-1}ζ} \leq \ip{ζ}{(λ + A)^{-1}ζ}$. Likewise, for $ζ \in q\HS$ the supremum for $A_q$ runs over $\Set{w \in q\HS | \ip{w}{Aw} < \infty}$ with the same objective as the one for $A$, so $\ip{ζ}{(λ + A_q)^{-1}ζ} \leq \ip{ζ}{(λ + A)^{-1}ζ}$. Now recall that for $s \in (0, 1)$ and $t \in [0, \infty]$, with our conventions,
    \begin{equation}\label{eq:power-integral-representations}
        t^{-s} = \frac{\sin(π s)}{π} \int_0^\infty{λ}^{-s} (λ + t)^{-1} \, \dd{λ}\,, \qquad t^{s} = \frac{\sin(π s)}{π} \int_0^\infty{λ}^{s - 1} \big(1 - λ(λ + t)^{-1}\big) \, \dd{λ}\,.
    \end{equation}
    Integrating against the spectral measure $μ^C_ζ$ and using Tonelli's theorem (all integrands are nonnegative) yields
    \begin{align}
        \ip{ζ}{C^{-s}ζ} &= \frac{\sin(π s)}{π} \int_0^\infty{λ}^{-s} \ip{ζ}{(λ+C)^{-1}ζ} \, \dd{λ}\, \\
        \ip{ζ}{C^{s}ζ} &= \frac{\sin(π s)}{π} \int_0^\infty{λ}^{s-1} \big(\norm{ζ}^2 - λ\ip{ζ}{(λ+C)^{-1}ζ}\big) \, \dd{λ}
    \end{align}
    for every $C \in \widehat{\BO_+}$, and both claims follow from the resolvent inequalities for $s \in (0,1)$. For $s = 1$, $\ip{ζ}{C^{-1}ζ} = \lim_{λ \downarrow 0} \ip{ζ}{(λ + C)^{-1}ζ}$ by monotone convergence, which gives the statements about $A^{-1}$, while $\ip{ζ}{A_qζ} = \ip{ζ}{Aζ}$ for $ζ \in q\HS$ by definition, and $A \leq B$ is the assumption.

    For \ref{item:ci-concavity} note first that $μ\ip{\,\cdot\,}{A\,\cdot\,} + (1-μ)\ip{\,\cdot\,}{B\,\cdot\,}$ is a lower semicontinuous positive form and hence defines an element of $\widehat{\BO_+}$ (\Autoref{sec:prelim-vn}). For $λ > 0$ and $μ \in (0,1)$, the objective function in \eqref{eq:resolvent-variational-generalized} for $C = μA + (1-μ)B$ is the convex combination, with weights $μ$ and $1 - μ$, of the objective functions for $A$ and for $B$, and its supremum runs over the vectors $w$ with $\ip{w}{Aw} < \infty$ and $\ip{w}{Bw} < \infty$, a subset of the sets for $A$ and for $B$; hence $\ip{ζ}{(λ + C)^{-1}ζ} \leq μ\ip{ζ}{(λ + A)^{-1}ζ} + (1-μ)\ip{ζ}{(λ + B)^{-1}ζ}$, i.e.\ $C \mapsto \ip{ζ}{(λ + C)^{-1}ζ}$ is convex on $\widehat{\BO_+}$. By \eqref{eq:power-integral-representations} and Tonelli's theorem, $C \mapsto \ip{ζ}{C^{-s}ζ}$ is then convex and $C \mapsto \ip{ζ}{C^{s}ζ}$ is concave for $s \in (0, 1)$; for $s = 1$, $\ip{ζ}{Cζ}$ is affine by definition and $\ip{ζ}{C^{-1}ζ} = \lim_{λ \downarrow 0} \ip{ζ}{(λ + C)^{-1}ζ}$ is a pointwise limit of convex functions.
\end{proof}

\subsection{From faithful to non-faithful functionals}\label[subappendix]{app:nonfaithful-proofs}

Connes' spatial derivative and the results about it that we use, namely its basic properties (\Autoref{prop:spatial-derivative-properties}) and Haagerup's identity \eqref{eq:inverse_weight}, are proved in \cite{connes_spatial_1980, haagerup_operator_1979a, hiai_lectures_2021} for a second argument $ψ$ that is a \emph{faithful} normal semifinite weight; only the first argument is allowed to be non-faithful \cite[Lemma~11.3, Prop.~11.19]{hiai_lectures_2021}. For positive normal functionals that are not faithful, the spatial derivative is introduced in \cite[Ch.~4]{ohya_quantum_1993} and its basic properties are stated in \cite[Eqs.~(4.4)--(4.7)]{ohya_quantum_1993}, without proofs and with the value $0$ instead of $+\infty$ on the orthogonal complement of the support of the second argument. The purpose of this subsection is to show how these statements, in the convention of \Autoref{def:spatial-derivative}, follow from the faithful case. The arguments contain no new ideas and use two standard devices.
\begin{itemize}
    \item \emph{Reduction to the support.} One cuts down by the support projection $e = s(ψ) \in \M'$, i.e.\ one passes to the reduced algebra $e\M'e$ and the induced algebra $\M e$ on $e\HS$, where $ψ$ is faithful, and applies the known results there. This is the mechanism of Hiai's reduction lemma \cite[Lemma~11.24]{hiai_lectures_2021}, which treats the support of the \emph{first} argument in this way, and of the support conventions for the relative modular operators of non-faithful functionals in the standard form \cite{araki_relative_1977}, \cite[App.~C]{araki_positive_1982}, \cite[Sec.~10.1]{hiai_lectures_2021}. \Autoref{lem:bounded-vectors-nonfaithful} in \Cref{app:bounded-vectors-proof} and the proof of \Autoref{prop:spatial-derivative-properties} in \Cref{app:spatial-derivative-general-proof} are of this kind.
    \item \emph{Completion.} One makes a non-faithful functional faithful by adding a weight that is supported on the complement of its support, applies the known result, and removes the additional term again by restricting to the original support; see \cite[Eq.~(4.6)]{ohya_quantum_1993} and the proofs of \cite[Thm.~5.7]{ohya_quantum_1993} and \cite[Thm.~11.7(4)]{hiai_lectures_2021}. The proof of Haagerup's identity for non-faithful functionals in \Cref{app:haagerup-proof} is of this kind.
\end{itemize}
In finite dimensions, where $\Delta(φ/ψ) = \bigoplus_i \hat{φ}_i \otimes \hat{ψ}_i^{-1}$ by \eqref{eq:finite-dim-spatial-derivative}, all statements of this subsection are obvious if $\hat{ψ}_i^{-1}$ is read as the generalized inverse, with the value $+\infty$ on the kernel of $\hat{ψ}_i$.

\subsubsection{Reduction to the support of the second argument}\label[subsubappendix]{app:bounded-vectors-proof}

Let $ψ \in (\M')_*^+$ and put $e \coloneqq s(ψ) \in \M'$, $\mathcal{K} \coloneqq e\HS$. The reduced algebra $e\M'e$, restricted to $\mathcal{K}$, is a von Neumann algebra on $\mathcal{K}$ on which $ψ_e \coloneqq{ψ}|_{e\M'e}$ is faithful, and its commutant on $\mathcal{K}$ is the induced algebra $\M e|_{\mathcal{K}} = \{xe|_{\mathcal{K}}\, :\, x \in \M\}$ (see \cite[Lemma~11.24]{hiai_lectures_2021} and \cite[Ch.~II]{takesaki_theory_1979}). The map $x \mapsto xe|_{\mathcal{K}}$ is a normal $*$-homomorphism of $\M$ onto $\M e|_{\mathcal{K}}$ whose kernel is $\M(\IdentityMatrix - z)$, where $z \in Z(\M) = Z(\M')$ is the central support of $e$, i.e.\ the smallest central projection dominating $e$. Consequently, every $φ \in \M_*^+$ induces a well-defined element of $(\M e|_{\mathcal{K}})_*^+$ by
\begin{equation}\label{eq:induced-functional}
    φ_e(xe|_{\mathcal{K}}) \coloneqq{φ}(xz)\,, \qquad x \in \M\,.
\end{equation}
Note that $φ_e$ is \emph{not} the compression $x \mapsto{φ}(exe)$, which is not even well defined on $\M e|_{\mathcal{K}}$ in general; when $\M$ is a factor, $z = \IdentityMatrix$ and $φ_e$ is simply $φ$ viewed as a functional on $\M e \cong \M$. The following lemma shows that the spatial derivative with respect to $ψ$ is, on $\mathcal{K}$, Connes' spatial derivative with respect to the faithful functional $ψ_e$. We write $\Delta_{\mathcal{K}}(\,\cdot\,/\,\cdot\,)$ for spatial derivatives on the Hilbert space $\mathcal{K}$.

\begin{lemma}[Reduction to the support]\label{lem:bounded-vectors-nonfaithful}
    Let $ψ \in (\M')_*^+$, let $e = s(ψ)$, $\mathcal{K} = e\HS$, $z$ and $ψ_e$ be as above, and let $D(\mathcal{K}, ψ_e)$, $R^{ψ_e}$ and $θ^{ψ_e}$ be the objects of \Autoref{def:spatial-derivative} on $\mathcal{K}$ for the algebra $\M e|_{\mathcal{K}}$, its commutant $e\M'e|_{\mathcal{K}}$ and the faithful functional $ψ_e$.
    \begin{enumerate}[label=(\alph*)]
        \item\label{item:bv-domain} $D(\HS, ψ) = D(\mathcal{K}, ψ_e)$.
        \item\label{item:bv-theta} For $ξ \in D(\HS, ψ)$ one has $θ^ψ(ξ, ξ) = θ^ψ(ξ, ξ)\, z$ and $θ^ψ(ξ, ξ)\, e|_{\mathcal{K}} = θ^{ψ_e}(ξ, ξ)$. Consequently,
        \begin{equation}\label{eq:theta-reduction}
            \widehat{\varphi}\big(θ^ψ(ξ, ξ)\big) = \widehat{\varphi}_e\big(θ^{ψ_e}(ξ, ξ)\big)\,, \qquad{ξ} \in D(\HS, ψ)\,,
        \end{equation}
        for every positive linear functional $\widehat{\varphi}$ on $\M$, normal or not, where $\widehat{\varphi}_e(xe|_{\mathcal{K}}) \coloneqq \widehat{\varphi}(xz)$ is defined as in \eqref{eq:induced-functional}.
        \item\label{item:bv-form} For $φ \in \M_*^+$ the form $q_φ$ of \Autoref{def:spatial-derivative} is the form $ξ \mapsto φ_e(θ^{ψ_e}(ξ, ξ))$ on $D(\mathcal{K}, ψ_e)$. Hence the restriction of $\Delta(φ/ψ)$ to $\mathcal{K}$ is $\Delta_{\mathcal{K}}(φ_e/ψ_e)$, the spatial derivative of $φ_e$ with respect to the faithful functional $ψ_e$.
    \end{enumerate}
\end{lemma}
\begin{proof}
    \ref{item:bv-domain} If $ξ \in D(\HS, ψ)$, then $(\IdentityMatrix - e)ξ = 0$ (take $y = \IdentityMatrix - e$), and restricting the defining inequality to $y \in e\M'e$ shows $ξ \in D(\mathcal{K}, ψ_e)$. Conversely, let $ξ \in D(\mathcal{K}, ψ_e)$ with constant $C$ and let $y \in \M'$. Then $b \coloneqq (ey^*ye)^{1/2} \in e\M'e$ and $\norm{yξ}^2 = \norm{yeξ}^2 = \norm{bξ}^2 \leq C\, ψ(b^2) = C\, ψ(y^*y)$, so $ξ \in D(\HS, ψ)$.

    \ref{item:bv-theta} Let $ξ \in D(\HS, ψ)$ and $R \coloneqq R^ψ(ξ)$. Since $ψ(y^*y) = ψ(ey^*ye)$, we have $η_ψ(y) = η_ψ(ye)$ for all $y \in \M'$, so that the GNS space of $ψ_e$ is the subspace $π_ψ(e)\HS_ψ = \overline{η_ψ(e\M'e)}$ of $\HS_ψ$, and $R^{ψ_e}(ξ)$ is the restriction of $R$ to this subspace. As $Rπ_ψ(e) = eR$, this gives $θ^{ψ_e}(ξ, ξ) = Rπ_ψ(e)R^*|_{\mathcal{K}} = θ^ψ(ξ, ξ)\, e|_{\mathcal{K}}$. Moreover, the range of $R$ is contained in $\overline{\M'ξ} \subset z\HS$, because $ξ \in e\HS \subset z\HS$ and $z$ is central; hence $θ^ψ(ξ, ξ) = θ^ψ(ξ, ξ)\, z$. Finally, $\widehat{\varphi}_e$ is well defined for every linear functional $\widehat{\varphi}$ on $\M$, since the kernel of $x \mapsto xe|_{\mathcal{K}}$ is $\M(\IdentityMatrix - z)$, and the two identities give \eqref{eq:theta-reduction}: $\widehat{\varphi}(θ^ψ(ξ,ξ)) = \widehat{\varphi}(θ^ψ(ξ,ξ)z) = \widehat{\varphi}_e(θ^ψ(ξ,ξ)e|_{\mathcal{K}}) = \widehat{\varphi}_e(θ^{ψ_e}(ξ,ξ))$.

    \ref{item:bv-form} The first statement follows from \ref{item:bv-domain} and \eqref{eq:theta-reduction}. By \cite[Lemma~11.3, Thms.~A.11 and~A.12]{hiai_lectures_2021}, applied on $\mathcal{K}$, this form is closable and its closure is the form of $\Delta_{\mathcal{K}}(φ_e/ψ_e)$, which by \Autoref{def:spatial-derivative} is the restriction of $\Delta(φ/ψ)$ to $\mathcal{K}$.
\end{proof}

In particular, the statements of \Autoref{def:spatial-derivative} for a non-faithful functional $ψ$, namely that $D(\HS, ψ)$ is a subspace with closure $e\HS$ and that $q_φ$ is closable, follow from the corresponding statements for the faithful functional $ψ_e$.

\subsubsection{Properties of the spatial derivative}\label[subsubappendix]{app:spatial-derivative-general-proof}

\begin{proof}[Proof of Proposition~\ref{prop:spatial-derivative-properties}]
    If $ψ$ is faithful, or an n.f.s.\ weight, then \ref{item:sd-support} is \cite[Thm.~11.7(3), Prop.~11.19]{hiai_lectures_2021}, the inversion formula is \cite[Thm.~11.8]{hiai_lectures_2021}, the statements of \ref{item:sd-linearity} about the first argument and \ref{item:sd-covariance} are \cite[Prop.~11.17]{hiai_lectures_2021}, and \ref{item:sd-standard-form} is \cite[Prop.~11.6]{hiai_lectures_2021}. Now let $ψ \in (\M')_*^+$ be arbitrary, and let $e$, $\mathcal{K}$, $ψ_e$ and $φ_e$ be as in \Cref{app:bounded-vectors-proof}. By \Autoref{lem:bounded-vectors-nonfaithful}\ref{item:bv-form}, on $\mathcal{K} = s(ψ)\HS$ the operator $\Delta(φ/ψ)$ is $\Delta_{\mathcal{K}}(φ_e/ψ_e)$, to which the results for a faithful second argument apply, and it is also the spatial derivative of \cite[Ch.~4]{ohya_quantum_1993}; on $(\IdentityMatrix - s(ψ))\HS$ it is $+\infty$ by our convention and $0$ by the convention used there.

    \ref{item:sd-support} Since $z \geq e$, the induced functional $φ_e$ has the support $s(φ)e|_{\mathcal{K}}$ in $\M e|_{\mathcal{K}}$, so the claim follows from the faithful case, applied on $\mathcal{K}$; it is \cite[Eq.~(4.4)]{ohya_quantum_1993}, by which the operator defined there has the support projection $s(φ)s(ψ)$.

    \ref{item:sd-inverse} On $s(φ)s(ψ)\HS$ the identity $\Delta(φ/ψ)^{-1} = \Delta(ψ/φ)$ is the inversion formula \cite[Eq.~(4.7)]{ohya_quantum_1993}; it follows from the faithful case by cutting down to the supports of $ψ$ (\Autoref{lem:bounded-vectors-nonfaithful}) and of $φ$ \cite[Lemma~11.24]{hiai_lectures_2021}. On the three other subspaces $(\IdentityMatrix - s(φ))s(ψ)\HS$, $s(φ)(\IdentityMatrix - s(ψ))\HS$ and $(\IdentityMatrix - s(φ))(\IdentityMatrix - s(ψ))\HS$, which together with the first one reduce both operators, the values are given by \ref{item:sd-support}: $\Delta(φ/ψ)^{-1}$ takes the values $0^{-1} = \infty$, $\infty^{-1} = 0$ and $\infty^{-1} = 0$, and $\Delta(ψ/φ)$ the values $\infty$, $0$ and $\infty$. Thus the two operators coincide except on the last subspace, where the first one vanishes; this also gives the inequality, and the equality of the spectral measures for vectors in $s(φ)\HS \cup s(ψ)\HS$, which have no component in the last subspace.

    \ref{item:sd-linearity} The statements about the first argument follow from the faithful case, applied on $\mathcal{K}$, since $φ \mapsto φ_e$ is additive, positively homogeneous, monotone and normal (see also \cite[Eq.~(4.5)]{ohya_quantum_1993}). Now let $ψ_1 \leq ψ_2$. Then $D(\HS, ψ_1) \subset D(\HS, ψ_2)$ and $R^{ψ_2}(ξ) = R^{ψ_1}(ξ)V$ with the contraction $V: η_{ψ_2}(y) \mapsto η_{ψ_1}(y)$, so that $θ^{ψ_2}(ξ, ξ) \leq θ^{ψ_1}(ξ, ξ)$ for $ξ \in D(\HS, ψ_1)$ (cf.\ \cite[Lemma~4.1]{ohya_quantum_1993}) and hence $\ip{ξ}{\Delta(φ/ψ_2)ξ} \leq \ip{ξ}{\Delta(φ/ψ_1)ξ}$ for these vectors. Since $D(\HS, ψ_1)$ is dense in the form domain of $\Delta(φ/ψ_1)$ with respect to the form norm (\Autoref{def:spatial-derivative}) and the form of $\Delta(φ/ψ_2)$ is lower semicontinuous, the inequality extends to the form domain of $\Delta(φ/ψ_1)$, and outside of it the right hand side is $+\infty$.

    \ref{item:sd-covariance} Since $(aφa^*)_e(xe) = φ(a^*xaz) = φ_e(a_e^*\, xe\, a_e)$, we have $(aφa^*)_e = a_e φ_e a_e^*$, and with $a_e \coloneqq ae|_{\mathcal{K}}$ the faithful case, applied on $\mathcal{K}$, gives $\Delta_{\mathcal{K}}((aφa^*)_e/ψ_e) = a_e\, \Delta_{\mathcal{K}}(φ_e/ψ_e)\, a_e^*$, which is the claim.

    \ref{item:sd-standard-form} Here the second argument $σ'$ is faithful, because $ξ_σ$ is cyclic for $\M$, so this is the faithful case.
\end{proof}

\subsubsection{Haagerup's identity for non-faithful functionals}\label[subsubappendix]{app:haagerup-proof}

\begin{proof}[Proof of Lemma~\ref{lem:haagerup-identity-finite-index}]
    For n.f.s.\ weights, in particular for faithful $σ$ and $ω'$, the identity \eqref{eq:haagerup-identity-finite-index} is Haagerup's identity \eqref{eq:inverse_weight}, i.e.\ Haagerup's theorem \cite[Thms.~5.9, 6.13]{haagerup_operator_1979a} in the spatial-derivative form of \cite[Rem.~11.26]{hiai_lectures_2021} (Step~2 below). The other steps extend it to non-faithful functionals: Step~1 determines the supports, Step~3 is bookkeeping, and Steps~4 and~5 remove the faithfulness of $σ$ and of $ω'$ by the completion device. For the first argument this is done exactly as in the proof of \cite[Thm.~11.7(4)]{hiai_lectures_2021} (see also \cite[Eq.~(4.6)]{ohya_quantum_1993} and the end of the proof of \cite[Cor.~11.25]{hiai_lectures_2021}, where it is remarked that identities of this type extend to non-faithful first arguments in this way); for the second argument we did not find the corresponding statement in the literature and include the short argument.

    \emph{Step 1: the support equalities.} Let $f \coloneqq s(σ) \in \N$. By the bimodule property of $E$, $E(\IdentityMatrix - f) = (\IdentityMatrix - f)E(\IdentityMatrix)(\IdentityMatrix - f)$, so $E_*(σ)(\IdentityMatrix - f) = σ\big((\IdentityMatrix - f)E(\IdentityMatrix)(\IdentityMatrix - f)\big) = 0$ and hence $s(E_*(σ)) \leq f$. Conversely, let $g \coloneqq \IdentityMatrix - s(E_*(σ)) \in \M$, so that $σ(E(g)) = 0$. Since $σ(y) = 0$ implies $yf = 0$ for $y \in \N_+$, we get $E(g)f = 0$, hence $E(fgf) = fE(g)f = 0$, and faithfulness of $E$ gives $fgf = 0$, i.e.\ $gf = 0$ and $g \leq \IdentityMatrix - f$. Thus $s(E_*(σ)) = f$. The same argument applies to $E^{-1}$, which is a faithful $\M'$-bimodule map from $\N'_+$ to $\M'_+$ by \Autoref{lem:finite_index_bounded_inverse_weight}.

    \emph{Step 2: faithful $σ$ and $ω'$.} If $σ$ is faithful, then $E_*(σ)(x^*x) = σ(E(x^*x)) = 0$ implies $E(x^*x) = 0$ and hence $x = 0$, so $E_*(σ)$ is a faithful positive normal functional, i.e.\ an n.f.s.\ weight on $\M$; likewise $E^{-1}_*(ω')$ is an n.f.s.\ weight on $\N'$ if $ω'$ is faithful. Hence all four functionals in \eqref{eq:haagerup-identity-finite-index} are n.f.s.\ weights, and the first identity is Haagerup's identity \eqref{eq:inverse_weight} \cite[Rem.~11.26]{hiai_lectures_2021}. The second identity follows from the first by \Autoref{prop:spatial-derivative-properties}\ref{item:sd-inverse}. This proves the last assertion of the lemma.

    \emph{Step 3: reduction to the first identity.} Now let $σ$ and $ω'$ be arbitrary, and write $f \coloneqq s(σ) \in \N$, $e' \coloneqq s(ω') \in \M'$ and $D \coloneqq (\IdentityMatrix - f)(\IdentityMatrix - e')\HS$. It suffices to prove the first identity in \eqref{eq:haagerup-identity-finite-index}. Indeed, by Step 1 and \Autoref{prop:spatial-derivative-properties}\ref{item:sd-support} all four operators in \eqref{eq:haagerup-identity-finite-index} are $+\infty$ on $D$, and by \Autoref{prop:spatial-derivative-properties}\ref{item:sd-inverse}, applied to the pairs $(E_*(σ), ω')$ and $(σ, E^{-1}_*(ω'))$, each side of the second identity coincides on $D^\perp$ with the generalized inverse of the corresponding side of the first identity.

    \emph{Step 4: the second argument is an n.f.s.\ weight.} We first show that $\Delta(E_*(σ)/ω') = \Delta(σ/E^{-1}_*(ω'))$ for every $σ \in \N_*^+$ whenever $ω'$ is an n.f.s.\ weight on $\M'$, all spatial derivatives being those of \Autoref{def:spatial-derivative}. Choose a semifinite normal weight $σ_0$ on $\N$ with $s(σ_0) = \IdentityMatrix - f$ (for instance $σ_0(y) \coloneqq τ_0\big((\IdentityMatrix - f)y(\IdentityMatrix - f)\big)$ for an n.f.s.\ weight $τ_0$ on $(\IdentityMatrix - f)\N(\IdentityMatrix - f)$) and put $σ_1 \coloneqq σ + σ_0$, an n.f.s.\ weight on $\N$. Then $E_*(σ_1) = σ_1 \circ E$ is an n.f.s.\ weight on $\M$: it is normal and faithful, and it is semifinite because $E\big((xu)^*(xu)\big) = u^*E(x^*x)u \leq \norm{x}^2 \norm{E(\IdentityMatrix)}\, u^*u$ for $x \in \M$ and $u \in \N$ shows that $\mathfrak{n}_{E_*(σ_1)} \supset \M\,\mathfrak{n}_{σ_1}$, which is $σ$-weakly dense in $\M$. Likewise $E^{-1}_*(ω') = ω' \circ E^{-1}$ is an n.f.s.\ weight on $\N'$, since $\mathfrak{n}_{E^{-1}_*(ω')} \supset \N'\,\mathfrak{n}_{ω'}$ by the same estimate with $\Ind E$ in place of $E(\IdentityMatrix)$. Haagerup's identity \eqref{eq:inverse_weight} therefore gives $\Delta(E_*(σ_1)/ω') = \Delta(σ_1/E^{-1}_*(ω'))$. By additivity in the first argument, which holds for normal weights \cite[Prop.~11.17(1)]{hiai_lectures_2021}, $\Delta(E_*(σ_1)/ω')$ is the form sum of $\Delta(E_*(σ)/ω')$ and $\Delta(E_*(σ_0)/ω')$, and the supports of these two operators are $s(E_*(σ)) = f$ and $s(E_*(σ_0)) = \IdentityMatrix - f$ by Step 1 (whose argument applies verbatim to the weight $σ_0$) and \cite[Prop.~11.19]{hiai_lectures_2021}. Hence $\Delta(E_*(σ)/ω')$ is the restriction of $\Delta(E_*(σ_1)/ω')$ to $f\HS$, extended by $0$. In the same way, since $s(σ) = f$ and $s(σ_0) = \IdentityMatrix - f$, $\Delta(σ/E^{-1}_*(ω'))$ is the restriction of $\Delta(σ_1/E^{-1}_*(ω'))$ to $f\HS$, extended by $0$. This proves the claim of Step 4.

    \emph{Step 5: completion of the second argument.} We first show the counterpart of \cite[Eq.~(4.6)]{ohya_quantum_1993} for the second argument. Let $φ \in \M_*^+$ and $ψ \in (\M')_*^+$ with support $e \coloneqq s(ψ)$, and let $ψ_0$ be a normal weight on $\M'$ with $ψ_0(y) = ψ_0\big((\IdentityMatrix - e)y(\IdentityMatrix - e)\big)$ for all $y \in \M'_+$ such that $ψ_1 \coloneqq ψ + ψ_0$ is an n.f.s.\ weight. We claim that $e\HS$ reduces the spatial derivative $\Delta(φ/ψ_1)$ of \Autoref{def:spatial-derivative} and that the restriction of $\Delta(φ/ψ_1)$ to $e\HS$ is the finite part of $\Delta(φ/ψ)$; in other words, $\Delta(φ/ψ)$ is obtained from $\Delta(φ/ψ_1)$ by replacing it with $+\infty$ on $(\IdentityMatrix - e)\HS$. The proof is similar to the one of \cite[Lemma~11.24(b)]{hiai_lectures_2021}. For $y \in \mathfrak{n}_{ψ_1}$ we have $ψ(y^*y) \leq ψ_1(y^*y)$, and for every $y \in \M'$ we have $ye \in \mathfrak{n}_{ψ_1}$ with $ψ_1\big((ye)^*(ye)\big) = ψ(y^*y)$. Hence $D(\HS, ψ) \subset D(\HS, ψ_1)$, and $ξ \in D(\HS, ψ_1)$ implies $eξ \in D(\HS, ψ)$, so that $D(\HS, ψ) = eD(\HS, ψ_1)$. Moreover, $P: η_{ψ_1}(y) \mapsto η_{ψ_1}(ye)$ is a contractive idempotent and thus an orthogonal projection on $\HS_{ψ_1}$, and $W: η_{ψ_1}(ye) \mapsto η_ψ(y)$ is a unitary from $P\HS_{ψ_1}$ onto $\HS_ψ$. For $ξ \in D(\HS, ψ_1)$ one checks on the vectors $η_{ψ_1}(y)$, $y \in \mathfrak{n}_{ψ_1}$, that $R^{ψ_1}(eξ) = R^{ψ_1}(ξ)P = R^{ψ}(eξ)WP$ and $R^{ψ_1}\big((\IdentityMatrix - e)ξ\big) = R^{ψ_1}(ξ)(\IdentityMatrix - P)$, and therefore
    \begin{equation}
        θ^{ψ_1}(ξ, ξ) = R^{ψ_1}(ξ)\big(P + (\IdentityMatrix - P)\big)R^{ψ_1}(ξ)^* = θ^{ψ}(eξ, eξ) + θ^{ψ_1}\big((\IdentityMatrix - e)ξ, (\IdentityMatrix - e)ξ\big)\,.
    \end{equation}
    Applying $φ$ shows that the form $ξ \mapsto φ(θ^{ψ_1}(ξ, ξ))$ on $D(\HS, ψ_1)$ is the orthogonal sum of the form $q_φ$ of \Autoref{def:spatial-derivative} for $ψ$ on $D(\HS, ψ) \subset e\HS$ and of a form on $(\IdentityMatrix - e)D(\HS, ψ_1) \subset (\IdentityMatrix - e)\HS$. The same is then true for the closures, which is the claim, by \Autoref{def:spatial-derivative}.

    Now let $ω' \in (\M')_*^+$ be arbitrary, with support $e' = s(ω')$ as in Step~3. Choose an n.f.s.\ weight $τ'_0$ on $(\IdentityMatrix - e')\M'(\IdentityMatrix - e')$ and put $ω'_0(y) \coloneqq τ'_0\big((\IdentityMatrix - e')y(\IdentityMatrix - e')\big)$ and $ω'_1 \coloneqq ω' + ω'_0$, an n.f.s.\ weight on $\M'$. Then $E^{-1}_*(ω'_1) = E^{-1}_*(ω') + κ'_0$ with $κ'_0 \coloneqq ω'_0 \circ E^{-1}$ is an n.f.s.\ weight on $\N'$ by the argument of Step~4, the functional $E^{-1}_*(ω')$ has the support $e'$ by Step~1, and $κ'_0(y) = κ'_0\big((\IdentityMatrix - e')y(\IdentityMatrix - e')\big)$ for $y \in \N'_+$ by the bimodule property of $E^{-1}$, since $e' \in \M'$. By Step~4,
    \begin{equation}
        \Delta(E_*(σ)/ω'_1) = \Delta(σ/E^{-1}_*(ω'_1))\,.
    \end{equation}
    We apply the claim to both sides, with
    \begin{equation*}
        (\M, \M', E_*(σ), ω', ω'_0) \quad \text{on the left} \qquad \text{and} \qquad (\N, \N', σ, E^{-1}_*(ω'), κ'_0) \quad \text{on the right}
    \end{equation*}
    in the roles of $(\M, \M', φ, ψ, ψ_0)$. It shows that both $\Delta(E_*(σ)/ω')$ and $\Delta(σ/E^{-1}_*(ω'))$ are obtained from this operator by replacing it with $+\infty$ on $(\IdentityMatrix - e')\HS$; hence they coincide. This proves the first identity in \eqref{eq:haagerup-identity-finite-index} and completes the proof.
\end{proof}

\subsection{Entropic quantities}\label[subappendix]{app:divergence-proofs}

The well-definedness of the relative entropy and the finiteness of the Petz--\renyi divergences of order $α < 1$ rest on the following identities for the spectral measures entering \Autoref{def:divergences}. They are the ones behind the classical argument that Araki's relative entropy is well defined, see \cite{araki_relative_1977} and \cite[Eqs.~(5.1)--(5.2), Prop.~5.1]{ohya_quantum_1993}, where the relative entropy is defined through the operator $\Delta(σ/ρ')$.

\begin{lemma}[Spectral measures and first moments]\label{lem:first-moment}
    In the setting of \Autoref{def:divergences}, let $μ$ be the spectral measure of $\v{ρ}$ with respect to $\Delta(ρ'/σ)$, a probability measure on $[0, \infty]$, and $ν$ the spectral measure of $\v{ρ}$ with respect to $\Delta(σ/ρ')$. Then $μ(\set{0}) = 0$, $ν$ is the image of $μ$ under $λ \mapsto 1/λ$, and
    \begin{equation}\label{eq:first-moment}
        μ(\set{\infty}) = ν(\set{0}) = ρ(\IdentityMatrix - s(σ)), \qquad \int_{[0, \infty)} λ \, \dd{ν}(λ) = \bra{\v{ρ}} \Delta(σ/ρ') \ket{\v{ρ}} = σ(s(ρ))\,.
    \end{equation}
    In particular, $D(ρ\|σ) = \int \log{λ} \, \dd{μ}(λ) = - \int \log{λ} \, \dd{ν}(λ)$. Moreover, for $α \in [0, 1)$,
    \begin{equation}\label{eq:petz-bound}
        \bra{\v{ρ}} \Delta(ρ'/σ)^{α - 1} \ket{\v{ρ}} = \int{λ}^{1 - α} \, \dd{ν}(λ) \leq{σ}(s(ρ))^{1-α}\,, \qquad \text{so that } D_α(ρ\|σ) \geq - \log{σ}(s(ρ))\,,
    \end{equation}
    with equality in the first relation for $α = 0$, i.e.\ $D_0(ρ\|σ) = -\log{σ}(s(ρ))$, and $D_α(ρ\|σ) = +\infty$ if and only if $ρ(s(σ)) = 0$.
\end{lemma}
\begin{proof}
    Since $\v{ρ} \in s(ρ')\HS$ and $\norm{(\IdentityMatrix - s(σ))\v{ρ}}^2 = ρ(\IdentityMatrix - s(σ))$, the statements about $μ(\set{0})$, about $ν$ and about the atoms in \eqref{eq:first-moment} follow from \Autoref{prop:spatial-derivative-properties}\ref{item:sd-support} and~\ref{item:sd-inverse}. The vector $\v{ρ}$ is $ρ'$-bounded, and since $η_{ρ'}(y) \mapsto y\v{ρ}$ is an isometry onto $\overline{\M'\v{ρ}}$, we have $θ^{ρ'}(\v{ρ}, \v{ρ}) = [\M'\v{ρ}] = s(ρ)$ \cite[p.~68]{ohya_quantum_1993}. Hence $\bra{\v{ρ}} \Delta(σ/ρ') \ket{\v{ρ}} = σ(s(ρ))$ by \eqref{eq:spatial-derivative-form}, applied with the roles of $\M$ and $\M'$ interchanged. Finally, $ν$ is a probability measure on $[0, \infty)$, so \eqref{eq:petz-bound} is Jensen's inequality for the concave function $λ \mapsto λ^{1-α}$ (with $0^{1-α} = 0$, consistent with $\infty^{α - 1} = 0$), and $\int λ^{1-α} \, \dd{ν}$ vanishes if and only if $ν$ is concentrated at $0$, i.e.\ if and only if $ρ(\IdentityMatrix - s(σ)) = 1$.
\end{proof}

\begin{proof}[Proof of Lemma~\ref{lem:relative-entropy-well-defined}]
    By \Autoref{lem:first-moment}, $D(ρ\|σ) = -\int \log λ \, \dd{ν}(λ)$ for a probability measure $ν$ on $[0, \infty)$ with $\int λ \, \dd{ν} = σ(s(ρ))$. Since $\log λ \leq λ$, the positive part of $\int \log λ \, \dd{ν}$ is finite, so that $D(ρ\|σ)$ is well defined with values in $(-\infty, +\infty]$, and Jensen's inequality for the concave function $\log$ gives $\int \log λ \, \dd{ν} \leq \log σ(s(ρ))$. Note that $D(ρ\|σ) = +\infty$ if $ν(\set{0}) > 0$, i.e.\ if $s(ρ) \not\leq s(σ)$.
\end{proof}

\begin{proof}[Proof of Lemma~\ref{lem:sandwiched-support-restriction}]
    In \cite[Def.~2]{berta_renyi_2018} and \cite[Sec.~2]{jencova_renyi_2021} the restriction to states $ω'$ with $s(ω') \geq s(ρ')$ is part of the definition of the sandwiched divergence for $α < 1$ (see \Autoref{rem:divergence-conventions}); the lemma shows that under our conventions it does not change the value of the infimum. The proof is the usual perturbation argument, in which $ω'$ is mixed with a small multiple of $ρ'$ in order to enlarge its support.

    The first statement was already observed in \Autoref{rem:divergence-conventions}. Let $ω' \in \S(\M')$ have a finite value and let $ε > 0$. The state $ω'_ε \coloneqq (ω' + ερ')/(1 + ε)$ has support $s(ω') \vee s(ρ') \geq s(ρ')$, and $\Delta(ω'_ε/σ) \geq (1+ε)^{-1}\Delta(ω'/σ)$ by \Autoref{prop:spatial-derivative-properties}\ref{item:sd-linearity}. Since $α' < 0$, \Autoref{lem:compression-inequality}\ref{item:ci-monotone} gives $\bra{\v{ρ}} \Delta(ω'_ε/σ)^{α'}\ket{\v{ρ}} \leq (1+ε)^{-α'} \bra{\v{ρ}} \Delta(ω'/σ)^{α'}\ket{\v{ρ}}$, and letting $ε \to 0$ proves the first equality in \eqref{eq:sandwiched-support-restriction}. The second one is \Autoref{prop:spatial-derivative-properties}\ref{item:sd-inverse}, since $\v{ρ} \in s(ρ')\HS \subset s(ω')\HS$ for the states in question. Finally, if $ρ(s(σ)) = 0$ then $\v{ρ} \in (\IdentityMatrix - s(σ))\HS$, where every $\Delta(ω'/σ)^{α'}$ vanishes.
\end{proof}

\begin{proof}[Proof of Proposition~\ref{prop:representation-independence}]
    For the divergences defined through an optimization over vector states, the independence of the representation and of the vector representative is \cite[Lemma~3]{berta_renyi_2018}. We follow the argument given there, which transports the spatial derivatives along the partial isometry intertwining two vector representatives and then applies the operator Jensen inequality (here in the form of \Autoref{lem:compression-inequality}\ref{item:ci-compression}); the modifications account for the optimization over all normal states of $\M'$ instead of vector states, and for our support conventions. For clarity we separate the change of the vector representative within one representation (Step~1) from the change of the representation, which is reduced to Step~1 by passing to the direct sum of the two representations (Step~2).

    We fix $s \coloneqq \abs{α'} \in (0, 1]$ and consider the functions $f_+(t) = t^{s}$ and $f_-(t) = t^{-s}$ on $[0, \infty]$.

    \emph{Step 1: two vectors in the same representation.} Let $ξ, η \in \HS$ both represent $ρ$. Define $u'(xξ) \coloneqq xη$ for $x \in \M$ and $u' \coloneqq 0$ on $(\M{ξ})^\perp$. Since $\norm{xη}^2 = ρ(x^*x) = \norm{xξ}^2$, this is a well-defined partial isometry with $u'^*u' = p \coloneqq [\M{ξ}]$ and $u'u'^* = q \coloneqq [\Mη]$; it commutes with $\M$ on $\overline{\M{ξ}}$ and on its ($\M$-invariant) complement, so $u' \in \M'$, and $u'ξ = η$, $u'^*η = ξ$. Let $e \coloneqq s(σ) \in \M$, $\mathcal{K} \coloneqq e\HS$, $z$ the central support of $e$, and for $ω' \in (\M')_*^+$ let $ω'_e(ye|_{\mathcal{K}}) \coloneqq{ω}'(yz)$ be the induced functional on $\M'e|_{\mathcal{K}}$, as in \eqref{eq:induced-functional} with the roles of $\M$ and $\M'$ interchanged, and write $\Delta_{\mathcal{K}}(ω') \coloneqq \Delta_{\mathcal{K}}(ω'_e/σ_e)$ for the finite part of $\Delta(ω'/σ)$. Note that $e$ commutes with $p, q, u' \in \M'$, and that $u'_e \coloneqq u'e|_{\mathcal{K}}$ is a partial isometry in $\M'e|_{\mathcal{K}}$ from $pe\mathcal{K}$ onto $qe\mathcal{K}$. Since the $+\infty$ part of every $\Delta(ω'/σ)$ contributes $f_\pm(\infty)\, \norm{(\IdentityMatrix - e)ξ}^2 = f_\pm(\infty)\, ρ(\IdentityMatrix - e) = f_\pm(\infty)\, \norm{(\IdentityMatrix - e)η}^2$, the same for $ξ$ and $η$, it suffices to compare the contributions of the finite parts, i.e.\ the quantities $\ip{eξ}{f(\Delta_{\mathcal{K}}(ω'))\, eξ}$ and $\ip{eη}{f(\Delta_{\mathcal{K}}(ω'))\, eη}$.

    (a) For the vector functionals we have $ω_η = u' ω_ξ u'^*$ on $\M'$, i.e.\ $ω_η(y) = ω_ξ(u'^* y u')$, hence $(ω_η)_e = u'_e (ω_ξ)_e u'^*_e$, and \Autoref{prop:spatial-derivative-properties}\ref{item:sd-covariance} gives $\Delta_{\mathcal{K}}(ω_η) = u'_e \Delta_{\mathcal{K}}(ω_ξ) u'^*_e$. By \Autoref{prop:spatial-derivative-properties}\ref{item:sd-support} the operator $\Delta_{\mathcal{K}}(ω_ξ)$ is supported in $s(ω_ξ|_{\M'})e\mathcal{K} = pe\mathcal{K}$, so $\Delta_{\mathcal{K}}(ω_η)$ is its transport to $qe\mathcal{K}$ by $u'_e$, extended by $0$. Since $eη = u'_e\, eξ$, the spectral measures of $eη$ with respect to $\Delta_{\mathcal{K}}(ω_η)$ and of $eξ$ with respect to $\Delta_{\mathcal{K}}(ω_ξ)$ coincide. Together with the common atom at $\infty$ this shows that $ξ$ and $η$ have the same spectral measure with respect to $\Delta(ω_ξ/σ)$ and $\Delta(ω_η/σ)$, respectively.

    (b) Let $ω' \in \S(\M')$ and put $\check{ω}' \coloneqq u'^* ω' u'$, i.e.\ $\check{ω}'(y) = ω'(u' y u'^*)$, a positive normal functional with $\check{ω}'(\IdentityMatrix) = ω'(q) \leq 1$. By \Autoref{prop:spatial-derivative-properties}\ref{item:sd-covariance}, $\Delta_{\mathcal{K}}(\check{ω}') = u'^*_e \Delta_{\mathcal{K}}(ω') u'_e$, which is the transport by $u'^*_e$ of the compression of $\Delta_{\mathcal{K}}(ω')$ to $qe\mathcal{K}$ (\Autoref{lem:compression-inequality}), extended by $0$ on $(\IdentityMatrix - pe)\mathcal{K}$. Since $eξ = u'^*_e\, eη \in pe\mathcal{K}$, \Autoref{lem:compression-inequality}\ref{item:ci-compression} yields
    \begin{equation}\label{eq:ri-step1}
        \ip{eξ}{f_+(\Delta_{\mathcal{K}}(\check{ω}'))\, eξ} \geq \ip{eη}{f_+(\Delta_{\mathcal{K}}(ω'))\, eη}\,, \qquad
        \ip{eξ}{f_-(\Delta_{\mathcal{K}}(\check{ω}'))\, eξ} \leq \ip{eη}{f_-(\Delta_{\mathcal{K}}(ω'))\, eη}\,.
    \end{equation}
    If $\check{ω}' \neq 0$, then $c \coloneqq \check{ω}'(\IdentityMatrix)^{-1} \geq 1$ and $c\check{ω}' \in \S(\M')$; by \Autoref{prop:spatial-derivative-properties}\ref{item:sd-linearity}, replacing $\check{ω}'$ by $c\check{ω}'$ multiplies the left hand sides of \eqref{eq:ri-step1} by $c^{s} \geq 1$ and $c^{-s} \leq 1$, respectively, which only improves the inequalities. If $\check{ω}' = 0$, then $s(ω') \perp q$, so $\Delta_{\mathcal{K}}(ω')$ vanishes on $qe\mathcal{K} \ni eη$; thus the right hand side of the first inequality in \eqref{eq:ri-step1} is $0$, and the right hand side of the second is $+\infty$ unless $eη = 0$, in which case $ρ(e) = 0$ and all quantities involved vanish. In all cases we conclude, adding the common contribution of the $+\infty$ parts, that for every $ω' \in \S(\M')$ there is an $ω'' \in \S(\M')$ with $\ip{ξ}{f_+(\Delta(ω''/σ))ξ} \geq \ip{η}{f_+(\Delta(ω'/σ))η}$ and $\ip{ξ}{f_-(\Delta(ω''/σ))ξ} \leq \ip{η}{f_-(\Delta(ω'/σ))η}$. Hence the supremum of $\ip{\,\cdot\,}{f_+(\Delta(ω'/σ))\,\cdot\,}$ over $ω' \in \S(\M')$ evaluated at $ξ$ dominates the one evaluated at $η$, and the infimum of $\ip{\,\cdot\,}{f_-(\Delta(ω'/σ))\,\cdot\,}$ evaluated at $ξ$ is dominated by the one evaluated at $η$. Exchanging the roles of $ξ$ and $η$ gives equality.

    \emph{Step 2: enlarging the representation.} Let $\hat\HS \coloneqq \HS \oplus \tilde\HS$ with the faithful normal representation $\hat{π}(x) \coloneqq x \oplus \tilde{π}(x)$, and let $\hat\M \coloneqq \hat{π}(\M)$, $\hat{σ} \coloneqq{σ} \circ \hat{π}^{-1}$. The vectors $\v{ρ} \oplus 0$ and $0 \oplus \tilde{\v{ρ}}$ both represent $ρ$ in this representation, so by Step 1 it suffices to show that the quantities in \ref{item:ri-unoptimized} and \ref{item:ri-sandwiched} computed for $(\hat\HS, \v{ρ} \oplus 0)$ agree with those computed for $(\HS, \v{ρ})$ (the same argument then applies to $(\hat\HS, 0 \oplus \tilde{\v{ρ}})$ and $(\tilde\HS, \tilde{\v{ρ}})$). Let $P \coloneqq \IdentityMatrix_{\HS} \oplus 0$. Then $P \in \hat\M'$, the elements of $\hat\M'$ are the block operators $Y = (Y_{ij})_{i,j=1,2}$ with $Y_{11} \in \M'$, and $P\hat\M'P = \M' \oplus 0$. Let $\hat{e} \coloneqq \hat{π}(e) = e \oplus \tilde{π}(e)$, $\hat{\mathcal{K}} \coloneqq \hat{e}\hat\HS = \mathcal{K} \oplus \tilde{\mathcal{K}}$, and $\hat{z} = \hat{π}(z)$ the central support of $\hat{e}$ in $\hat\M$. We claim that for every $\hat{ω}' \in (\hat\M')_*^+$ with $s(\hat{ω}') \leq P$, writing $ω'(y) \coloneqq \hat{ω}'(y \oplus 0)$ for the corresponding element of $(\M')_*^+$,
    \begin{equation}\label{eq:ri-direct-sum}
        \Delta_{\hat{\mathcal{K}}}(\hat{ω}'_{\hat e}/\hat{σ}_{\hat e}) = \Delta_{\mathcal{K}}(ω'_e/σ_e) \oplus 0_{\tilde{\mathcal{K}}}\,.
    \end{equation}
    Indeed, the GNS space of $\hat{σ}_{\hat e}$ is that of $σ_e$, a vector $ζ = ζ_1 \oplus{ζ}_2 \in \hat{\mathcal{K}}$ is $\hat{σ}_{\hat e}$-bounded iff $ζ_1$ is $σ_e$-bounded and $ζ_2$ is $\tilde{σ}_{\tilde e}$-bounded, and then $R^{\hat{σ}_{\hat e}}(ζ) = R^{σ_e}(ζ_1) \oplus R^{\tilde{σ}_{\tilde e}}(ζ_2)$, so that $θ^{\hat{σ}_{\hat e}}(ζ, ζ)$ is the block operator with entries $R(ζ_i)R(ζ_j)^*$. As $\hat{ω}'$ is supported in $P$, which commutes with $\hat{z}$, only the $(1,1)$ entry contributes: $\hat{ω}'_{\hat e}(θ^{\hat{σ}_{\hat e}}(ζ,ζ)) = \hat{ω}'(P θ^{\hat{σ}_{\hat e}}(ζ,ζ)\hat{z} P) = ω'\big(R^{σ_e}(ζ_1)R^{σ_e}(ζ_1)^* z\big) = ω'_e(θ^{σ_e}(ζ_1, ζ_1))$. Since the $σ$-bounded vectors form a core, \eqref{eq:ri-direct-sum} follows from \eqref{eq:spatial-derivative-form}. Now \ref{item:ri-unoptimized} follows by applying \eqref{eq:ri-direct-sum} to $\hat{ω}' = ω_{\v{ρ} \oplus 0}|_{\hat\M'}$, whose support $[\hat\M(\v{ρ} \oplus 0)] = [\M\v{ρ}] \oplus 0$ lies below $P$ and which corresponds to $ω' = ρ'$; the $+\infty$ parts contribute the same atom $ρ(\IdentityMatrix - e)$ on both sides. For \ref{item:ri-sandwiched}, every $ω' \in \S(\M')$ defines the state $\hat{ω}'(Y) \coloneqq{ω}'(Y_{11})$ of $\hat\M'$, which is supported in $P$ and by \eqref{eq:ri-direct-sum} has the same value at $\v{ρ} \oplus 0$ as $ω'$ has at $\v{ρ}$; this shows that the supremum (infimum) over $\S(\hat\M')$ is at least (at most) the one over $\S(\M')$. Conversely, for $\hat{ω}' \in \S(\hat\M')$ the compression $P\hat{ω}'P$ is supported in $P$, and exactly as in Step 1(b) (with $P$ in place of $q$, using \Autoref{prop:spatial-derivative-properties}\ref{item:sd-covariance} for $a = P$, \eqref{eq:ri-direct-sum} and \Autoref{lem:compression-inequality}\ref{item:ci-compression}) its value at $\v{ρ} \oplus 0$, after normalization, is at least (at most) that of $\hat{ω}'$ for $f_+$ ($f_-$). This proves \ref{item:ri-sandwiched}.

    The final statements follow: if every normal state of $\M'$ is a vector state, then the two optimizations in \eqref{eq:def-sandwiched-renyi} run over the same set; the vector-state version is representation independent by \cite[Lemma~3]{berta_renyi_2018}, and ours by what we just proved, so the two agree in every representation, and \cite[Prop.~2.4]{jencova_renyi_2021} identifies the result with the definition of \cite{jencova_renyi_2018, jencova_renyi_2021}.
\end{proof}

\begin{proof}[Proof of Proposition~\ref{prop:divergence-properties}]
    \ref{item:dp-scaling} follows from $\Delta(ω'/λσ) = λ^{-1}\Delta(ω'/σ)$ (\Autoref{prop:spatial-derivative-properties}\ref{item:sd-linearity}; note that the $+\infty$ parts agree), which shifts $\log \Delta$ by $-\log{λ}$ and multiplies $\Delta^{α - 1}$ and $\Delta^{α'}$ by $λ^{1-α}$ and $λ^{-α'}$, respectively. For \ref{item:dp-monotone}, \Autoref{prop:spatial-derivative-properties}\ref{item:sd-linearity} gives $\Delta(ω'/σ_1) \geq \Delta(ω'/σ_2)$ for all $ω' \in \S(\M')$, and \Autoref{lem:compression-inequality}\ref{item:ci-monotone}, applied with $s = \abs{α - 1} \in (0,1]$ resp.\ $s = \abs{α'} \in (0,1]$, shows that $\bra{\v{ρ}}\Delta^{α-1}\ket{\v{ρ}}$ and $\bra{\v{ρ}}\Delta^{α'}\ket{\v{ρ}}$ change in the direction that, after taking $\frac{1}{α-1}\log$ resp.\ $\sup_{ω'}\frac{1}{α'}\log$, makes the divergence decrease. For $D$ one uses $\log t = \int_0^\infty \big((1 + λ)^{-1} - (λ + t)^{-1}\big)\dd{λ}$ and the resolvent inequality from the proof of \Autoref{lem:compression-inequality}, both integrals being well defined by \Autoref{lem:relative-entropy-well-defined}. The remaining items are references; the cited works use the standard form or the vector-state form of the definitions, which is legitimate by \Autoref{prop:representation-independence}, and the inequalities \eqref{eq:alt-inequalities} are stated in \cite{jencova_renyi_2021} for states $σ$ only, but extend to $σ \in \M_*^+$ by \ref{item:dp-scaling}. The two finiteness conditions in \ref{item:dp-limits} are equivalent, and the limit of $\widetilde{D}_α$ as $α \searrow 1$ follows from the one of $D_α$, by \eqref{eq:alt-inequalities}: for $α > 1$ we have $2 - 1/α \in (1, 2)$ and $D_{2 - 1/α}(ρ\|σ) \leq \widetilde{D}_α(ρ\|σ) \leq D_α(ρ\|σ)$, where $2 - 1/α \searrow 1$ as $α \searrow 1$ (for $α = \infty$ use that $\widetilde{D}_α$ is nondecreasing in $α$).
\end{proof}

\begin{proof}[Proof of Lemma~\ref{lem:dmax-variational}]
    If $ρ \leq λσ$, then $\norm{x\v{ρ}}^2 = ρ(x^*x) \leq λ\, σ(x^*x)$ for all $x \in \M$. Hence $\v{ρ}$ is $σ$-bounded with $\norm{R^σ(\v{ρ})}^2 \leq λ$, so that $θ^σ(\v{ρ}, \v{ρ}) \leq λ\, \IdentityMatrix$, and \eqref{eq:spatial-derivative-form}, with the roles of $\M$ and $\M'$ interchanged, gives $\bra{\v{ρ}} \Delta(ω'/σ) \ket{\v{ρ}} = ω'\big(θ^σ(\v{ρ}, \v{ρ})\big) \leq λ$ for every $ω' \in \S(\M')$. For the converse it suffices to assume the bound for the vector functionals $ω' = ω_ζ|_{\M'}$ with unit vectors $ζ \in \HS$, see e.g.\ \cite[Lemma~9]{berta_renyi_2018}. Consequently, the numbers $λ > 0$ with $ρ \leq λσ$ are exactly the upper bounds of the supremum in \eqref{eq:dmax-variational}, which is strictly positive because $ρ \neq 0$, and taking the infimum over $λ$ gives \eqref{eq:dmax-variational}.
\end{proof}

\subsection{Operator-valued weights of finite index}\label[subappendix]{app:finite-index-proof}

\begin{proof}[Proof of Lemma~\ref{lem:finite_index_bounded_inverse_weight}]
    The proof is the standard argument that a positive map is bounded by its value at the identity, written out for operator-valued weights.

    We prove \ref{item:fi-dual}; the proof of \ref{item:fi-forward} is identical with $E$ in place of $E^{-1}$ and $E(\IdentityMatrix)$ in place of $\Ind E$. Write $T \coloneqq E^{-1}$. For $x' \in \N'_+$ we have $x' \leq \norm{x'}\IdentityMatrix$, so that $T(\norm{x'}\IdentityMatrix) = T(x') + T(\norm{x'}\IdentityMatrix - x') \geq T(x')$ by additivity and positivity of $T$, i.e.\
    \begin{equation}
        T(x') \leq \norm{x'}\, T(\IdentityMatrix) = \norm{x'}\, \Ind E \leq \norm{x'}\, \norm{\Ind E}\, \IdentityMatrix \qquad \text{in } \widehat{\M'_+}\,.
    \end{equation}
    An element of the extended positive part that is dominated by a multiple of the identity is a bounded operator (\Autoref{sec:prelim-vn}), so $T$ maps $\N'_+$ into $\M'_+$ with $\norm{T(x')} \leq \norm{\Ind E}\norm{x'}$. Since $T$ is additive and positively homogeneous on $\N'_+$ and every element of $\N'$ is a linear combination of four positive elements, $T$ extends uniquely to a linear map $\N' \to \M'$, which is positive by construction and therefore bounded with $\norm{T} = \norm{T(\IdentityMatrix)} = \norm{\Ind E}$ \cite[Cor.~2.9]{paulsen_completely_2003}. Normality of $T$ as an operator-valued weight means that $x'_α \nearrow x'$ in $\N'_+$ implies $T(x'_α) \nearrow T(x')$, i.e.\ $ω'(T(x'_α)) \nearrow{ω}'(T(x'))$ for every $ω' \in (\M')_*^+$; thus $ω' \circ T$ is a positive functional on $\N'$ that is normal in the sense of \Autoref{sec:prelim-vn}, i.e.\ $ω' \circ T \in (\N')_*^+$, and its norm is $ω'(T(\IdentityMatrix)) = ω'(\Ind E)$. The last statement follows by applying \ref{item:fi-forward} to $σ = ρ|_{\N}$ and \ref{item:fi-dual} to $ω' = ρ|_{\M'}$.
\end{proof}

\subsection{Conditional expectations and inverse weights in finite dimensions}\label[subappendix]{app:finite-dim-proofs}

\begin{proof}[Proof of Proposition~\ref{prop:finite-dim-cond-exp}]
	For conditional expectations, \ref{item:fdce-form} is the well-known parametrization of the conditional expectations onto a finite-dimensional subalgebra by states on the blocks of the relative commutant, see e.g.\ \cite[Lemma~1 and App.~A]{magan_quantum_2021} and the references given there. The argument applies verbatim to positive $\N$-bimodule maps $T: \BO \to \N$: since the $q_j$ are central in $\N$, we have $T(x) = \sum_j T(x_j)$; for $b \in \BO[\mathcal{V}_j]$ the element $T(\IdentityMatrix_{\mathcal{L}_j} \otimes b) \in \N q_j$ commutes with $\N$ and is therefore of the form $\Tr(τ_j b)\, q_j$ with a unique $τ_j \in \BO[\mathcal{V}_j]_+$, and the bimodule property gives \eqref{eq:finite-dim-cond-exp} on the elementary tensors $a \otimes b$ with $a \in \BO[\mathcal{L}_j]$. The only difference is that $T(\IdentityMatrix) = \sum_j (\Tr τ_j)\, q_j$ need not be the identity. The normal semifinite operator-valued weights $T: \BO_+ \to \widehat{\N_+}$ are maps of this kind: the set $\Set{x \in \BO | T(x^*x) \in \N_+}$ is a linear subspace (as $(x+y)^*(x+y) \leq 2x^*x + 2y^*y$), which is weakly dense and hence all of $\BO$, so that $T(\IdentityMatrix) \in \N_+$ and \Autoref{lem:finite_index_bounded_inverse_weight}\ref{item:fi-forward} applies. Concerning faithfulness, if $τ_j v = 0$ for a unit vector $v$, then $T$ vanishes on the nonzero projection $\IdentityMatrix_{\mathcal{L}_j} \otimes \ketbra{v}{v}$; if instead $τ_j \geq ε \IdentityMatrix$ for all $j$, then $T(x^*x) \geq ε \bigoplus_j \Tr_{\mathcal{V}_j}[(x^*x)_j] \otimes \IdentityMatrix_{\mathcal{V}_j}$, which vanishes only if $x q_j = 0$ for all $j$, i.e.\ $x = 0$. Finally, \ref{item:fdce-predual} and \ref{item:fdce-trace} follow by inserting \eqref{eq:finite-dim-cond-exp} into $σ \circ T$ and $\Tr \circ T$: by \eqref{eq:finite-dim-density-matrices}, $σ(T(x)) = \sum_j \Tr\big[(\hat{σ}_j \otimes τ_j)\, x_j\big]$, and $\Tr(T(x)) = \sum_j \dim \mathcal{V}_j \, \Tr\big[(\IdentityMatrix \otimes τ_j)\, x_j\big]$ agrees with $\Tr(x) = \sum_j \Tr(x_j)$ for all $x$ iff $τ_j = \IdentityMatrix_{\mathcal{V}_j} / \dim \mathcal{V}_j$ for all $j$.
\end{proof}

\begin{proof}[Proof of Proposition~\ref{prop:finite-dim-inverse-weight}]
	For a single block, i.e.\ if $\N$ is a factor, the form of $E^{-1}$ and the formula $\Ind E = \Tr τ^{-1}$ for the index can be found in \cite[Prop.~2.2, Thm.~2.5 and Cor.~2.6]{tsui_index_1991} (for $\M$ acting on $\HS$ instead of $\HS \otimes \HS$), and the dual conditional expectation $E'$ is discussed in \cite[App.~C]{magan_quantum_2021}; for several blocks the index is the sum of the indices of the blocks, see e.g.\ \cite[Prop.~2.3]{teruya_index_1992}. The proof below derives the explicit form of $E^{-1}$ for a general subalgebra $\N \subset \BO$ directly from its characterization through Haagerup's identity \eqref{eq:inverse_weight}.

	As in the proof of \Autoref{prop:finite-dim-cond-exp}, semifiniteness of $E^{-1}$ gives $E^{-1}(\IdentityMatrix) \in \M'_+$, so $E^{-1}$ is a positive linear $\M'$-bimodule map $\N' \to \M'$ by \Autoref{lem:finite_index_bounded_inverse_weight}\ref{item:fi-dual}. For $c = \IdentityMatrix_{\mathcal{L}_j} \otimes b \otimes \IdentityMatrix_{\HS} \in \N' \cap \M$, the element $E^{-1}(c) \in \M'$ commutes with $\M'$ (because $c$ does and $E^{-1}$ is an $\M'$-bimodule map), hence lies in $\M' \cap \M = Z(\M) = \complex \IdentityMatrix$, and by positivity $E^{-1}(c) = \Tr(τ'_j b)\, \IdentityMatrix$ for a unique $τ'_j \in \BO[\mathcal{V}_j]_+$. As $\N'$ is spanned by the products $c\, (\IdentityMatrix_{\HS} \otimes z)$ with such $c$ and $z \in \BO$, and $E^{-1}(c\, (\IdentityMatrix_{\HS} \otimes z)) = E^{-1}(c)\, (\IdentityMatrix_{\HS} \otimes z)$, we obtain $E^{-1}(y) = \IdentityMatrix_{\HS} \otimes \sum_j \Tr_{\mathcal{V}_j}[(τ'_j \otimes \IdentityMatrix_{\HS})\, y_j]$, and consequently $E^{-1}_*(ω')(y) = \sum_j \Tr[(τ'_j \otimes \hat{ω}')\, y_j]$, i.e.\ $E^{-1}_*(ω')$ has density matrices $τ'_j \otimes \hat{ω}'$. To identify the $τ'_j$, we evaluate Haagerup's identity \eqref{eq:inverse_weight} for faithful $σ$ and $ω'$ using \eqref{eq:finite-dim-spatial-derivative}: the pair $(\M, \M')$ consists of a single block, with density matrices \eqref{eq:finite-dim-cond-exp-predual} and $\hat{ω}'$, while the pair $(\N, \N')$ has the blocks $\mathcal{L}_j \otimes (\mathcal{V}_j \otimes \HS)$ of \eqref{eq:finite-dim-inclusion-algebras}, with density matrices $\hat{σ}_j$ and $τ'_j \otimes \hat{ω}'$. Hence
	\begin{equation}
		\bigoplus_{j} \hat{σ}_j \otimes τ_j \otimes (\hat{ω}')^{-1} = \Delta(E_*(σ)/ω') = \Delta(σ/E^{-1}_*(ω')) = \bigoplus_{j} \hat{σ}_j \otimes (τ'_j)^{-1} \otimes (\hat{ω}')^{-1}\,,
	\end{equation}
	and taking $\hat{σ}_j = \IdentityMatrix$ and $\hat{ω}' = \IdentityMatrix$ gives $τ'_j = τ_j^{-1}$. This proves \ref{item:fdiw-form} and \ref{item:fdiw-predual}; note that the same computation shows directly that the map \eqref{eq:finite-dim-inverse-weight} satisfies \eqref{eq:inverse_weight}, so that in finite dimensions existence and uniqueness of $E^{-1}$ are elementary. For \ref{item:fdiw-index}, taking $y = \IdentityMatrix$ in \eqref{eq:finite-dim-inverse-weight} gives $E^{-1}(\IdentityMatrix) = (\sum_j \Tr τ_j^{-1})\, \IdentityMatrix$, and since $\Tr τ_j = 1$, the Cauchy--Schwarz inequality $(\dim \mathcal{V}_j)^2 = \big(\Tr\, τ_j^{1/2} τ_j^{-1/2}\big)^2 \leq \Tr τ_j \cdot \Tr τ_j^{-1}$ holds with equality iff $τ_j = \IdentityMatrix_{\mathcal{V}_j} / \dim \mathcal{V}_j$, i.e.\ iff $E$ is trace-preserving by \Autoref{prop:finite-dim-cond-exp}\ref{item:fdce-trace}. Finally, $E' = (\Ind E)^{-1} E^{-1}$ is a unital positive $\M'$-bimodule map, hence a conditional expectation ($E'(b) = E'(\IdentityMatrix)\, b = b$ for $b \in \M'$), and it is faithful because $E^{-1}$ is.
\end{proof}

\subsection{The Pythagorean identity for operator-valued weights}\label[subappendix]{app:pythagorean-proof}

For a normal faithful conditional expectation, the Pythagorean identity \eqref{eq:pythagorean-identity} is a theorem of Petz \cite[Thm.~5.15]{ohya_quantum_1993}. \Autoref{prop:pythagorean_identity} asserts it for a general operator-valued weight $T$, and this generality is needed even if one is only interested in conditional expectations $E$: on the side of the commutants the identity is applied to $T = E^{-1}$, and since $E^{-1}(\IdentityMatrix) = \Ind E$, the map $E^{-1}$ is not a conditional expectation; it is a scalar multiple of one only if $\Ind E$ is a scalar, which need not be the case if $\M$ is not a factor (\Autoref{sec:duality-additive-form}). We do not reprove the identity for operator-valued weights but reduce it to the case of a conditional expectation, by writing $T = hF$ with a conditional expectation $F$ and the central element $h = T(\IdentityMatrix)$ of $\N$. The only analytic ingredient is the behaviour of the relative entropy under the central perturbation $σ \mapsto σ(h\,\cdot\,)$ of its second argument (Step~3 below), which is a very special case of the perturbation theory of the relative entropy, cf.\ \cite[Thm.~12.6, Cor.~12.8]{ohya_quantum_1993}.

\begin{proof}[Proof of Proposition~\ref{prop:pythagorean_identity}]
    Let $T: \M_+ \to \widehat{\N_+}$ be a normal faithful semifinite operator-valued weight with $h \coloneqq T(\IdentityMatrix)$ and $\Ind T$ bounded, so that $T: \M \to \N$ and $T^{-1}: \N' \to \M'$ are bounded positive normal maps (\Autoref{lem:finite_index_bounded_inverse_weight}) and \Autoref{lem:haagerup-identity-finite-index} holds with $T$ in place of $E$ (see the standing assumptions of \Autoref{sec:duality-relations}). We show that \eqref{eq:pythagorean-identity} follows from its special case of a conditional expectation.

    \emph{Step 1: $h$ is central in $\N$ and bounded below.} The argument for the centrality of $h$ is the same as the one for the index in \Autoref{sec:operator_valued_weights}, and the lower bound $h \geq \norm{\Ind T}^{-1}\IdentityMatrix$ proved below generalizes the well-known fact that the index of a conditional expectation between factors, where $h = \IdentityMatrix$, is at least one \cite{kosaki_extension_1986}. For every unitary $u \in \N$ the bimodule property gives $u^* h u = T(u^* u) = h$, hence $h \in \N \cap \N' = Z(\N)$. To bound $h$ from below, let $ξ \in \HS$ be a unit vector and let $ρ \in \S(\M)$ be its vector state, with vector representative $\v{ρ} \coloneqq ξ$ and induced functionals $ρ|_{\mathcal{A}}$, $\mathcal{A} \in \set{\M, \N, \M', \N'}$, as in the standing assumptions of \Autoref{sec:duality-relations}; put $τ' \coloneqq T^{-1}_*(ρ') \in (\N')_*^+$. By positivity of $T$, the first-moment identity \eqref{eq:first-moment} applied to $T_*(ρ|_{\N}) \in \M_*^+$ (note that $T_*(ρ|_{\N})(s(ρ)) = \bra{\v{ρ}} T(s(ρ)) \ket{\v{ρ}}$), and Haagerup's identity in the form of \Autoref{lem:haagerup-identity-finite-index},
    \begin{equation}\label{eq:ovw-h-lower-bound-1}
        \bra{\v{ρ}} h \ket{\v{ρ}} \geq \bra{\v{ρ}} T(s(ρ)) \ket{\v{ρ}} = \bra{\v{ρ}} \Delta(T_*(ρ|_{\N})/ρ') \ket{\v{ρ}} = \bra{\v{ρ}} \Delta(ρ|_{\N}/τ') \ket{\v{ρ}}\,.
    \end{equation}
    To bound the right hand side from below, we apply \Autoref{lem:first-moment} with the roles of the algebra and its commutant interchanged, namely to the algebra $\N'$, its vector state $ρ|_{\N'}$ with representative $\v{ρ}$ (whose induced functional on the commutant $\N$ is $ρ|_{\N}$) and the second argument $τ'$: if $μ$ denotes the spectral measure of $\v{ρ}$ with respect to $\Delta(ρ|_{\N}/τ')$, a probability measure on $[0, \infty]$ with $μ(\set{0}) = 0$, then \eqref{eq:first-moment} states that $\int λ^{-1} \, \dd{μ}(λ) = τ'(s(ρ|_{\N'})) \leq τ'(\IdentityMatrix) = ρ'(\Ind T) = \bra{\v{ρ}} \Ind T \ket{\v{ρ}}$. If $μ(\set{\infty}) > 0$, the right hand side of \eqref{eq:ovw-h-lower-bound-1} is $+\infty$; otherwise Jensen's inequality for the convex function $λ \mapsto λ^{-1}$ on $(0, \infty)$ gives
    \begin{equation}\label{eq:ovw-h-lower-bound-2}
        \bra{\v{ρ}} \Delta(ρ|_{\N}/τ') \ket{\v{ρ}} = \int λ \, \dd{μ}(λ) \geq \Big( \int λ^{-1} \, \dd{μ}(λ) \Big)^{-1} \geq \big( \bra{\v{ρ}} \Ind T \ket{\v{ρ}} \big)^{-1} \geq \norm{\Ind T}^{-1}\,.
    \end{equation}
    Since the unit vector $ξ = \v{ρ}$ was arbitrary, $h \geq \norm{\Ind T}^{-1} \IdentityMatrix$. In particular $h$ is invertible and $\log h \in Z(\N)$ is bounded.

    \emph{Step 2: Reduction to a conditional expectation.} The map $F \coloneqq h^{-1} T: \M \to \N$ is positive, normal, faithful and unital, and it satisfies $F(axb) = aF(x)b$ for $a, b \in \N$ because $h^{-1} \in Z(\N)$; hence $F$ is a normal faithful conditional expectation of $\M$ onto $\N$, and $T = hF$. For $σ \in \N_*^+$ let $σ_h \coloneqq h^{1/2} σ h^{1/2}$ in the notation of \Autoref{prop:spatial-derivative-properties}\ref{item:sd-covariance}, i.e.\ $σ_h(x) = σ(hx)$ for $x \in \N$. Then $σ_h \in \N_*^+$, $σ_h(\IdentityMatrix) = σ(h) \geq \norm{\Ind T}^{-1} σ(\IdentityMatrix)$, and
    \begin{equation}\label{eq:ovw-reduction}
        T_*(σ) = σ \circ (hF) = σ_h \circ F = F_*(σ_h)\,.
    \end{equation}

    \emph{Step 3: Central scaling of the second argument.} The following identity extends the scaling rule of \Autoref{prop:divergence-properties}\ref{item:dp-scaling} from scalars to central elements. We claim that for every $ω \in \S(\N)$ with a vector representative $\v{ω} \in \HS$ and every $σ \in \N_*^+$,
    \begin{equation}\label{eq:central-scaling}
        D(ω\|σ_h) = D(ω\|σ) - ω(\log h)\,.
    \end{equation}
    Let $ω' \coloneqq ω_{\v{ω}}|_{\N'}$ and $\mathcal{K} \coloneqq s(ω')\HS \ni \v{ω}$. By \Autoref{prop:spatial-derivative-properties}\ref{item:sd-support}, $\mathcal{K}$ reduces $\Delta(σ/ω')$ and $\Delta(σ_h/ω')$; let $\Delta_{\mathcal{K}}$ and $\Delta_{\mathcal{K}}^h$ be the restrictions of these operators to $\mathcal{K}$. By \Autoref{lem:first-moment}, $D(ω\|σ) = -\int \log λ \, \dd{ν}(λ)$ with $ν$ the spectral measure of $\v{ω}$ with respect to $\Delta_{\mathcal{K}}$, and likewise $D(ω\|σ_h) = -\int \log λ \, \dd{ν_h}(λ)$ with $ν_h$ the spectral measure of $\v{ω}$ with respect to $\Delta_{\mathcal{K}}^h$. Since $h \in \N$ commutes with $s(ω') \in \N'$, it leaves $\mathcal{K}$ invariant, and \Autoref{prop:spatial-derivative-properties}\ref{item:sd-covariance} with $a = h^{1/2}$ gives $\Delta_{\mathcal{K}}^h = h_{\mathcal{K}}^{1/2} \Delta_{\mathcal{K}} h_{\mathcal{K}}^{1/2}$, where $h_{\mathcal{K}} \coloneqq h|_{\mathcal{K}}$. Moreover, $h_{\mathcal{K}}$ commutes with $\Delta_{\mathcal{K}}$: a unitary $u \in Z(\N)$ satisfies $uσu^* = σ$, so the same lemma gives $u_{\mathcal{K}} \Delta_{\mathcal{K}} u_{\mathcal{K}}^* = \Delta_{\mathcal{K}}$ with the unitary $u_{\mathcal{K}} \coloneqq u|_{\mathcal{K}}$, i.e.\ $u_{\mathcal{K}}$ commutes with all spectral projections of $\Delta_{\mathcal{K}}$, and $h$ is a linear combination of unitaries of $Z(\N)$. Consequently $\Delta_{\mathcal{K}}^h = h_{\mathcal{K}} \Delta_{\mathcal{K}}$ is the product of two commuting positive self-adjoint operators, the first of which is bounded and invertible, and the functional calculus for the commuting pair $(h_{\mathcal{K}}, \Delta_{\mathcal{K}})$ gives $\log \Delta_{\mathcal{K}}^h = \log h_{\mathcal{K}} + \log \Delta_{\mathcal{K}}$ (with $\log 0 = -\infty$ on $\ker \Delta_{\mathcal{K}}$). Taking the expectation value in $\v{ω}$ is legitimate because $\log h_{\mathcal{K}}$ is bounded while $-\bra{\v{ω}} \log \Delta_{\mathcal{K}} \ket{\v{ω}} = D(ω\|σ) \in (-\infty, +\infty]$, and it yields \eqref{eq:central-scaling}, as $\bra{\v{ω}} \log h \ket{\v{ω}} = ω(\log h)$.

    \emph{Step 4: Conclusion.} Let $σ \in \S(\N)$. Applying \eqref{eq:pythagorean-identity} for the conditional expectation $F$ to the state $σ_h/σ(h)$, and then the scaling rule of \Autoref{prop:divergence-properties}\ref{item:dp-scaling} on both sides (the two terms $\log σ(h)$ cancel), we obtain from \eqref{eq:ovw-reduction} and \eqref{eq:central-scaling} with $ω = ρ|_{\N}$ that
    \begin{align}
        D(ρ\|T_*(σ)) = D(ρ\|F_*(σ_h)) &= D(ρ\|F_*(ρ|_{\N})) + D(ρ|_{\N}\|σ_h) \nonumber\\
        &= D(ρ\|F_*(ρ|_{\N})) - ρ(\log h) + D(ρ|_{\N}\|σ)\,. \label{eq:ovw-final}
    \end{align}
    For $σ = ρ|_{\N}$ this reads $D(ρ\|T_*(ρ|_{\N})) = D(ρ\|F_*(ρ|_{\N})) - ρ(\log h)$, because $D(ω\|ω) = 0$ for every state $ω \in \S(\N)$ with vector representative $\v{ω}$: on the one hand $D(ω\|ω) \geq -\log ω(s(ω)) = 0$ by \Autoref{lem:relative-entropy-well-defined}, and on the other hand, if $μ$ denotes the spectral measure of $\v{ω}$ with respect to $\Delta(ω'/ω)$, $ω' \coloneqq ω_{\v{ω}}|_{\N'}$, then $μ(\set{\infty}) = ω(\IdentityMatrix - s(ω)) = 0$, and Jensen's inequality together with \eqref{eq:first-moment}, applied with the roles of the algebra and its commutant interchanged, gives $D(ω\|ω) = \int \log λ \, \dd{μ}(λ) \leq \log \int λ \, \dd{μ}(λ) = \log ω'(s(ω')) = 0$. Substituting this into \eqref{eq:ovw-final} gives \eqref{eq:pythagorean-identity}; all terms lie in $(-\infty, +\infty]$ and $ρ(\log h)$ is finite, so no undefined differences occur.
\end{proof}

\printbibliography[heading=bibintoc]

\end{document}